%% file: main-TEAC-final.tex
\documentclass[acmsmall, nonacm=true]{acmart}

\input{parts-frontmatter/frontmatter.tex}

\input{parts-frontmatter/packages.tex}

\begin{document}
%\input{parts-body/todo.tex}

% abstract & make title page
\input{parts-body/abstract.tex}

% introduction - cheating part
\input{parts-body/introductionv3.tex}

%% BODY TECHNICAL PART 

% general LP optimal local tolls
\input{parts-supplementary/supplementary-optimaltolls.tex}

% analytical solution and simplified LP
\input{parts-supplementary/supplementary-explicitsolution.tex}

% optimal tolls for large n
\input{parts-supplementary/supplementary-largen.tex}

% constant tolls
\input{parts-supplementary/supplementary-constanttolls.tex}

% marginal cost tolls
\input{parts-supplementary/supplementary-marginalcost.tex}

% conclusions
\input{parts-supplementary/conclusions.tex}

%% APPENDIX
% supplementary material
\appendix
\section{Appendix}

% optimal toll
\subsection{Results used in \cref{sec:optimal-tolls}}
\label{subsec:appendix-optimal-toll}
\input{parts-appendix/appendix-optimal-toll.tex}

% explicit expression
\subsection{Results used in \cref{sec:explicit-expression}}
\input{parts-appendix/appendix-explicit-expression.tex}

% explicit expression
\subsection{Results used in \cref{sec:largen}}
\input{parts-appendix/appendix-largen.tex}

% congestion-independent tolls
\subsection{Results used in \cref{sec:congestion-indep}}
\input{parts-appendix/appendix-constant-toll.tex}

% acknowledgements and biblio
\include{parts-body/ack-biblio}

\end{document}

%% file: parts-frontmatter/frontmatter.tex
\acmJournal{TEAC}
\title{Optimal Taxes in Atomic Congestion Games}
\author{Dario Paccagnan}
   \affiliation{%
   \institution{Imperial College London}   
   \department{Department of Computing}
   \city{London}
   \country{UK}
   }
\email{d.paccagnan@imperial.ac.uk}
\author{Rahul Chandan}
\affiliation{%
   \institution{University of California at Santa Barbara}   
   \department{Center for Control, Dynamical Systems and Computation}
   \city{Santa Barbara}
   \state{CA}
   \postcode{93105}
   \country{USA}
   }
\email{rchandan@ucsb.edu}
\author{Bryce L. Ferguson}
\affiliation{%
   \institution{University of California at Santa Barbara}   
   \department{Center for Control, Dynamical Systems and Computation}
   \city{Santa Barbara}
   \state{CA}
   \postcode{93105}
   \country{USA}
}
\email{blf@ucsb.edu}
\author{Jason R. Marden}
\affiliation{%
   \institution{University of California at Santa Barbara}   
   \department{Center for Control, Dynamical Systems and Computation}
   \city{Santa Barbara}
   \state{CA}
   \postcode{93105}
   \country{USA}
}
\email{jrmarden@ucsb.edu}
\renewcommand{\shortauthors}{Paccagnan, et al.}

\ccsdesc[500]{Theory of computation}
\ccsdesc[500]{Theory of computation~Algorithmic game theory}
\ccsdesc[500]{Theory of computation~Quality of equilibria}
\keywords{
Congestion games, 
optimal taxation mechanisms,
price of anarchy, 
Nash equilibrium,
coarse correlated equilibrium,
selfish routing, 
approximation algorithms.
}

\thanks{This manuscript has been deposited to the arXiv on 22 Nov 2019. A preliminary version appeared in ``Proceedings of the 14th Workshop on the Economics of Networks, Systems and Computation'' (NetEcon '19), see \cite{DBLP:conf/netecon/ChandanPFM19}. This research was carried out when the first author was affiliated with the University of California at Santa Barbara. This work was supported by the Swiss National Science Foundation Grant P2EZP2\_181618, ONR Grant \#N00014-20-1-2359, AFOSR Grant \#FA9550-20-1-0054.
}

%% file: parts-frontmatter/packages.tex
\usepackage[utf8]{inputenc}
\usepackage{amsmath}
\usepackage{amsthm}
\usepackage{amsfonts} 
\usepackage{bbm}
\usepackage{xcolor}
\usepackage{pgfplots} 
\usepackage{graphicx} 
\usepackage{xfrac}
\usepackage{tikz}
\usepackage{hyperref}
\usepackage[capitalise]{cleveref}
\usepackage[font=small, tablename=Table]{caption}
\usepackage{subcaption}
\usepackage{tcolorbox}
\usepackage{tikz-network}
\usepackage{textpos}
\usepackage{xfp}
\usepackage{pgffor}
\usepackage{siunitx}
\usepackage{multirow}
\usepackage{wrapfig}
\usepackage{enumitem}
\usepackage[short,nocomma]{optidef}

\newtheorem{example}{Example}
\newtheorem{lemma}{Lemma}
\newtheorem{theorem}{Theorem}
\newtheorem{corollary}{Corollary}

\newcommand{\poa}{\mathrm{PoA}}
\newcommand{\bu}{\bar{u}}

\newcommand{\cg}{\mathcal{G}}
\newcommand{\be}{\begin{equation}}
\newcommand{\ee}{\end{equation}}	
\newcommand{\actionset}{\mathcal{A}}

\newcommand{\taxmechanism}{T}
\newcommand{\mc}[1]{\mathcal{#1}}
\newcommand{\mb}[1]{\mathbb{#1}}
\newcommand{\resset}{\mc{E}}

\newcommand{\PNEcost}{{\tt{NECost}}}

\newcommand{\mincost}{{\tt{MinCost}}}
\newcommand{\SC}{{\rm{SC}}}
\renewcommand{\C}{{\rm{C}}}
\newcommand{\opt}[1]{{#1}^{\rm opt}}
\newcommand{\NE}[1]{{#1}^{\rm ne}}
\newcommand{\pig}[1]{{#1}^{\rm p}}
\pgfplotsset{compat=1.15}
\newcommand{\myf}{f^{\infty}}
\newcommand{\Thathat}{\bar{T}}

\renewcommand{\u}{u}
\renewcommand{\v}{v}
\newlength{\mynegspace}
\newcommand{\mybeta}{c_2}
\newcommand{\mygamma}{c_1}
\newcommand{\mynu}{\lambda}
\newcommand{\nquad}{\mkern-15mu}

\newcommand{\jRC}{u}
\newcommand{\ellRC}{v}

\renewcommand{\star}{\ast}
\newcommand{\q}{q}

\usetikzlibrary{shapes.misc}
\tikzset{cross/.style={cross out, draw, 
         minimum size=2*(#1-\pgflinewidth), 
         inner sep=0pt, outer sep=0pt}}
         
\usepackage{xr}
\definecolor{pnasblueback}{RGB}{205,217,235}

%% file: parts-body/abstract.tex
\begin{abstract}
\vspace*{1mm}
\begin{center}
First version:  {\it November 2019}.\quad This version:  {\it March 2021.}
\end{center}
\vspace*{3.5mm}
{\bf Abstract.} How can we design mechanisms to promote efficient use of shared resources? Here, we answer this question in relation to the well-studied class of atomic congestion games, used to model a variety of problems, including traffic routing.
Within this context, a methodology for designing tolling mechanisms that minimize the system inefficiency (price of anarchy) exploiting solely local information is so far missing in spite of the scientific interest. 
In this manuscript we resolve this problem through a tractable linear programming formulation that applies to and beyond polynomial congestion games. When specializing our approach to the polynomial case, we obtain tight values for the optimal price of anarchy and corresponding tolls, uncovering an unexpected link with load balancing games.
We also derive optimal tolling mechanisms that are constant with the congestion level, generalizing the results of \cite{caragiannis2010taxes} to polynomial congestion games and beyond.
Finally, we apply our techniques to compute the efficiency of the marginal cost mechanism. Surprisingly, optimal tolling mechanism using only local information perform closely to existing mechanism that utilize global information \cite{BiloV19}, while the marginal cost mechanism, known to be optimal in the continuous-flow model, has lower efficiency than that encountered levying \emph{no} toll.
All results are tight for pure \mbox{Nash equilibria, and extend to coarse correlated equilibria.}
\end{abstract}

% make title page
\maketitle

%% file: parts-body/introductionv3.tex
\section{Introduction}
Modern society is based on large-scale systems often at the service of end-users, e.g., transportation and communication networks. As the performance of such systems heavily depends on the interaction between the users' individual behaviour and the underlying infrastructure (e.g., drivers on a road-traffic network), the operation of such systems requires interdisciplinary considerations at the confluence between economics, engineering, and computer science.

A common issue arising in these settings is the performance degradation incurred when the users' individual objectives are not aligned to the ``greater good''. A prime example of how users' behaviour degrades the performance is provided by road-traffic routing: when drivers choose routes that minimize their individual travel time, the aggregate congestion could be much higher compared to that of a centrally-imposed routing. A fruitful paradigm to tackle this issue is to employ appropriately designed tolling mechanisms, as widely acknowledged in the economic and computer science literature \cite{morrison1986survey,bergendorff1997congestion,laffont2009theory}. 

Pursuing a similar line of research, this paper focuses on the design of tolling mechanisms that maximize the system efficiency associated with self-interested decision making in atomic \emph{congestion games}, i.e., that minimize the resulting \emph{price of anarchy} \cite{koutsoupias1999worst}. Within this context, we develop a methodology to compute the most efficient \emph{local} tolling mechanism, i.e., the most efficient mechanism whose tolls levied on each resource depend only on the local properties of that resource. We do so for both the congestion-aware and congestion-independent settings, whereby tolls have or do not have access to the current congestion levels. A summary of our results, including a comparison with the literature, can be found in \cref{tab:naware_vs_nagnostic}. Perhaps surprisingly, optimal local tolls deliver a price of anarchy close to that of existing tolls designed using global information \cite{BiloV19}.

%%%%%%%%%%%%%%%%%%%%%%%%%%%%%%%%%%%%%%%%%%%%%%%%%%%%%%%%%%%%%%%
\subsection{Congestion games, local and global mechanisms}
Congestion games were introduced in 1973 by Rosenthal  \cite{rosenthal1973class}, and since then have found applications in diverse fields, such as energy markets \cite{paccagnan2018nash}, machine scheduling \cite{suri2007selfish}, wireless data networks \cite{tekin2012atomic}, sensor allocation \cite{marden2013distributed}, network design \cite{anshelevich2008price}, and many more. While our results can be applied to a variety of problems, we consider traffic routing as our prime application.

In a congestion game, we are given a set of users $N=\{1,\dots,n\}$, and a set of resources $\mc{E}$. Each user can choose a subset of the set of resources which she intends to use. We list all feasible choices for user $i\in N$ in the set $\mc{A}_i\subseteq 2^{\mc{E}}$. The cost for using each resource $e\in\mc{E}$ depends only on the total number of users concurrently selecting that resource, and is denoted with $\ell_e :\mb{N}\rightarrow \mb{R}_{\ge0}$. Once all users have made a choice $a_i\in\mc{A}_i$, each user incurs a cost obtained by summing the costs of all resources she selected.  Finally, the system cost represents the cost incurred by \emph{all} users
\begin{equation}
\SC(a)=\sum_{i\in N}\sum_{e\in a_i} \ell_e(|a|_e),
\label{eq:systemcost}
\end{equation}
where $|a|_e$ is the number of users selecting resource $e$ in allocation $a=(a_1,\dots,a_n)$. We denote with $\mc{G}$ the set containing all congestion games with a maximum of $n$ agents, and where all resource costs $\{\ell_e\}_{e\in\mc{E}}$ belong to a common set of functions $\mc{L}$.

\paragraph{Local and global tolling mechanisms}
We assume users to be self-interested, and observe that self-interested decision making often results in a highly suboptimal system cost \cite{pigou}. Consequently,  there has been considerable interest in the application of tolling mechanisms to influence the resulting outcome \cite{morrison1986survey,bergendorff1997congestion,hearn1998solving,cole2006much,caragiannis2010taxes,BiloV19}. Formally, a tolling mechanism $T:G\times e\mapsto \tau_e$ is a map that associates an instance $G\in\mc{G}$ and a resource $e\in\mc{E}$ to the corresponding toll $\tau_e$. Here $\tau_e:\{1,\dots,n\}\rightarrow \mb{R}$ is a congestion-dependent toll, i.e., $\tau_e$ maps the number of users in resource $e$ to the corresponding toll. As a result, user $i\in N$ incurs a cost factoring both the cost of the resources and the tolls, i.e.,
\be
\C_i(a)=\sum_{e\in a_i} [\ell_e(|a|_e)+\tau_e(|a|_e)].
\label{eq:Ci}
\ee
Designing tolling mechanisms that utilize global information (such as knowledge of all resource costs, or knowledge of the feasible sets $\{\mc{A}_i\}_{i=1}^n$) is often difficult, as that might require the central planner to access private users information, in addition to the associated computational burden. Within the context of traffic routing (see \cref{ex:trafficroute} below), a global mechanism might produce a toll on edge $e$ that depends on the structure of the network, on the travel time over all edges, as well as on the exact origin and destination of every user. On the contrary, local tolling mechanisms require much less information, are scalable, accommodate resources that are dynamically added or removed, and are robust against a number of variations, e.g., the common scenario where drivers modify their destination. Therefore, a significant portion of the literature has focused on designing tolls that exploit only local information. Formally, we say that a tolling mechanism $T$ is \emph{local} if the toll on each resource only uses information on the cost function $\ell_e$ of the same resource $e$, and no other information.  If this is the case, we write $\tau_e = T(\ell_e)$ with slight abuse of notation. On the contrary, if the mechanism utilizes additional information on the instance, we say it is \emph{global}.

\begin{example}[traffic routing]
\label{ex:trafficroute}
Within the context of traffic routing, $\mc{E}$ represents the set of edges defining the underlying road network over which users wish to travel.
For this purpose, each user $i\in N$ can select any path connecting her origin to her destination node, thus producing a list of feasible paths $\mc{A}_i$. 
The travel time incurred by a user transiting on edge $e\in\mc{E}$ is captured by the function $\ell_e$, and depends only on the number of users traveling through that very edge. In this context, the function $\ell_e$ takes into consideration geometric properties of the edge, such as its length, the number of lanes and speed limit \cite{wardrop1952road}. The system cost in \eqref{eq:systemcost} represents the time spent on the network by \emph{all} users, whereas tolls are imposed on the edges to incentivize users in selecting paths \mbox{that minimize the total travel time \eqref{eq:systemcost}.}\end{example}

%%%%%%%%%%%%%%%%%%%%%%%%%%%%%%%%%%%%%%%%%%%%%%%%%%%%%%%%%%%%%%%
\subsection{Performance metrics}
The performance of a tolling mechanism is typically measured by the ratio between the system cost incurred at the worst-performing emergent allocation and the minimum system cost. As users are assumed to be self-interested, the emergent allocation is described by any of the following equilibrium notions: pure or mixed Nash equilibrium, correlated or coarse correlated equilibrium -- each being a superset of the previous \cite{roughgarden2015intrinsic}.\footnote{For a congestion game, existence of pure Nash equilibria (and thus of all other equilibrium notions mentioned above) is guaranteed even in the presence of tolls, due to the fact that the resulting game is potential.} When considering pure Nash equilibria, the performance of a mechanism $T$, referred to as the \emph{price of anarchy} \cite{koutsoupias1999worst}, is defined as
\begin{equation}
\poa(T) = \sup_{G\in\mc{G}} \frac{\PNEcost(G,T)}{\mincost(G)},
\label{eq:poadef}
\end{equation}
where $\mincost(G)$ is the minimum social cost for instance $G$ as defined in \eqref{eq:systemcost}, and $\PNEcost(G,T)$ denotes the highest social cost at a Nash equilibrium obtained when employing the mechanism $T$ on the game $G$. Similarly, it is possible to define the price of anarchy for mixed Nash, correlated and coarse correlated equilibria. While these different metrics need not be equal in general, they do \emph{coincide} within the setting of interest to this manuscript, as we will later clarify. Therefore, in the following, we will use $\poa(T)$ to refer to the efficiency values of any and all \mbox{these equilibrium classes.}

%%%%%%%%%%%%%%%%%%%%%%%%%%%%%%%%%%%%%%%%%%%%%%%%%%%%%%%%%%%%%%%
\subsection{Related work}
The interest in the design of tolls dates back to the early 1900s \cite{pigou}. Since then, a large body of literature in the areas of transportation, economics, and computer science has investigated this approach \cite{morrison1986survey,cole2006much,bergendorff1997congestion,hearn1998solving,sandholm2002evolutionary}. Designing tolling mechanisms that optimize the efficiency is particularly challenging in the context of (atomic) congestion games, as observed, e.g., by \cite{harks2019pricing}, in part due to the multiplicity of the equilibria. While most of the research \cite{awerbuch2005price,christodoulou2005price,aland2006exact, roughgarden2015intrinsic} has focused on providing efficiency bounds for given schemes or in the un-tolled case, much less is known regarding the \emph{design} question. Owing to the technical difficulties, only partial results are available for global tolling mechanisms \cite{caragiannis2010taxes, fotakis2008cost, BiloV19}, while results for local mechanisms are limited to affine \mbox{congestion games \cite{caragiannis2010taxes}.}

Reference \cite{caragiannis2010taxes} initiated the study of tolling mechanisms in the context of congestion games, restricting their attention to affine resource costs. Relative to this setting, they show how to compute a congestion-independent global toll yielding a (tight) price of anarchy of $2$ for mixed Nash equilibria, in addition to a congestion-independent local toll yielding a (tight) price of anarchy of $1+2/\sqrt{3}\approx 2.155$ for pure Nash equilibria. Our result pertaining to the design of optimal \emph{congestion-independent local} tolls generalizes this finding to any polynomial (and non polynomial) congestion game and holds tightly for both pure Nash and coarse correlated equilibria.

More recently \cite{BiloV19} studies congestion-dependent global tolls for pure Nash and coarse correlated equilibria, in addition to one-round walks from the empty state. Relative to unweighted congestion games, they derive tolling mechanisms yielding a price of anarchy for coarse correlated equilibria equal to $2$ for affine resource costs; and similarly for higher order polynomials. While the latter work provides a number of interesting insights (e.g., some closed form price of anarchy expressions), all the derived tolling schemes require global information such as network and user knowledge - an often impractical scenario. In contrast, our results on \emph{optimal local tolls} focuses on the design of optimal mechanisms that exploit local information only, and thus are more widely applicable. Even if utilizing much less information, optimal local mechanisms are still competitive. For example, congestion-dependent optimal local tolls yield a price of anarchy of $2.012$ for coarse correlated equilibria and affine congestion games (to be compared with a value of $2$ mentioned above).

Related works have also explored modifications of the setup considered here:  \cite{caragiannis2010taxes} also studies singleton congestion games, \cite{fotakis2008cost} focuses on symmetric network congestion games, \cite{milchtaich2020internalization} on altruistic congestion games, while \cite{HarksKKM15, HoeferOS08, JelinekKS14} consider tolling a subset of the resources. Preprint \cite{skopalik2020improving} focuses on the computation of approximate Nash equilibria in atomic congestion games. Therein, the authors design modified resource costs leveraging a methodology \mbox{similar to that developed in~\cite{chandan2019when, DBLP:conf/netecon/ChandanPFM19,paccagnan2019incentivizing}.}

Finally, we note that the design of tolling mechanisms is far better understood when the original congestion game is replaced by its continuous-flow approximation, as uniqueness of the Nash equilibrium is guaranteed. Limited to this setting, marginal cost tolls produce an equilibrium which is always optimal~\cite{Beckmann1956}. Within the atomic setting, our work demonstrates that marginal cost tolls do not improve - and instead significantly deteriorate - the resulting system efficiency.

%%%%%%%%%%%%%%%%%%%%%%%%%%%%%%%%%%%%%%%%%%%%%%%%%%%%%%%%%%%%%%%
\subsection{Preview of our contributions}
The core of our work is centred on designing optimal local tolling mechanisms and on deriving their corresponding prices of anarchy for various classes of congestion games. Our work develops on parallel directions, and the contributions include the following:
\begin{itemize}
	\item[i)] The design of optimal local tolling mechanisms;
	\item[ii)] The design of optimal congestion-independent local tolling mechanisms;
	\item[iii)] The study of marginal cost tolls and their inefficiency.
\end{itemize}
\cref{tab:naware_vs_nagnostic} highlights the impact of our results on congestion games with polynomial cost functions of varying degree, though our methodology extends further. The following paragraphs describe our contributions in further details, while supporting Python and Matlab code can be found in~\cite{Chandancode19}. 

\input{parts-body/table-optimal-poa.tex}

\paragraph{Optimal local tolls}
In \cref{thm:main-thm,thm:simplifiedLP-explicit-solution} we provide a methodology for computing optimal local tolling mechanisms for congestion games. The resulting price of anarchy values for the case of polynomial congestion games are presented in \cref{tab:naware_vs_nagnostic} (fourth column), where we provide a comparison with those derived in \cite{caragiannis2010taxes,BiloV19} (third column), which instead make use of global information, for example, by letting the tolling function on resource $e$ depend on \emph{all} the other resource costs. Perhaps unexpectedly, the efficiency of optimal tolls designed using only local information is almost identical to that of existing tolls designed using global information \cite{caragiannis2010taxes,BiloV19}. In addition to providing similar performances by means of less information, local tolls are robust against uncertain scenarios (e.g., modifications in the origin/destination pairs) and can be computed efficiently.

Extensive work has focused on quantifying the price of anarchy for load balancing games, that is congestion games where all action sets are singletons, i.e., $\mc{A}_i\subseteq\mc{E}$. Within this setting, and when all resource costs are affine and identical, the price of anarchy is $\approx2.012$ \cite{suri2007selfish, CaragiannisFKKM11}. Our results connect with this line of work, demonstrating that optimally tolled affine congestion games have a price of anarchy  matching this value, and are tight already within this class. Stated differently, the price of anarchy of un-tolled and optimally tolled affine load balancing games with identical resources is the same. We believe such statement holds true more generally.

\paragraph{Optimal congestion-independent local tolls}
Our methodology can also be exploited to derive optimal local mechanisms under more stringent structural constraints. One such constraint, studied in numerous settings, consists in the use of congestion-independent mechanisms, which are attractive because of their simplicity. A linear program to compute optimal congestion-independent local mechanisms is presented in \cref{thm:congestion-indep}, while the corresponding optimal prices of anarchy for the case of polynomial congestion games are displayed in \cref{tab:naware_vs_nagnostic} (fifth column) and derived in \cref{cor:constanttollsd=2} as well as \cref{sec:congestion-indep}. All the results are novel, except for the case of $d=1$, which recovers \cite{caragiannis2010taxes}. We observe that the performance of congestion-independent mechanisms is comparable with that of congestion-aware mechanisms for polynomials of low degree ($d\le 3$), and still a good improvement over the un-tolled setup for high degree polynomials. In these cases, congestion-independent mechanisms are not only robust and simple to implement, but also relatively efficient.

\paragraph{Marginal cost tolls are worse than no tolls}
In non-atomic congestion games, any Nash equilibrium resulting from the application of the marginal contribution mechanism is optimal, i.e., it has a price of anarchy equal to one. \cref{cor:marginalcost} shows how to utilize our approach to compute the efficiency of the marginal cost mechanism in the \emph{atomic} setup. The resulting values of the price of anarchy are presented in the last column of Table~\ref{tab:naware_vs_nagnostic} for polynomial congestion games. While the marginal cost mechanism ensures that a Nash equilibrium is optimal (i.e., its price of stability is one), utilizing the marginal cost mechanism on the atomic model yields a price of anarchy that is \emph{worse} than that experienced levying no toll at all (compare the second and last column in Table~\ref{tab:naware_vs_nagnostic}). In other words, the design principle derived from the continuous-flow model does not carry over to the original atomic setup. This phenomenon manifests itself already in very simple settings, as we demonstrate in \cref{example:tolledpigouvian}.
We conclude by observing that our result differs significantly from that in \cite{meir2016marginal}, where the authors show that, as the number of agents grow large, marginal cost tolls become optimal. This difference stems from the fact that, in \cite{meir2016marginal}, the network structure is \emph{fixed} as the number of agents grows, whereas the worst case instance in our setting is a function of the number of agents.

%%%%%%%%%%%%%%%%%%%%%%%%%%%%%%%%%%%%%%%%%%%%%%%%%%%%%%%%%%%%%%% 
\subsection{Techniques and high-level ideas}
Underlying the developments presented above are a number of technical results, which stem from the observation that, in the majority of the existing literature, the set $\mc{L}$ contains all resource costs of the form $\ell(x)=\sum_{j=1}^m\alpha_jb_j(x)$ with $\alpha_j\ge0$, and given basis functions $\{b_1, \dots,b_m\}$. This describes the fact that each resource cost featured in the corresponding game belongs to a known class of functions. For example, in polynomial congestion games of maximum degree $d$, each resource is associated to a cost of the form $\alpha_1+\alpha_2 x+\dots +\alpha_{d+1} x^d$, corresponding to the choice of basis functions $\{1,\dots,x^d\}$. More generally, the decomposition $\ell(x)=\sum_{j=1}^m\alpha_jb_j(x)$ allows us to leverage a common framework to study different classes of problems not limited to polynomial congestion~games.
  
  In this context, we first show in \cref{thm:main-thm} that an optimal local tolling mechanisms is a linear map from the set of resource costs to the set of tolls. More precisely, we show that, for every resource cost $\ell(x)=\sum_{j=1}^m \alpha_j b_j(x)$, there exists an optimal local mechanism satisfying $\opt{T}(\ell)=\opt{T}(\sum_{j=1}^m \alpha_jb_j)=\sum_{j=1}^m \alpha_j \opt{T}(b_j)$, where the mechanism is obtained as a linear combination of $\opt{T}(b_j)$, with the \emph{same} coefficients $\alpha_j$ used to define $\ell$.\footnote{As word of caution, we  remark that linearity of the optimal mechanism in the sense clarified above does \emph{not} mean that the corresponding tolls are linear in the congestion level, i.e., does \emph{not} mean that $\tau_e(x)=a_ex+b_e$ for some $a_e,b_e$.} It is worth noting that this first result allows for a decoupling argument, whereby an optimal tolling function $\opt{T}(b_j)$ can be  separately computed for each of the basis $b_j$. 
The key idea underpinning the result on linearity of optimal tolls lies in observing that any congestion game utilizing resource cost functions with coefficients $\alpha_j\in\mb{R}_{>0}$ and a possibly non-linear tolling mechanism $T$, can be mapped to a corresponding congestion game where i) all coefficients $\alpha_j$ are identical to one, ii) only the linear part of the tolling mechanism $T$ is used, and iii) the price of anarchy is identical to that of the original game (as the number of resources grows). 
Complementary to this, our second result in \cref{thm:main-thm} reduces the problem of designing optimal basis tolls $\{\opt{T}(b_j)\}_{j=1}^m$ to a polynomially solvable linear program that also returns the tight value of the optimal price of anarchy. We do so by building upon the results in \cite{chandan19optimaljournal, paccagnan2018distributed}, which allow us to determine the performance of a tolling mechanism through the solution of a linear program (see \cref{linprog:characterize_poa}), in a similar spirit to \cite{BiloV19}, although with a provably tight characterization for \emph{any} number of agents and tolling function. We exploit this result and construct a polynomially-sized linear program that, for a given basis $b_j$, \mbox{searches over all linear tolls to find~$\opt{T}(b_j)$.}

When the basis functions are convex and increasing (e.g., in the well-studied polynomial case), we are able to explicitly solve the linear program and provide an analytic expression for the optimal tolling function, as well as a semi-analytic expression for optimal price of anarchy (\cref{thm:simplifiedLP-explicit-solution}). The fundamental idea consists in showing that the set of active constraints at the solution gives rise to a telescopic recursion, whereby the optimal toll to be levied when $u+1$ agents are selecting a resource can be written as a function of the optimal toll to be levied when only $u$ agents are present. This is the most technical part of the manuscript, and the expression of the optimal price of anarchy reveals an unexpected connection with that for un-tolled load balancing games on identical machines \cite{suri2007selfish, CaragiannisFKKM11}.
While the sizes of the linear programs appearing in \cref{thm:main-thm,thm:simplifiedLP-explicit-solution} grow (polynomially) with the number of agents $n$, in \cref{sec:largen} we show how to design optimal tolling mechanisms that apply to any $n$ (possibly infinite). Our approach consists in two steps: we first solve a linear program of finite dimension, and then extend its solution to arbitrary $n$.

Congestion-independent optimal local mechanisms as well as the efficiency of the marginal cost mechanism can also be  computed through linear programs (\cref{{thm:congestion-indep}} and \cref{cor:marginalcost}). In the first case, admissible tolls basis $\{\opt{T}(b_j)\}_{j=1}^m$ are constrained to be constant with the congestion, while in the second case the expression for marginal cost tolls is substituted in the program. We provide analytical solutions to these two programs for the case of polynomial congestion games.

%%%%%%%%%%%%%%%%%%%%%%%%%%%%%%%%%%%%%%%%%%%%%%%%%
\subsection{Organization}
In \cref{sec:optimal-tolls} we derive linear programs to compute optimal local tolling mechanisms. We also provide optimal price of anarchy values for polynomial congestion games. In \cref{sec:explicit-expression} we obtain an explicit solution to these programs that applies when resource costs are convex and increasing.
\cref{sec:largen} generalizes the previous results to arbitrarily large number of agents. In \cref{sec:congestion-indep} and \cref{sec:pigouvtolls} we derive congestion independent tolling mechanism and evaluate the efficiency of the marginal cost mechanism. In these sections we also specialize \mbox{the results to the polynomial case.}

%% file: parts-body/table-optimal-poa.tex
\begin{table}[t!]%
\centering%
    \begin{tabular}{S[table-format=1]||S[table-format=5.3]|S[table-format=3]|S[table-format=3.3]|S[table-format=4.3]|S[table-format=2.2]}%
    {$d$}&{No toll}&{Global toll}&{Optimal local toll}&{Optimal constant local toll}&{Marginal cost toll}\\
        &{\cite{aland2006exact}}&{from \cite{caragiannis2010taxes,BiloV19}}&  {(this work)}&{(this work)}&{(this work)}\\
        \hline%
        1 & 2.50 & 2 & 2.012 &  2.15
        &3.00\\
        2 & 9.58 &  5 & 5.101 & 5.33 &13.00\\
        3 & 41.54 & 15 & 15.551 &  18.36 &57.36\\
        4 & 267.64 & 52 & 55.452 &  89.41 &391.00\\
        5 & 1513.57 & 203 & 220.401 &  469.74 &2124.21\\
        6 & 12345.20 & 877 & 967.533 &  3325.58 &21337.00\\
    \end{tabular}%
\vspace*{0.5mm}
    \caption{Price of anarchy values for congestion games with resource costs of degree at most $d$. All results are tight for pure Nash and also hold for coarse correlated equilibria. 
    The columns feature the price of anarchy with no tolls, with global tolls from \cite{caragiannis2010taxes,BiloV19}, with optimal local tolls, with optimal constant (i.e. congestion-independent) local tolls, and with marginal cost tolls, respectively. Columns four, five, and six, are composed of entirely novel results, except for the case of constant tolls with $d=1$, which recovers \cite{caragiannis2010taxes}. Note that i) optimal tolls relying only on local information perform closely to optimal tolls designed using global information, with a difference in performance below $1\%$ for $d=1$; ii) congestion-independent tolls result in a price of anarchy that is comparable to that obtained using congestion-aware local tolls for polynomials of low degree. The code used to generate this table can be downloaded from~\cite{Chandancode19}.}%
\label{tab:naware_vs_nagnostic}%
\vspace*{-7.5mm}
\end{table}

%% file: parts-supplementary/supplementary-optimaltolls.tex
\section{Optimal Tolling Mechanisms}
\label{sec:optimal-tolls}
In this section we develop a methodology to compute optimal local tolling mechanisms through the solution of tractable linear programs. To ease the notation, we introduce the set of integer triplets $\mc{I}=\{(x,y,z)\in\mb{Z}_{\ge0}^3~\text{s.t.}~1\le x+y+z\le n~\text{and either}~xyz=0~\text{or}~x+y+z=n\}$, for given~$n\in\mb{N}$.
\begin{theorem}
\label{thm:main-thm}
A local mechanism minimizing the price of anarchy over congestion games with $n$ agents, resource costs $\ell(x)=\sum_{j=1}^m\alpha_j b_j(x)$, $\alpha_j\ge0$, and basis functions $\{b_1,\dots,b_m\}$ is given by
	\be
	\label{eq:optimaltolls_expression}
    \opt{T}(\ell) = \sum_{j = 1}^{m} \alpha_j\cdot \opt{\tau}_j, 
    \quad\text{where}~~\opt{\tau}_j:\{1,\dots,n\}\rightarrow \mb{R}, 
    \quad\opt{\tau}_j(x) = \opt{f}_{j}(x) - b_j(x)~~
    \ee
and $\opt{\rho}_j\in\mb{R}$, $\opt{f}_{j}:\{1,\dots,n\}\rightarrow\mb{R}$ solve the following linear programs (one per each $b_j$)
	\begin{maxi}
	{{\scriptstyle f \in \mathbb{R}^n,\,\rho \in \mathbb{R}}}
	{\nquad\rho} {\label{eq:mainLPopt}}{}
	\addConstraint{\!\!\!\!\!\!\!\!\!\!\!b_j(x+z)(x+z)-\rho b_j(x+y)(x+y)+f(x+y)y-f(x+y+1)z}{\ge 0\quad}{\forall\,(x,y,z) \in \mc{I},}
	\end{maxi}
where we define $b_j(0)=f(0)=f(n+1)=0$. Correspondingly, $\poa(\opt{T})=\max_{j}\{1/\opt{\rho}_j\}$.\footnote{If we require tolls to be non-negative, an optimal mechanism is as in \eqref{eq:optimaltolls_expression}, where we set $\opt{\tau}_j(x) = \opt{f}_{j}(x)\cdot \opt{\poa} - b_j(x)$.} These results are tight for pure Nash equilibria, and extend to coarse correlated equilibria.
\end{theorem}

\begin{figure}[h!]
\centering
\includegraphics[width=0.85\linewidth]{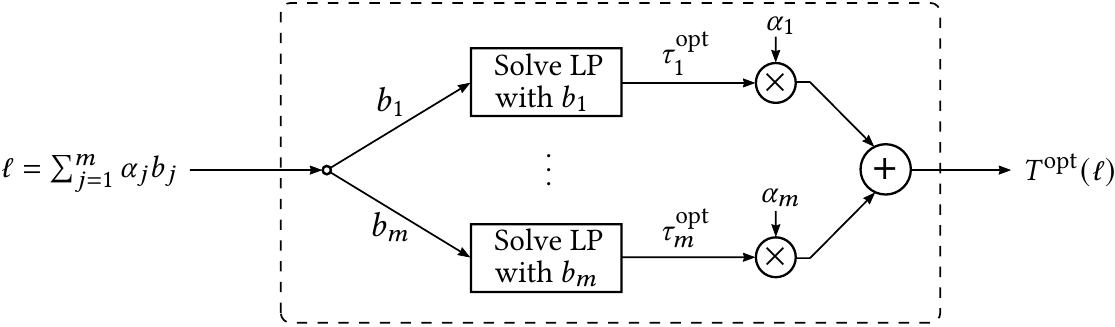}
\caption{Graphical representation of the main result on the design of optimal tolls. The input consists of a resource cost $\ell(x)$ expressed as a combination of basis $b_j(x)$ with coefficients $\alpha_j$. For each basis, we compute the associated optimal toll $\opt{\tau}_j(x) = \opt{f}_j(x)-b_j(x)$ by solving the linear program (LP) appearing in \eqref{eq:mainLPopt}. The resulting optimal toll is obtained as the linear combination of $\opt{\tau}_j(x)$ with the same coefficients $\alpha_j$. The quantities $\opt{\tau}_j(x)$ can be precomputed and stored in a library, offloading the solution of the linear programs.}
\label{fig:structureoptimaltolls}
\end{figure}

The above statement contains two fundamental results. The first part of the statement shows that an optimal tolling mechanism applied to the function $\ell(x)=\sum_{j=1}^m \alpha_j b_j(x)$ can be obtained as the linear combination of $\opt{\tau}_j(x)$, with the \emph{same} coefficients $\alpha_j$ used to define $\ell$. Complementary to this, the second part of the statement provides a practical technique to compute $\opt{\tau}_j(x)$ for each of the basis $b_j(x)$ as the solution of a tractable linear program. A graphical representation of this process is included in \cref{fig:structureoptimaltolls}, while Python/Matlab code to design optimal tolls can be found in~\cite{Chandancode19}.

We solved the latter linear programs for $n=100$ and polynomials of maximum degree $1\le d\le6$. The corresponding results are displayed in \cref{tab:naware_vs_nagnostic}, while \cref{sec:largen} shows that these results hold identically for arbitrarily large $n$.
In the case of $d=1$, the optimal price of anarchy is approximately $2.012$, matching that of un-tolled load balancing games on identical machines \cite{suri2007selfish, CaragiannisFKKM11}. We observe that, in this restricted setting, the price of anarchy cannot be improved \emph{at all} through local tolling mechanisms. In fact, no matter what non-negative tolling mechanism we are given, we can always construct a load balancing game on identical machines with a price of anarchy no lower than~$2.012$.\footnote{To do so, it is sufficient to utilize the instance in \cite[Thm 3.4]{suri2007selfish}, where the the resource cost $x$ used therein is replaced with $x+\tau(x)$. The Nash equilibrium and the optimal allocation will remain unchanged, yielding the same price of anarchy value.}

We conclude observing that the decomposition of resource costs as linear combination of basis functions is, strictly speaking, not required for \cref{thm:main-thm} to hold. Nevertheless, pursuing this approach would require to solve a linear program for each function in $\mc{L}$, a task that becomes daunting when $\mc{L}$ contains infinitely many functions, e.g., in the case of polynomial congestion games. In this case, \cref{thm:main-thm} allows to compute optimal tolls by solving \mbox{\emph{finitely} many linear programs.}

\begin{proof}
We divide the proof in two parts to ease the exposition.
\paragraph{Part 1.} We show that any local mechanism minimizing the price of anarchy over all linear local mechanisms, does so also over all linear and non-linear local mechanisms. We let $\opt\taxmechanism$ be a mechanism that minimizes the price of anarchy over all linear local mechanisms, i.e., \mbox{over all $T$ satisfying}
\[
T\left(\sum_{j=1}^m\alpha_j b_j\right)=\sum_{j=1}^m\alpha_j T(b_j),
\]
for all $\alpha_j\ge 0$. We intend to show that $\poa(\opt\taxmechanism) \le \poa(\taxmechanism)$ for any possible $T$ (linear or non-linear). Towards this goal, assume, for a contradiction, that there exists a tolling mechanism $\hat\taxmechanism$ such that 
\be
\poa(\opt\taxmechanism) > \poa(\hat\taxmechanism).
\label{eq:assumptioncontrad}
\ee
Let $\cg_b$ be the class of games in which any resource $e$ can only utilize a resource cost $\ell_e \in \{b_1, \dots, b_m\}$. Since $\cg_b \subset \cg$, we have
\be
    \poa(\hat\taxmechanism) \geq\> \sup_{G \in \cg_b} \frac{\PNEcost(G,\hat\taxmechanism)}{\mincost(G)}.
    \label{eq:poa_on_a_subset}
\ee
Additionally, let $\cg(\mathbb{Z}_{\ge0}) \subset \cg$ be the class of games with $\alpha_j \in \mathbb{Z}_{\ge0}$ for all $j \in \{1, \dots, m\}$, for all resources in $\resset$. Construct the mechanism $\Thathat$ by ``linearizing'' the mechanism $\hat{T}$, i.e., as 
\[
\Thathat(\ell)=\Thathat\left(\sum_{j=1}^{m} \alpha_jb_j\right)=\sum_{j=1}^{m} \alpha_j\hat\taxmechanism(b_j).
\]
We observe that the efficiency of any instance $G\in\mc{G}_b$ to which the tolling mechanism $\hat{T}$ is applied, coincides with that of an instance $G\in\mc{G}(\mathbb{Z}_{\ge0})$ to which $\Thathat$ is applied, and vice-versa. Thus, 
\be
\sup_{G \in \cg_b} \frac{\PNEcost(G,\hat\taxmechanism)}{\mincost(G)}=\> \sup_{G \in \cg(\mathbb{Z}_{\ge0})} 
   \frac{\PNEcost\left(G,\Thathat\right)}{\mincost(G)} = \poa(\Thathat),\label{eq:natural_numbers}
\ee
where the last equality holds due to \cref{lem:resource_decomposition} in \cref{subsec:appendix-optimal-toll}. Putting together \cref{eq:assumptioncontrad,eq:poa_on_a_subset,eq:natural_numbers} gives
\be
\poa(\opt\taxmechanism)>\poa(\Thathat).
\label{eq:togethercontrad}
\ee
Since $\opt\taxmechanism$ minimizes the price of anarchy over all linear mechanisms, and since $\Thathat$ is linear by construction, it must be $\poa(\opt\taxmechanism)\le\poa(\Thathat)$, a contradiction of \eqref{eq:togethercontrad}. Thus, $\opt\taxmechanism$ minimizes the price of anarchy over any mechanism.

\paragraph{Part 2.} We will derive a linear program to design optimal linear mechanisms. Putting this together with the claim in Part 1 will conclude the proof. Towards this goal, we will prove that any mechanism of the form 
\be
\label{eq:moregeneral}
T(\ell)=\sum_{j=1}^m \alpha_j\opt{\tau}_j\quad\text{with}\quad\opt{\tau}_j(x) = \mynu\cdot\opt{f}_j(x)-b_j(x)
\ee
is optimal, regardless of the value of $\mynu\in\mathbb{R}_{>0}$. While this is slightly more general than needed, setting $\mynu=1$ will give the first claim. Additionally, setting $\mynu=\opt\poa$ will give the second claim as this choice will ensure non-negativity of the tolls.

Before turning to the proof, we recall a result from \cite{chandan19optimaljournal} that allows us to compute the price of anarchy for \emph{given} linear tolling mechanism $\taxmechanism(\ell)=\sum_{j=1}^m\alpha_j\tau_j$. Upon defining $f_j(x) = b_j(x) + \tau_j(x)$ for all $1 \leq x \leq n$ and $j \in \{1, \dots, m\}$, the authors show that the price of anarchy of $T$ computed over congestion games $\mc{G}$ is identical for pure Nash and coarse correlated equilibria and is given by $\poa(\taxmechanism) = 1/\opt{\rho}$, where $\opt{\rho}$ is the value of the following program
	\begin{maxi}
	{{\scriptstyle \rho \in \mathbb{R}, \nu \in \mathbb{R}_{\geq 0}}}
	{\nquad\rho} {\label{linprog:characterize_poa}}{}
	\addConstraint{\!\!\!\!\!\!\!\!\!\!\!\!b_j(x+z)(x+z) - \rho b_j(x+y)(x+y) + \nu [f_j(x+y)y - f_j(x+y+1)z]}{\ge 0\quad}{\forall\,(x,y,z) \in \mc{I},}
	\addConstraint{}{}{\forall j\in\{1,\dots,m\},}
	\end{maxi}
We also remark that, when all functions $\{f_j\}_{j=1}^m$ are non-decreasing, it is sufficient to only consider a reduced set of constraints, following a similar argument to that in \cite[Cor. 1]{paccagnan2018distributed}. In this case, the linear program simplifies to
	\begin{maxi}
	{{\scriptstyle \rho \in \mathbb{R}, \nu \in \mathbb{R}_{\geq 0}}}
	{\rho} {\label{eq:simplifiedLP-compute}}{}
	\addConstraint{b_j(\v)\v - \rho b_j(\u)\u + \nu [f_j(\u)\u - f_j(\u+1)\v]}{\ge 0}
	\addConstraint{}{}{\hspace*{-40mm}\forall u,v\in\{0,\dots,n\}\quad u+\v \leq n,\quad  \forall j \in \{1, \dots, m\},}
	\addConstraint{b_j(\v)\v - \rho b_j(\u)\u + \nu [f_j(\u)(n-\v) - f_j(\u+1)(n-\u)]}{\ge 0}
	\addConstraint{}{}{\hspace*{-40mm}\forall u,v\in\{0,\dots,n\}\quad u+\v > n,\quad  \forall j \in \{1, \dots, m\}.}
	\end{maxi}
We now leverage \eqref{linprog:characterize_poa} to prove that any mechanism in \eqref{eq:moregeneral} is optimal, as required. Towards this goal, we begin by observing that the optimal price of anarchy obtained when the resource costs are generated using all the basis functions $\{b_1,\dots,b_m\}$ is no smaller than the optimal price of anarchy obtained when the resource costs are generated using a single basis function $\{b_j\}$ at a time (and therefore is no smaller than the highest of these optimal price of anarchy values). This follows readily since the former class of games is a superset of the latter. Additionally, observe that a set of tolls minimizing the price of anarchy over the games generated using a single basis function $\{b_j\}$ is precisely that in \eqref{eq:moregeneral}. This is because minimizing the price of anarchy amounts to designing $f_j$ to maximize $\rho$ in \eqref{linprog:characterize_poa}, i.e., to solving the following program
\begin{maxi*}
	{{\scriptstyle f\in\mb{R}^n}}
	{\nquad\max_{\rho \in \mathbb{R}, \nu \in \mathbb{R}_{\geq 0}} \rho} {}{}
	\addConstraint{\!\!\!\!\!\!b_j(x+z)(x+z) - \rho b_j(x+y)(x+y) + \nu [f(x+y)y - f(x+y+1)z]}{\ge 0\quad}{\forall\,(x,y,z) \in \mc{I},}
	\end{maxi*}
which can be equivalently written as
\begin{maxi*}
	{{\scriptstyle \tilde{f}\in\mb{R}^n,\,\rho\in\mb{R}}}
	{\nquad \rho}{}{}
	\addConstraint{\nquad b_j(x+z)(x+z) - \rho b_j(x+y)(x+y) + \tilde{f}(x+y)y - \tilde{f}(x+y+1)z}{\ge 0\quad}{\forall\,(x,y,z) \in \mc{I},}
\end{maxi*}
where we defined $\tilde f = \nu \cdot f$. While $\opt{f}_j$ is defined in \eqref{eq:mainLPopt} precisely as the solution of this last program, resulting in a price of anarchy of $1/\opt\rho_j$, note that $\mynu\cdot \opt{f}_j$ is also a solution since its price of anarchy matches $1/\opt\rho_j$ (in fact, it can be computed using \eqref{linprog:characterize_poa} for which $(\rho,\nu)=(\opt\rho_j,1/\mynu)$ are feasible). 

The above reasoning shows that the optimal price of anarchy for a game with resource costs generated by $\{b_1,\dots,b_m\}$ must be no smaller than $\max_{j} \{1/\opt{\rho}_j\}$. We now show that this holds with equality. Towards this goal, we note, thanks to \eqref{linprog:characterize_poa}, that utilizing tolls as in \eqref{eq:assumptioncontrad} for a game generated by $\{b_1,\dots,b_m\}$ results in a price of anarchy of precisely $\max_j\{1/\opt\rho_j\}$. This follows as $(\min_j\{\opt\rho_j\},1/\mynu)$ is feasible for this program for any choice of  $\mynu>0$. This proves, as requested, that any tolling mechanism defined in \eqref{eq:moregeneral} is optimal.

We now verify that the choice $\mynu = \opt\poa = \max_j\{1/\opt\rho_j\}$ ensures positivity of the tolls, which is equivalent to $\opt f_j(x)-b_j(x)/\mynu \ge0$ for all $x\in\{1,\dots,n\}$. This follows readily, as setting $x=z=0$ in \eqref{eq:mainLPopt} results in the constraint $f(y)- \rho b_j(y)\ge0$ for all $y\in\{1,\dots,n\}$. Since $\opt f_j$ and $\opt \rho_j$ must be feasible for this constraint, we have $\opt f_j(y)- \opt\rho_j b_j(y)\ge0$. One concludes observing that $\opt f_j(y)- b_j(y)/\mynu \ge \opt f_j(y)- \opt\rho_j b_j(y)\ge0$, since $\mynu\ge 1/\opt\rho_j$. We conclude remarking that all results hold for both Nash and coarse correlated equilibria, as they were derived from \eqref{linprog:characterize_poa}.
\end{proof}

%% file: parts-supplementary/supplementary-explicitsolution.tex
\section{Explicit solution and simplified Linear program}
\label{sec:explicit-expression}
In this section we derive a simplified linear program as well as an analytical solution to the problem of designing optimal tolling mechanisms. We do so under the assumption that all basis functions are positive, increasing, and convex in the discrete sense.\footnote{We say that a function $f:\{1,\dots,n\}\rightarrow\mb{R}$ is convex if $f(x+1)-f(x)$ is non-decreasing in its domain.\label{foot:cvx}}
\begin{theorem}
\label{thm:simplifiedLP-explicit-solution}
Consider congestion games with $n$ agents, where resource costs take the form $\ell(x)=\sum_{j=1}^m\alpha_jb_j(x)$, $\alpha_j\ge0$, and basis $b_j:\{1,\dots,n\}\rightarrow\mb{R}$ are positive, convex, strictly increasing.\footnote{The result also holds if convexity and strict increasingness of $b_j(x)$ are weakened to strict convexity of $b_j(x)x$ and $b_j(n)>b_j(n-1)$. One such example is that of $b_j(x)=\sqrt{x}$.}

\begin{itemize}[leftmargin=*]
\item[i)]
A tolling mechanism minimizing the price of anarchy is as in \eqref{eq:optimaltolls_expression}, where each $\opt{f}_j:\{1,\dots,n\}\rightarrow\mb{R}$ solves the following simplified linear program
\begin{maxi}
	{{\scriptstyle f\in\mb{R}^n,\,\rho\in\mb{R}}}
	{\nquad \rho}{\label{eq:simplifiedLP-cvx}}{\opt{\rho}_j=}
	\addConstraint{\!\!\!\!\!\!\!\!\!\!\!\!\!\! b_j(v)v - \rho b_j(u)u + f(u)u - f(u+1)v}{\ge 0\quad}{\forall u,v\in\{0,\dots,n\}\quad u+v \leq n,}
	\addConstraint{\!\!\!\!\!\!\!\!\!\!\!\!\!\! b_j(v)v - \rho b_j(u)u + f(u)(n-u) - f(u+1)(n-u)}{\ge0\quad}{\forall u,v\in\{0,\dots,n\}\quad u+v > n,}
\end{maxi}
with $f(0)=f(n+1)=0$. The corresponding optimal price of anarchy is $\max_j\{1/\opt{\rho}_j\}$.
\item[ii)] An explicit expression for each $\opt{f}_j$ is given by the following recursion, where $\opt{f}_j(1)=b_j(1)$, 
\be
\label{eq:fopt-cvx}
\begin{split}
\opt{f}_j(u+1)&=\min_{v\in\{1,\dots,n\}}
\beta(u,v) \opt{f}_j(u) + \gamma(u,v) - 
\delta(u,v) \opt{\rho}_j,\\
\beta(u,v)=
\frac{\min\{u,n-v\}}{\min\{v,n-u\}}&,\quad
\gamma(u,v)=
\frac{b(v)v}{\min\{v,n-u\}},\quad
\delta(u,v)=
\frac{b(u)u}{\min\{v,n-u\}},
\end{split}
\ee
\be
\label{eq:defrhoopt}
\opt{\rho}_j=\hspace*{-1mm} 
\min_{(v_1,\dots,v_n)\in\{1,\dots,n\}^{n-1}\times \{0,\dots,n\}}
\hspace*{-2mm} 
\frac
{
(n-v_n)\left(\prod_{u=1}^{n-1} \beta_u{b_j(1)}  + \sum_{u=1}^{n-2} \left(\prod_{i=u+1}^{n-1} \beta_i \right) \gamma_u+ \gamma_{n-1}\right)+b(v_n)v_n}
{
(n-v_n)\left(
\sum_{u=1}^{n-2}\left(\prod_{i=u+1}^{n-1}\beta_i\right)\delta_u+ \delta_{n-1}
\right)+b(n)n},
\ee
where we use the short-hand notation $\beta_u$ instead of $\beta(u,v_u)$, and similarly for $\gamma_u$ and $\delta_u$.
\end{itemize}
\end{theorem}
Before delving into the proof, we observe that the key difficulty in designing optimal tolls resides in the expressions of $\opt\rho_j$ arising from \eqref{eq:defrhoopt}. Nevertheless, for any possible choice of $\bar{\rho}_j$ that approximates $\opt\rho_j$ from below, i.e., $\bar{\rho}_j \le \opt\rho_j$, one can directly utilize the recursion in \eqref{eq:fopt-cvx} to design a valid tolling mechanism. The resulting price of anarchy would then amount to $\max_j\{1/\bar{\rho}_j\} > \max_j\{1/\opt{\rho}_j\}$. This follows from the ensuing proof.

\begin{proof}
As shown in \cref{thm:main-thm}, computing an optimal tolling mechanism amounts to utilizing \eqref{eq:optimaltolls_expression}, where each $\opt\tau_j$ has been designed through the solution of the program in \eqref{eq:mainLPopt}. In light of this, we prove the theorem as follows: first, we consider a simplified linear program, where only a subset of the constraints enforced in \eqref{eq:mainLPopt} are considered. Second, we show that a solution of this simplified program is given by $(\opt\rho_j,\opt{f}_j)$ as defined above. Third, we show that $\opt{f}_j$ is non-decreasing, thus ensuring that $(\opt\rho_j,\opt{f}_j)$ is also feasible for the original over constrained program in \eqref{eq:mainLPopt}. From this we conclude that $(\opt\rho_j,\opt{f}_j)$ must also be a solution of \eqref{eq:mainLPopt}, i.e., the second claim in the Theorem. We conclude with some cosmetics, and transform the simplified linear program whose solution is given by $(\opt\rho_j,\opt{f}_j)$ in \eqref{eq:simplifiedLP-cvx}, thus obtaining the first claim. Throughout the proof, we drop the index $j$ from $b_j$ as the proof can be repeated for each basis separately.

\subsubsection*{Simplified linear program}
We begin rewriting the program in \eqref{eq:mainLPopt}, where instead of the indices $(x,y,z)$, we use the corresponding indices $(u,v,x)$ defined as $u=x+y$, $v=x+z$. The constraint indexed by $(u,v,x)$ reads as $b(v)v-\rho b(u)u+f(u)(u-x)-f(u+1)(v-x)\ge0$. We now consider only the constraints where $x$ is set to $x=\min\{0,u+v-n\}$, and $u,v\in\{0,\dots,n\}$, Such constraints read as $b(v)v-\rho b(u)u+\min\{u,n-v\}f(u)-\min\{v,n-u\}f(u+1)\ge0$.\footnote{Note that considering all these constraints with $u,v\in\{0,\dots,n\}$ results precisely in \eqref{eq:simplifiedLP-cvx}. To see this, simply distinguish the cases based on whether $u+v\le n$ or $u+v>n$.} Finally, we exclude the constraints with $v=0$, $u\in\{1,\dots,n-1\}$ and obtain the following \emph{simplified linear program}
\begin{maxi}
	{{\scriptstyle f\in\mb{R}^n,\,\rho\in\mb{R}}}
	{\rho}{\label{eq:oversimplifiedLP-cvx}}{}
	\addConstraint{b(v)v-\rho b(u)u+\min\{u,n-v\}f(u)-\min\{v,n-u\}f(u+1)}{\ge 0}{}
	\addConstraint{}{}{\hspace*{-30mm}\forall (u,v)\in \{0,\dots,n\}\times\{1,\dots,n\}\cup (n,0)}
\end{maxi}
\subsubsection*{Proof that $(\opt{\rho},\opt{f})$ solve \eqref{eq:oversimplifiedLP-cvx}} 
Towards the stated goal, we begin by observing that $(\opt{\rho},\opt{f})$ is feasible by construction. For $u=0$ this follows as the tightest constraints in \eqref{eq:oversimplifiedLP-cvx} read as $\opt{f}(1)\ge b(1)$ and we selected $\opt{f}(1)= b(1)$. Feasibility is immediate to verify for $u\in\{1,\dots,n-1\}$,  $v\in\{1,\dots,n\}$ as applying its definition gives $\opt{f}(u+1)\le \beta(u,v) \opt{f}(u) + \gamma(u,v) - \delta(u,v) \opt{\rho}$. Using the expressions of $\beta,\gamma,\delta$, and rearranging gives exactly the constraint $(u,v)$ in \eqref{eq:oversimplifiedLP-cvx}. The only element of difficulty consists in showing that also the constraints with $u=n$, $v\in\{0,\dots,n\}$ are satisfied. Towards this goal, we observe that utilizing the recursive definition of $\opt{f}$ we obtain an expression for $\opt{f}(n)$ as a function of $\opt{\rho}$ with a nested succession of minimizations, which can be jointly extracted as follows
\[
\opt{f}(n)=\min_{v_{n-1}}\left\{\dots+\min_{v_{n-2}}\left\{\dots+\min_{v_{1}}\left\{\dots\right\}\right\}\right\}
=
\min_{v_{n-1}}\min_{v_{n-2}}\dots\min_{v_{1}}\left\{\dots\right\}.
\]
This holds as $\opt{f}(u+1)=\min_{v_u}\left[\beta_u\min_{v_{u-1}}(\beta_{u-1}\opt{f}(u-1)-\delta_{u-1}\opt{\rho}+\gamma_{u-1})-\delta_u\opt{\rho}+\gamma_u\right]$, and since $\beta_u\geq0$, the latter simplifies to $\opt{f}(u+1)=\min_{v_u}\min_{v_{u-1}}\beta_u\beta_{u-1}\opt{f}(u-1)-(\beta_u\delta_{u-1}+\delta_u)\opt{\rho}+\beta_u\gamma_{u-1}+\gamma_u$. Repeating the argument recursively gives the desired expression. Hence,
\[\begin{split}
\opt{f}(n)&=\min_{(v_1,\dots,v_{n-1})^{}\in\{1,\dots,n\}^{n-1}}
\prod_{u=1}^{n-1} \beta_u {b_j(1)} + \sum_{u=1}^{n-2} \left(\prod_{i=u+1}^{n-1} \beta_i \right) (\gamma_u-\delta_u\opt{\rho})+ (\gamma_{n-1}-\delta_{n-1}\opt{\rho})\\
	 &\doteq \!\! \min_{(v_1,\dots,v_{n-1})\in\{1,\dots,n\}^{n-1}} \q(v_1,\dots,v_{n-1};\opt{\rho}),
\end{split}\]
where we implicitly defined $q(v_1,\dots,v_{n-1};\opt\rho)$. The constraints we intend to verify read as $b(v)v-\rho b(n)n+(n-v)\opt{f}(n)\ge0$ for all $v\in\{0,\dots,n\}$, and can be equivalently written as $\min_{v_n\in\{0,\dots,n\}}[b(v_n)v_n-\rho b(n)n+(n-v_n)\opt{f}(n)]\ge0$. We substitute the resulting expression of $\opt{f}(n)$, extract the minimization over $v_n$ as in the above, and are therefore left with $\min_{(v_1,\dots,v_{n-1},v_n)\in\{1,\dots,n\}^{n-1}\times\{0,\dots,n\}} [b(v_n)v_n-\rho b(n)n+(n-v)\q(v_1,\dots,v_{n-1};\opt{\rho})]\ge0,$ which holds if and only if $b(v_n)v_n-\rho b(n)n+(n-v_n)\q(v_1,\dots,v_{n-1};\opt{\rho})\ge0$ for all possible tuples $(v_1,\dots,v_n)$. Rearranging these constraints and solving for $\opt{\rho}$ will result in a set of inequalities on $\opt\rho$ (one inequality for each tuple). Our choice of $\opt{\rho}$ in $\eqref{eq:defrhoopt}$ is precisely obtained by turning the most binding of these into an equality. This ensures that $(\opt{\rho},\opt{f})$ are feasible also when $u=n$.

We now prove, by contradiction, that $(\opt{\rho},\opt{f})$ is optimal. To do so, we assume that there exists $\hat{f}$, that is feasible and achieves a higher value $\hat{\rho}>\opt{\rho}$. Since $(\hat{f},\hat{\rho})$ is feasible, using the constraint with $u=0$, $v=1$, we have $\hat{f}(1)\le b(1)=\opt{f}(1)$. Observing that $\min\{v,n-u\}>0$ due to $v>0$, $u<n$ and leveraging the constraints with $u=1$ as well as the corresponding specific choice of $v=v_1^\ast$ (for given $u\in\{1,\dots,n-1\}$, we let $v_u^\ast$ be an index $v\in\{1,\dots,n\}$ where the minimum in \eqref{eq:fopt-cvx} is attained), it must be that $\hat{f}(2)$ satisfies
\[
\hat{f}(2)
\!\le\!\frac{b(v_1^\ast)v_1^\ast\!-\!\hat{\rho} b(1) + \min\{1,n-v_1^\ast\}\hat{f}(1)}{\min\{v_1^\ast,n-1\}}\!<\! \frac{b(v_1^\ast)v_1^\ast\!-\!\opt{\rho} b(1) + \min\{1,n-v_1^\ast\}\opt{f}(1)}{\min\{v_1^\ast,n-1\}}\!=\!\opt{f}(2).
\]
Here the first inequality follows by feasibility of $\hat{f}$, the second is due to $\hat{\rho}>\opt{\rho}$ and $\hat{f}(1)\le \opt{f}(1)$. The final equality follows due to the definition of $\opt{f}(2)$. Hence we have shown that $\hat{f}(2)<\opt{f}(2)$. Noting that the only information we used to move from level $u$ to $u+1$ is that $\hat{\rho}>\opt{\rho}$ and $\hat{f}(u)\le \opt{f}(u)$, one can apply this argument recursively up until $u=n-1$, and thus obtain $\hat{f}(n)<\opt{f}(n)$. Nevertheless, leveraging the constraints with $u=n$ and $v=v_n^\ast$ gives $b(v_n^\ast)v_n^\ast-\hat{\rho}b(n)n+(n-v_n^\ast)\hat{f}(n)\ge0$, or equivalently $\hat{\rho}\le (b(v_n^\ast)v_n^\ast+(n-v_n^\ast)\hat{f}(n))/(b(n)n)$. Thus
\[
\hat{\rho} \le
\frac{b(v_n^\ast)v_n^\ast+(n-v_n^\ast)\hat{f}(n)}{b(n)n}
\le \frac{b(v_n^\ast)v_n^\ast+(n-v_n^\ast)\opt{f}(n)}{b(n)n}=\opt{\rho},
\]
where we used the fact that $n-v_n^\ast\ge0$ and $\hat{f}(n)<\opt{f}(n)$. Note that $\hat{\rho}\le \opt{\rho}$ contradicts the assumption $\hat{\rho}>\opt{\rho}$, thus concluding this part of the proof.

\subsubsection*{Proof that $\opt{f}$ is non-decreasing}
By contradiction, let us assume $\opt{f}$ is decreasing at some index. \cref{lem:f_nondecreasing} in the Appendix shows that, if this is the case, then $\opt{f}$ continues to decrease, so that $\opt{f}(n)\le \opt{f}(n-1)$. Note that it must be $\opt{f}(n)>0$, as if it were $\opt{f}(n)\le 0$, then by definition of $\opt{\rho}$ we would have 
\[
\opt{\rho} = \min_{v \in\{0,\dots,n\}}\frac{b(v)v+(n-v)\opt{f}(n)}{nb(n)} =\frac{0+\opt{f}(n)}{b(n)}\le0,
\]
since the minimum is attained at the lowest feasible $v$ due to $b(v)v$ and $-v\opt{f}(n)$ non-decreasing and increasing, respectively. This is a contradiction as the price of anarchy is bounded already in the un-tolled setup.\footnote{To see this, consider the linear program used to determine the price of anarchy in the un-tolled case, i.e., \eqref{eq:simplifiedLP-compute} where we set $f_j(x)=b_j(x)$. When $\nu=1$, it is always possible to find $\rho>0$, so that the corresponding price of anarchy is bounded.}
It must therefore be that the price of anarchy is bounded also when optimal tolls are used. Additionally, as we have removed a number of constraints from the linear program, the corresponding price of anarchy will be even lower. Therefore it must be that $1/\opt{\rho}$ is non-negative and bounded, so that $\opt{\rho} >0$ contradicting the last equation.

Thus, in the following we proceed with the case of $\opt{f}(n)>0$.  It must be that
\[
\opt{\rho}=
	\!\!\!\min_{v \in\{0,\dots,n\}}\!\!\!\frac{b(v)v\!+\!(n\!-\!v)\opt{f}(n)}{nb(n)}
	\le \!\!\!
	\min_{v \in\{1,\dots,n\}}\!\!\!\frac{b(v)v\!+\!(n\!-\!v)\opt{f}(n)}{nb(n)}\le 
	 \!\!\!\min_{v \in\{1,\dots,n\}}\!\!\!\frac{b(v)v\!+\!(n\!-\!v)\opt{f}(n\!-\!1)}{nb(n)},
\] 
where the first inequality holds as we are restricting the domain of minimization, the second because $\opt{f}(n)\le \opt{f}(n-1)$ and $n-v\ge 0$. Let us observe that $\opt{f}(n)$ is defined as $f(n)=\min_{v\in\{1,\dots,n\}} [b(v)v+(n-v)\opt{f}(n-1)] - \opt{\rho} (n-1)b(n-1).$ Substituting $\min_{v\in\{1,\dots,n\}} [b(v)v+(n-v)\opt{f}(n-1)]=\opt{f}(n)+\opt{\rho} (n-1)b(n-1)$ in the former bound on $\opt{\rho}$, we get
\[
\opt{\rho}\le \frac{\opt{f}(n)+\opt{\rho} (n-1)b(n-1)}{nb(n)}
\quad\implies\quad
\opt{\rho}\le \frac{\opt{f}(n)}{nb(n)-(n-1)b(n-1)}.
\]

\noindent We want to prove that this gives rise to a contradiction. To do so, we will show that
\be
\label{eq:contradiction}
\frac{\opt{f}(n)}{nb(n)-(n-1)b(n-1)} <
 \min_{v \in\{0,\dots,n\}}\frac{b(v)v+(n-v)\opt{f}(n)}{nb(n)}.
\ee
As a matter of fact, if the latter inequality holds true, the proof is immediately concluded as
\[
\opt{\rho} \le \frac{\opt{f}(n)}{nb(n)-(n-1)b(n-1)} <
 \min_{v \in\{0,\dots,n\}}\frac{b(v)v+(n-v)\opt{f}(n)}{nb(n)}=\opt{\rho}\quad \implies\quad \opt{\rho}<\opt{\rho},
\]
where the first inequality has been shown above, the second is what remains to be proved, and the latter equality is by definition. Therefore, we are left to show \eqref{eq:contradiction}, which holds if we can show that $\forall v\in\{0,\dots,n\}$ it is
\[
g(v)\doteq\frac{h(v)+(n-v)\opt{f}(n)}{h(n)}-\frac{\opt{f}(n)}{h(n)-h(n-1)} >0,
\]
where $h:\mb{R}\to\mb{R}_{\geq0}$ is a function such that $h(v)=b(v)v$ for $v\in\{0,\dots,n\}$. We choose $h$ to be continuously differentiable, strictly increasing, and strictly convex; one such function always exists.\footnote{Observe that the function $b(v)v$ is positive, strictly increasing, and strictly convex in the discrete sense in its domain due to the assumptions.} We first consider the point $v=0$. Observe that $g(0)>0$ when $n>1$ as 
\[
g(0)=\frac{\opt{f}(n)}{b(n)}-\frac{\opt{f}(n)}{nb(n)-(n-1)b(n-1)}>0
\quad\iff\quad
\opt{f}(n)[(n-1)b(n)-(n-1)b(n-1)]>0,
\]
which holds as $\opt{f}(n)>0$, $n>1$, and $b(n)>b(n-1)$ strictly.

If $g'(v)\ge 0$ at $v=0$, the proof is complete as $g$ is convex and due to $g'(0)\ge 0$ it is non-decreasing for any $v\ge 0$ so that the constraint will be satisfied for all $v\ge 0$. 

If this is not the case, then $g'(0)< 0$, which we consider now.
Note that, at the point $v=n-1$, the derivative $g'(n-1)=[h'(n-1)-\opt{f}(n)]/h(n)$ satisfies 
\[
h(n) g'(n-1)=h'(n-1)-\opt{f}(n)
\ge h'(n-1) -(h(n-1)-h(n-2))\ge0
\]
where the last inequality is due to convexity, while the first inequality holds as $\opt{f}(n)\le h(n-1)-h(n-2)$ thanks  to \cref{lem:f_nondecreasing} and $n\ge 2$.\footnote{In fact, either $n$ is the first index starting from which $\opt{f}$ decreases (i.e. $\opt{f}(n)<\opt{f}(n-1)$) in which case $\opt{f}(n)\le\opt{\rho}[b(n-1)(n-1)-b(n-2)(n-2)]\le b(n-1)(n-1)-b(n-2)(n-2)$ due to $\opt{\rho}\le 1$, or the function starts decreasing at a $u+1<n$ in which case \cref{lem:f_nondecreasing} also shows that
\[
\opt{f}(n)\le\dots\le \opt{f}(u+1)\le \opt{\rho}[b(u)u-b(u-1)(u-1)]\le b(u)u-b(u-1)(u-1)\le b(n-1)(n-1)-b(n-2)(n-2),
\]
where the inequalities hold due to $\opt{\rho}\le 1$ and the convexity of $b(u)u$.} Therefore since $g'(0)<0$, $g'(n-1)\ge0$ and $g$ convex, there must exist an unconstrained minimizer $v^\star \in (0,n-1]$. We will guarantee that $g(v^\star)>0$ so that for any (real and thus integer) $v\in[0,n]$ it is $g(v)>0$. The unconstrained minimizer satisfies $\opt{f}(n)=h'(v^\star)$, which we substitute, and are thus left with proving the final inequality 
\[
\frac{h(v^\star)+(n-v)h'(v^\star)}{h(n)}-\frac{h'(v^\star)}{h(n)-h(n-1)} > 0,
\]
which is equivalent to 
\[
[h(n)-h(n-1)]h(v^\star)>h'(v^\star)[(n-v^\star)(h(n-1)-h(n))+h(n)],
\]
where we recall $0<v^\star\le n-1$. As the left hand side is positive due to $h$ increasing and $v^\star>0$, the inequality holds trivially if the right hand side is less or equal to zero, i.e., if $h(n)\le(n-v^\star)(h(n)-h(n-1))$. In the other case, when  $(n-v^\star)(h(n-1)-h(n))+h(n)>0$, we leverage the fact that $h'(v^\star)< (h(n)-h(v^\star))/(n-v^\star)$ by strict convexity of $h(x)$ in $x=v^\star>0$, so that
\[
\begin{split}
h'(v^\star)[(n-v^\star)(h(n-1)-h(n))+h(n)]&\!< \frac{h(n)-h(v^\star)}{n-v^\star}[(n-v^\star)(h(n-1)-h(n))+h(n)]\\
&\!=\!\frac{h(n)}{n-v^\star}[h(n)-h(v^\star)]\!+\![h(n)-h(n-1)][h(v^\star)-h(n)]\\
&\!\le[h(n)-h(n-1)]h(v^\star),
\end{split}
\]
where the last inequality follows since $[h(n)-h(n-1)][h(v^\star)-h(n)]\le0$ and from $\frac{h(n)-h(v^\star)}{n-v^\star}\le h(n)-h(n-1)$, which holds for $0<v^\star \le n-1$ by convexity. This concludes this part of the proof.

\subsubsection*{Proof that $(\opt{\rho},\opt{f})$ is feasible also for \eqref{eq:mainLPopt} and final cosmetics.} Recall from the first part of the proof that the constraints in \eqref{eq:mainLPopt} can be equivalently written as $b(v)v-\rho b(u)u+f(u)(u-x)-f(u+1)(v-x)\ge0$. Since $\opt{f}$ is non-decreasing,  following the argument in \cite[Cor. 1]{paccagnan2018distributed} one verifies that the tightest constraints are obtained when $x=\min\{0,u+v-n\}$. These constraints are already included in our simplified program of \eqref{eq:oversimplifiedLP-cvx}, with the exception of those with $v=0$ and $u\in\{0,\dots,n-1\}$ which we have removed. To show that also these hold, we note that the constraint with $v=0$ reads as $u\opt{f}(u)\ge\opt{\rho} ub(u)$, and is trivially satisfied for $u=0$. We now show that also the constraints with $v=0$, $u>0$ hold. To do so, we consider the constraint corresponding to $v=1$ 
\[
b(1)1-\rho b(u)u + u\opt{f}(u)-\opt{f}(u+1)\ge0.
\]
Since $\opt{f}$ is non-decreasing as shown in previous point then $\opt{f}(u+1)\ge \opt{f}(1)=b(1)$. Hence,
\[
0\le b(1)1-\rho b(u)u + u\opt{f}(u)-\opt{f}(u+1)\le  b(1)1-\opt{\rho} b(u)u + u\opt{f}(u)-b(1).
\]
Thus, from the left and right hand sides we obtain the desired result $u\opt{f}(u)\ge \opt{\rho} ub(u)$.  

We conclude with some cosmetics: the simplified linear program in \eqref{eq:oversimplifiedLP-cvx} is almost identical to that in \eqref{eq:simplifiedLP-cvx}, except for the constraints with $v=0$ and $u\in\{0,\dots,n-1\}$, which we have removed in \eqref{eq:oversimplifiedLP-cvx}. Nevertheless, we have just verified that an optimal solution does satisfy these constraints too. Hence, we simply add them back to obtain \eqref{eq:simplifiedLP-cvx}.
\end{proof}

%% file: parts-supplementary/supplementary-largen.tex
\section{Optimal Tolling Mechanisms for arbitrary number of agents}
\label{sec:largen}
While the linear programming formulations introduced in \eqref{eq:mainLPopt} and \eqref{eq:simplifiedLP-cvx} provide an optimal tolling mechanism and the corresponding optimal price of anarchy when the number of agents is upper-bounded by $n$ (finite), in this section we show how to design optimal tolling mechanisms for polynomial congestion games that apply to any $n$ (possibly infinite), by solving a linear program of fixed size. The resulting values of the price of anarchy are those already displayed in \cref{tab:naware_vs_nagnostic}.

For ease of exposition, we consider congestion games where the set of resource costs is produced by non-negative combinations of a \emph{single} monomial $x^d$, $d\ge 1$ at a time. This is without loss of generality, as one can derive optimal tolling mechanisms for polynomial congestion games with \emph{maximum degree} $d$, i.e., generated by $\{1,x,\dots,x^d\}$,  simply repeating the ensuing reasoning separately for all polynomials of degree higher than one and lower-equal to~$d$. No toll need to be applied to polynomials of order zero as the corresponding price of anarchy is one. 

The idea we leverage is as follows: first, we solve a linear program of \emph{fixed} size $\bar n$, from which we obtain a set of tolls that are then extended analytically to any number of agents. This produces a mechanism for which we are able to quantify the corresponding price of anarchy over games with possibly infinitely many agents. Such price of anarchy value is an upper bound on the true optimal price of anarchy over games with possibly infinitely many agents, as the mechanism we design is not necessarily optimal. At the same time, we solve the linear program in \eqref{eq:simplifiedLP-cvx}, and thus obtain the optimal price of anarchy for games with a maximum of $\bar n$ agents. The latter is a lower bound for the optimal price of anarchy over games with possibly infinitely many agents. Letting $\bar n$ grow, the upper bound matches the lower bound already for small values of $\bar n$, as showcased in \cref{tab:ub-lb}.

While the construction of the lower bound follows readily by solving the linear program in \eqref{eq:simplifiedLP-cvx} with $\bar{n}$ agents, in the following we describe the procedure to derive the upper bound. More specifically, we clarify i) what program of dimension $\bar n$ we solve; ii) how we extend its solution from $\bar n$ to infinity; and iii) how we compute the resulting price of anarchy over games with possibly infinitely many agents. In the remainder of this section, we will always select $\bar n$ finite and even.

As for the first point, we consider the following linear program
\begin{maxi}
	{{\scriptstyle f\in\mb{R}^{\bar n},\,\rho\in\mb{R}}}
	{\!\!\!\!\!\rho}{\label{eq:overconstraintedsimplifiedlp}}{}
	\addConstraint{\!\!\!\!\!\v^{d+1} - \rho \u^{d+1} + f(\u)\u - f(\u+1)\v}{\ge 0\quad}{\forall u,v\in\{0,\dots,{\bar n}\}\quad u+v \leq {\bar n},}
	\addConstraint{\!\!\!\!\!\v^{d+1} - \rho \u^{d+1} + f(\u)({\bar n}-\v) - f(\u+1)({\bar n}-\u)}{\ge0\quad}{\forall u,v\in\{0,\dots,{\bar n}\}\quad u+v > {\bar n},}
	\addConstraint{\!\!\!\!\!f(u)}{\le u^d\quad}{\hspace*{2.8mm}\forall u\in\{1,\dots,\bar n\}}
	\addConstraint{\!\!\!\!\!f(u)}{\ge f(u-1)\quad}{\hspace*{2.8mm}\forall u\in\{2,\dots,\bar n\}}
\end{maxi}
with the usual convention that $f(0)=f(\bar{n}+1)=0$. Note that the previous program is identical to that in \eqref{eq:simplifiedLP-cvx}  with $b(x)=x^d$, except that we have included two additional sets of constraints. We let $(f^{\rm opt}, \rho^{\rm opt})$ be a solution of this program and utilize it to define $\myf:\mb{N}\rightarrow\mb{R}$ as follows
\be
\myf(x)=
\begin{cases}
f^{\rm opt}(x)\quad&\text{for}\quad x\le \bar{n}/2\\
\beta \cdot x^d \quad&\text{for}\quad x> \bar{n}/2\\
\end{cases},\quad \text{where}\quad \beta = \frac{f^{\rm opt}(\bar{n}/2)}{(\bar{n}/2)^d}.
\label{eq:definefext}
\ee
Informally, the idea is to extend $\tau^\infty(x)=\myf(x)-b(x)$ from $\bar{n}/2$ to infinity with a polynomial of the same order of the original $x^d$. Note that $\beta\ge0$ is chosen so that the two expressions defining $\myf$ match for $x=\bar{n}/{2}$.\footnote{Observe that $\myf(1)\ge0$ since having $\myf(1)<0$ would always result in a lower performance, as shown in \cite{paccagnan2018distributed}. Therefore $\myf(\bar{n}/2)\ge0$ as it is feasible for \eqref{eq:overconstraintedsimplifiedlp}, which includes the constraint $f(x+1)\ge f(x)$. Hence, $\beta\ge0$.} While the expression of $\myf$ and all forthcoming quantities depends on the choice of $\bar{n}$, we do not make this explicit to ease the notation. \cref{lem:largen} in the Appendix ensures that the price of anarchy of $\myf$ is identical for pure Nash and coarse correlated equilibria, and is upper bounded over games with possibly infinitely many agents by $1/\rho^\infty$, where $\rho^\infty$ is given by
\[
\rho^{\infty}=\min\left\{{\opt\rho},~~{\beta - d\left(1+\frac{2}{\bar n}\right)^{d+1} \left(\frac{\beta}{d+1}\right)^{1+\frac{1}{d}}} \right\}.
\]
As clarified above this represents an upper bound on the optimal price of anarchy. The upper and lower bounds displayed in \cref{tab:ub-lb} have been computed according to the procedure just described, and demonstrate that, for a relatively small $\bar{n}=40$, the mechanism obtained from $\myf$ is approximately optimal up to the fifth decimal digit for polynomial congestion games with $d=1,2,3$. 

Finally, we observe that the tolling mechanism $T^\infty(\alpha \ell)=\alpha T^\infty(\ell)=\alpha\tau^{\infty} $, where $\tau^{\infty}(x)=\myf(x)-b(x)$  might not satisfy $\tau^{\infty}(x)\ge0$ for all $x\in\mb{N}$ (i.e., they might be monetary incentives and not tolls). Nevertheless, multiplying $\myf$ with a factor $\gamma>0$ produces tolls $\gamma\myf(x)-b(x)$ with identical price of anarchy (the proof of \cref{lem:largen} will hold with $\nu=1/\gamma$ in place of $\nu=1$). Therefore, one simply needs to consider tolls of the form $\gamma\myf(x)-b(x)$, where $\gamma$ is chosen sufficiently large to ensure that $\gamma\myf(x)-b(x)\ge0$ for all $x\in\mb{N}$; one such $\gamma$ always exists. When multiple basis are present, we select $\gamma$ as a common scaling factor to ensure non-negativity of all tolls basis.

\input{parts-body/table-ub-lb.tex}

%% file: parts-body/table-ub-lb.tex
\begin{table}[t!]%
\centering%
    \begin{tabular}{S[table-format=2]||S[table-format=2.6]|S[table-format=1.7]|S[table-format=2.6]|S[table-format=1.7]|S[table-format=3.6]|S[table-format=2.7]}%
    \multicolumn{1}{c||}{$\bar n$} & \multicolumn{2}{c|}{$d=1$} & \multicolumn{2}{c|}{$d=2$} & \multicolumn{2}{c}{$d=3$}\\
    &{LB}&{UB}& {LB}&{UB}& {LB}&{UB} \\
        \hline
        10 & 2.011825 & 2.038237 & 5.097187 & 5.316382 &15.530175 & 17.138429\\
        20 & 2.012067 & 2.019844 & 5.100974 & 5.147543 &15.550847 & 15.751993\\
        30 & 2.012067 & 2.014335 & 5.100974 & 5.119149 &15.550852 & 15.684195 \\
        40 & 2.012067 & 2.012067 & 5.100974 & 5.100974 &15.550852 & 15.550859\\
    \end{tabular}

\vspace*{1mm}
    \caption{Lower and upper bounds (LB and UB) on the values of the optimal price of anarchy for polynomial congestion games with arbitrarily large number of agents and $d=1,2,3$. The LB vs UB shows how the tolls derived from $\myf$ defined in  \eqref{eq:definefext} are approximately optimal up the fifth decimal digit when we select $\bar n=40$.}%
\label{tab:ub-lb}%
\vspace*{-8mm}
    \end{table}

%% file: parts-supplementary/supplementary-constanttolls.tex
\section{Congestion-Independent Tolling Mechanisms}
\label{sec:congestion-indep}
In this section we provide a general methodology to compute optimal congestion-independent local tolling mechanisms for games generated by $\{b_1,\dots,b_m\}$. We also specialize the result to polynomial congestion games providing explicit expressions for the tolls and the corresponding price of anarchy. In this section we consider basis functions that are \mbox{convex in the discrete sense (see \cref{foot:cvx}).}
\begin{theorem}
\label{thm:congestion-indep}
{A local congestion-independent mechanism minimizing the price of anarchy over congestion games with $n$ agents, and resource costs $\ell(x)=\sum_{j=1}^m\alpha_j b_j(x)$, $\alpha_j\ge0$, with convex positive non-decreasing basis functions $\{b_1,\dots,b_m\}$ is given by
	\be
	\label{eq:optconsttoll}
    \opt{T}(\ell) = \sum_{j = 1}^{m} \alpha_j\cdot \opt{\tau}_j, 
    \quad\text{where}~~\opt{\tau}_j \in\mb{R}, 
    \quad\opt{\tau}_j = \left(\frac{1}{\opt\nu}-1\right)b_j(1)
    \ee 
  and $\opt{\rho}\in\mb{R}$, $\opt{\nu}\in\mb{R}_{\ge0}$ solve the linear program 
\begin{maxi}
	{{\scriptstyle \rho \in \mathbb{R}, \nu \in [0,1]}}
	{\nquad\rho} {\label{eq:fixedLPopt}}{}
	\addConstraint{\nquad b_j(\v)\v - \rho b_j(\u)\u + \nu [b_j(\u)\u - b_j(\u+1)\v] + b_j(1)(1-\nu)(\u-\v)}{\ge 0}
	\addConstraint{}{}{\hspace*{-65mm}\forall u,v\in\{0,\dots,n\}\quad u+\v \leq n\quad u\ge v,\quad  \forall j \in \{1, \dots, m\},}
	\addConstraint{\nquad b_j(\v)\v - \rho b_j(\u)\u + \nu [b_j(\u)(n-\v) - b_j(\u+1)(n-\u)] + b_j(1)(1-\nu)(\u-\v)}{\ge 0}
	\addConstraint{}{}{\hspace*{-65mm}\forall u,v\in\{0,\dots,n\}\quad u+\v > n \quad u\ge v,\quad  \forall j \in \{1, \dots, m\}.}
\end{maxi}
where we define $b_j(0)=b_j(n+1)=0$. Correspondingly, $\poa(\opt{T})=1/\opt\rho$, and the optimal tolls are non-negative.\footnote{The result also holds if convexity of $b_j(x)$ is weakened to convexity of $b_j(x)x$. One example is that of $b_j(x)=\sqrt{x}$.} The result is tight for pure Nash equilibria and extends to coarse correlated equilibria.}
\end{theorem}

The optimal price of anarchy arising from the solution of \eqref{eq:fixedLPopt} for polynomials of order at most $d = 1,2,\ldots,6$ and $n=100$ are shown in the fifth column of \cref{tab:naware_vs_nagnostic}. Before proceeding with proving the theorem, we specialize its result to polynomial congestion games with $d\ge2$ and arbitrarily large $n$. This allows us to derive explicit expressions matching the values featured in \cref{tab:naware_vs_nagnostic} and holding for arbitrarily large $n$. We do not study the case of $d=1$ as this has been analyzed in \cite{caragiannis2010taxes}, resulting in an optimal price of anarchy of $1+2/\sqrt{3}\approx 2.15$, which we also recover through the solution of the linear program above.

\begin{corollary}
\label{cor:constanttollsd=2}
	Consider polynomial congestion games of maximum degree $d=2$ and arbitrarily large number of agents, i.e., congestion game where the cost on resource  $e$ is $\ell_e(x) = \alpha_e x^2+\beta_e x + \gamma_e$, with non-negative $\alpha_e,\beta_e,\gamma_e$.
An optimal congestion-independent mechanism satisfies
\be
T^{\rm opt}(\ell_e) = 3\alpha_e,\qquad\quad \poa(\opt T)=\frac{16}{3}\approx 5.33.
\label{eq:optimalconstant_d=2}
\ee
The result is tight for pure Nash equilibria and extends to coarse correlated equilibria.
\end{corollary}

\noindent Following a similar line of reasoning to that of \cref{cor:constanttollsd=2} (see the next page for its proof), it is possible to derive an expression for the optimal price of anarchy with constant tolls also in the case of $3\le d\le 6$, i.e.,
\be
\poa(\opt T)=
\frac
{\bu(\bu+1)^{d+1}-\bu^{d+1}[\bu+(\bu+2)^d]+(\bu+1)^{2d+1}-(\bu+1)^{d+1}}
{\bu(\bu+1)((\bu+1)^d-\bu^d)+(\bu+1)^{d+1}-\bu(\bu+2)^d-1},
\label{eq:analytical_expression}
\ee
where $\bar u$ is the floor of the unique real positive solution to $u^{d+1}+1=(u+1)^d + u$. For example,
\[
d=3 \quad\implies \quad\bu=2,\quad \poa(\opt T)=\frac{2\cdot 3^4-2^4\cdot(2+4^3)+3^7-3^4}
{2\cdot 3\cdot(3^3-2^3)+3^4-2\cdot4^3-1}=\frac{1212}{66}\approx 18.36.
\]
Similarly, with $d=4,\dots,6$, it is, respectively, $\poa(\opt T)=111588/1248\approx89.41$, $\poa(\opt T)=1922184/4092\approx 469.74$, $\poa(\opt T)=32963196/9912\approx 3325.58$, matching the values in \cref{tab:naware_vs_nagnostic}. 
While we do not formally prove the expression \eqref{eq:analytical_expression} in the interest of conciseness, the key idea consists in observing that the two most binding constraints appearing in the linear program of \cref{thm:congestion-indep} are those obtained with $(u,v)=(\bu,1)$ and $(u,v)=(\bu+1,1)$. Turning the corresponding inequalities into equalities and solving for $\rho$ and $\nu$ gives the result in \eqref{eq:analytical_expression}.

We now turn focus on proving \cref{thm:congestion-indep}, followed by \cref{cor:constanttollsd=2}.

\begin{proof}[Proof of \cref{thm:congestion-indep}]
The fact that optimal local congestion-independent tolls are linear in the sense that $\opt T(\ell) = \opt T(\sum_{j=1}^m \alpha_jb_j)=\sum_{j=1}^m \alpha_j \opt T(b_j)$ can be proven following the same steps of \cref{thm:main-thm}. Therefore it suffices to determine the best linear local congestion-independent toll. Towards this goal, we observe that the price of anarchy of a \emph{given} linear local constant toll $T(\sum_{j=1}^m \alpha_jb_j)=\sum_{j=1}^m \alpha_j  \tau_j$, $\tau_j\in\mb{R}_{\ge0}$ can be determined as the solution of the following program, which applies, thanks to \eqref{eq:simplifiedLP-compute}, since $f_j(x)=b_j(x)+\tau_j$ is non-decreasing
	\begin{maxi}
	{{\scriptstyle \rho \in \mathbb{R}, \nu \in \mathbb{R}_{\geq 0}}}
	{\nquad\rho} {\label{linprog:characterize_poa_nondecreasing}}{}
	\addConstraint{\nquad b_j(\v)\v - \rho b_j(\u)\u + \nu [(b_j(u)+\tau_j)\u - (b_j(u+1)+\tau_j)\v]}{\ge 0}
	\addConstraint{}{}{\hspace*{-50mm}\forall u,v\in\{0,\dots,n\}\quad u+\v \leq n,\quad  \forall j \in \{1, \dots, m\},}
	\addConstraint{\nquad b_j(\v)\v - \rho b_j(\u)\u + \nu [(b_j(u)+\tau_j)(n-\v) - (b_j(u+1)+\tau_j)(n-\u)]}{\ge 0}
	\addConstraint{}{}{\hspace*{-50mm}\forall u,v\in\{0,\dots,n\}\quad u+\v > n,\quad  \forall j \in \{1, \dots, m\}.}
	\end{maxi}
We also recall that \eqref{linprog:characterize_poa_nondecreasing} tightly characterizes the price of anarchy for pure Nash equilibria, and the corresponding bound extends to coarse correlated equilibria.
Determining the best non-negative toll amounts to letting $(\tau_1,\dots,\tau_m)\in\mb{R}^m_{\ge0}$ be decision variables, over which we need to maximize. While this would result in a bi-linear program, we define $\sigma=(\sigma_1,\dots,\sigma_m)\in\mb{R}^m_{\ge0}$ with $\sigma_j=\nu\tau_j$, and consider the following linear program
\begin{maxi}
	{{\scriptstyle \rho \in \mathbb{R},~\nu \in \mathbb{R}_{\geq 0},~\sigma \in\mathbb{R}^m_{\geq 0}}}
	{\nquad\rho} {\label{eq:LP_constanttoll_nusigma}}{}
	\addConstraint{\nquad \quad b_j(\v)\v - \rho b_j(\u)\u + \nu [b_j(\u)\u - b_j(\u+1)\v] + \sigma_j(\u - \v)}{\ge 0}
	\addConstraint{}{}{\hspace*{-50mm}\forall u,v\in\{0,\dots,n\}\quad u+\v \leq n,\quad  \forall j \in \{1, \dots, m\},}
	\addConstraint{\nquad b_j(\v)\v - \rho b_j(\u)\u + \nu [b_j(\u)(n-\v) - b_j(\u +1)(n-\u)] + \sigma_j(\u-\v)}{\ge 0}
	\addConstraint{}{}{\hspace*{-50mm}\forall u,v\in\{0,\dots,n\}\quad u+\v > n,\quad  \forall j \in \{1, \dots, m\}.}
\end{maxi}
which is an exact reformulation of \eqref{linprog:characterize_poa_nondecreasing}, except for the fact that we are not including the (non-linear) constraint requiring $\sigma_j=0$ whenever $\nu=0$. We will rectify this at the end by showing that $\opt\nu>0$.

\cref{lem:v>u_does_not_matter} in the Appendix leverages the fact that the basis functions are convex, positive, non-decreasing by assumption, so that only the constraints with $u\ge v, u\ge1$ and $(u,v)=(0,1)$ need to be accounted for in \eqref{eq:LP_constanttoll_nusigma}. Due to the fact that $u-v\ge0$ for $u\ge 1$, in order to maximize $\rho$, we choose $\sigma_j$ as large as possible.
 Observing that the only upper bound on $\sigma_j$ arises from the choice of $(u,v)=(0,1)$ and reads as $\sigma_j \le (1-\nu)b_j(1)$, we set $\sigma_j = (1-\nu)b_j(1)$, and translate the constraint $\sigma_j\ge0$ into $\nu\le 1$, thus obtaining \eqref{eq:fixedLPopt}. To conclude we are left to show that $\opt\nu$ solving \eqref{eq:fixedLPopt} is non-zero. To do so, note that solving the program for fixed $\nu=0$ results in $\rho=b(1)/b(n)$ (the tightest constraint is $(u,v)=(n,0)$), while an arbitrarily small but positive \mbox{$\nu$ would give a strictly higher $\rho$.}
  
 Once $\opt\nu$ is determined, the optimal tolls can be derived from $\opt\nu \opt\tau_j=(1-\opt\nu)b_j(1)$, recalling that $\opt\nu>0$, thus yielding \eqref{eq:optconsttoll}. Non-negativity of the tolls follow from the fact that we impose $\nu\le1$ so that $\opt{\tau}_j = \left(1/{\opt\nu}-1\right)b_j(1)\ge 0$
\end{proof}

We now focus on \cref{cor:constanttollsd=2}, and prove the result following a different approach other than directly applying \cref{thm:congestion-indep}, with the hope of providing the reader with an independent perspective.

\begin{proof}[Proof of \cref{cor:constanttollsd=2}]
We prove the claim in two steps. First, we show that the price of anarchy for any constant toll and pure Nash equilibria is lower-bounded by $16/3$. Second, we show that the price of anarchy of $T^{\rm opt}$ is upper-bounded by $16/3$ \mbox{for both Nash and coarse correlated equilibria.}

For the lower bound, it suffices to consider resource costs of the form $\ell_e(x) = \alpha_e x^2$, whereby any constant linear tolling mechanisms takes the form $T(\ell_e)=\alpha_e \tau$, for some scalar $\tau\ge 0$. For any $\tau \ge 3$ we consider the following problem instance: there are $8$ agents each with two actions $\NE{a}_i$ and $a_i^{\rm opt}$. In action $\NE{a}_i$, user $i$ selects $6$ of the available $8$ resources, which are associated to $\mygamma x^2/8 $; in $a_i^{\rm opt}$ user $i$ selects the remaining two resources with costs $\mygamma x^2/8 $, as well as one resource with cost $\mybeta x^2/8 $ (we will fix $\mygamma$ and $\mybeta$ at a later stage). Each player has a similar pair of actions, but each subsequent agent is offset by one from the prior user, as depicted in \cref{fig:fixedgame_high}.
\input{figures-tex/figureBF1}
In this game, the system and user costs can be computed as in the following $\SC(\NE{a})=\sum_{e}|a|_e \ell_e(|a|_e)=8\cdot 6 \ell_e(6)=8\cdot 6 c_1 b(6)/8$, $C_i(\NE{a})=\sum_{e\in a_i} [\ell_e(|a|_e)+\alpha_e\tau]=6\cdot (c_1 b(6)/8+c_1\tau/8)=6c_1(b(6)+\tau)/8$, and similarly
\be
\label{eq:game1_NE_SC}
\begin{split}
&\SC(\NE{a}) = 6(\mygamma b(6)) = 216 \mygamma,
\hspace*{20mm}\qquad\SC(a^{\rm opt}) = 2 \mygamma b(2) + \mybeta b(1) = 8\mygamma + \mybeta,
\\
&C_i(\NE{a}) = \frac{6}{8}\mygamma(b(6)+\tau) = \frac{1}{8}(216 \mygamma + 6\mygamma \tau),\qquad 
C_i(\opt{a}_i,\NE{a}_{-i}) = \frac{2}{8}\mygamma(49+\tau) + \frac{1}{8}\mybeta(1+\tau).
\end{split}
\ee
We normalize the costs in the game setting $\SC(\NE{a}) = 1$, which results in $\mygamma = 1/216$ from \eqref{eq:game1_NE_SC}. To ensure that the joint action $\NE{a}$ is a Nash equilibrium (at least weakly), we impose that \mbox{$C_i(\NE{a}) = C_i(\opt{a}_i,\NE{a}_{-i})$} for any player $i$. This condition is satisfied when $\mybeta = ({2\tau+59})/{(108(1+\tau))}$. 
Hence, the price of anarchy in this game is lower-bounded by $\SC(\NE{a})/\SC(a^{\rm opt}) = 1/\SC(a^{\rm opt})=1/({8\mygamma+\mybeta})$. This expression is no smaller than $16/3$ for any choice of $\tau\ge 3$ (in particular, equal when we set $\tau=3$), where we have utilized the values of $\mybeta$ from above and $\mygamma=1/216$.

For $\tau <3$ we construct a game with similar features. There are $3$ users each with actions $\NE{a}_i$ and $a_i^{\rm opt}$. In action $\NE{a}_i$, user $i$ selects $2$ of the $3$ available resources featuring a cost $\mygamma x^2/3$; in $a_i^{\rm opt}$ they select the remaining resource with cost $\mygamma x^2/3 $, as well as one resource with cost $\mybeta x^3/3 $. Each user has a similar pair of actions, but each subsequent agent is offset by one from the prior (see~\cref{fig:fixedgame_high}). In this game, we obtain the following system and user costs:
\[
\begin{split}
&\SC(\NE{a}) = 2(\mygamma b(2)) = 8 \mygamma,\qquad\hspace*{20mm}
\SC(a^{\rm opt}) =  \mygamma b(1) + \mybeta b(1) = \mygamma + \mybeta 
\\
&C_i(\NE{a}) = \frac{2}{3}\mygamma(b(2)+\tau) = \frac{1}{3}(8 \mygamma + 2\mygamma \tau),
\qquad C_i(\opt{a}_i,\NE{a}_{-i}) = \frac{\mygamma}{3}(9+\tau) + \frac{\mybeta}{3}(1+\tau).
\end{split} 
\]
As in the previous example, we set $\SC(\NE{a}) = 1$ implying $\mygamma = 1/8$. To ensure that the joint action $\NE{a}$ is a Nash equilibrium, we impose that \mbox{$C_i(\NE{a}) = C_i(\opt{a}_i,\NE{a}_{-i})$} for any player $i$, resulting in $\mybeta = ({\tau - 1})/{(8(1+\tau))}$. The resulting price of anarchy is lower-bounded by $1/\SC(a^{\rm opt})=1/(\mygamma+\mybeta)$. This quantity is no smaller than $16/3$ for any choice of $\tau<3$.

Finally, for the fixed toll $\tau = 3$, we upper-bound the price of anarchy at $16/3$. For ease of presentation, we first consider the case where the cost on resource $e$ is $\ell_e(x)=\alpha_e x^2$ for some $\alpha_e\ge0$. We will show at the end how to extend this result to the case of $\ell_e(x) = \alpha_e x^2+\beta_e x +\gamma_e$. 
Towards this goal, let $\NE{a}$ (resp. $\opt{a}$) be an equilibrium (resp. optimum) allocation in a congestion game $G$, with $n$ users, resources in $e\in\mc{E}$. The cost at equilibrium satisfies
\begingroup
\allowdisplaybreaks
\begin{align}
 \SC(\NE{a}) &\leq \sum_{i=1}^n C_i(a_i^{\rm opt},\NE{a}_{-i}) - \sum_{i=1}^n C_i(\NE{a}) + \SC(\NE{a}) \label{eq:ub_1}\\
& = \sum_{e \in \mc{E}} \alpha_e \big[ z_e((x_e +y_e+1)^2+\tau) - y_e ( (x_e+y_e)^2 + \tau ) + (x_e+y_e)^3 \big] \label{eq:ub_2} \\
&\leq \sum_{e \in \mc{E}} \alpha_e [ (x_e+z_e)((x_e+y_e+1)^2+3)- (x_e+y_e)((x_e+y_e)^2+3) + (x_e+y_e)^3 ] \label{eq:ub_4}\\
& = \sum_{e \in \mc{E}} \alpha_e \left[ 3((x_e+z_e)-(x_e+y_e))+(x_e+z_e)(x_e+y_e+1)^2 \right] \nonumber\\
& \leq \sum_{e \in \mc{E}} \alpha_e \left[4(x_e+z_e)^3+\frac{1}{4}(x_e+y_e)^3 \right]= 4 \SC(a^{\rm opt}) + \frac{1}{4}\SC(\NE{a}),\label{eq:ub_7} 
\end{align}
\endgroup
\noindent where $y_e = |\NE{a}|_e - x_e$, $z_e = |\opt a|_e - x_e$, and $x_e = |\{ i \in N \text{ s.t. } e \in \NE{a}_i \cap \opt a_i \}|$. Observe that \eqref{eq:ub_1} holds from the definition of Nash equilibrium, while \eqref{eq:ub_2} follows from the parameterization introduced in \cite{paccagnan2018distributed}, and substituting $b_e(x) = x^2$. Equation \eqref{eq:ub_4} follows by replacing $\tau = 3$ and by $x_e \geq 0$. To see that \eqref{eq:ub_7} holds for all integers $x_e,y_e \geq 0$, we define $\u=x_e+y_e\ge 0$, $\v=x_e+z_e\ge 0$, and divide the argument in two parts depending on whether the integer tuple $(\u,\v) \in \{ \u \geq 22 ~ {\rm or} ~ \v \geq 8 \}$ or not. For the case of $(\u,\v) \in \{ \u \geq 22 ~ {\rm or}~ \v \geq 8 \}$, we observe that 
$
	4\v^3  + \frac{1}{4}\u^3 - 3\v +3\u-\v(\u+1)^2
	\ge
	4\v^3 + \frac{1}{4}\u^3 - \v(\u^2+2\u+4)
	\ge
	4\v^3 + \frac{1}{4}\u^3 - \v(\u+2)^2,
$
and therefore we are left to prove
\be
4\v^3 + \frac{1}{4}\u^3 - \v(\u+2)^2\ge0.
\label{eq:ub_6_stricter}
\ee
For every fixed $\u\ge0$, differentiating with respect to $\v$ shows that the left hand side of \eqref{eq:ub_6_stricter} has a unique global minimum in the positive orthant at $\v = (\u+2)/{\sqrt{12}}$. For any $\u > 22$, this minimum satisfies \eqref{eq:ub_6_stricter}, thus for any $\v\geq 0$ and $\u > 22$ \eqref{eq:ub_6_stricter} is satisfied.
Additionally observe that, when $\v= 8$,  \eqref{eq:ub_6_stricter} is satisfied for each $\u \in \{0, \ldots , 22\}$. Further, for fixed $0\le\u\le 22$, the left hand side of \eqref{eq:ub_6_stricter} is increasing in $\v$ for $\v\ge8$. This implies that \eqref{eq:ub_6_stricter} holds for every $\v\geq 8$ as well. Therefore \eqref{eq:ub_6_stricter} (and consequently \eqref{eq:ub_7}) is satisfied for all $(\u,\v) \in \{ \u \geq 22 ~ {\rm or} ~ \v \geq 8 \}$. One can enumerate the finitely-many non-negative integers $(\u,\v)$ with $\u<22$, $\v<8$ and verify that \eqref{eq:ub_7} holds.

The inequality in \eqref{eq:ub_7} implies that the price of anarchy is upper bounded by $\frac{4}{1-1/4} = 16/3$ when resource costs take the form $\ell_e(x)=\alpha_e x^2$ for some $\alpha_e\ge0$. Observe that this bound holds for arbitrarily large $n$ and matches the solution of the linear program, stated in \cref{tab:naware_vs_nagnostic}. We now generalize this result to $\ell_e(x) = \alpha_e x^2+\beta_e x +\gamma_e$, where $\alpha_e$, $\beta_e$ and $\gamma_e$ are non-negative. To do so, we start from \eqref{eq:ub_1}, and note that \eqref{eq:ub_2} now contains the sum of three contributions: contributions relative to $\alpha_e$, contributions relative to $\beta_e$ and contributions relative to $\gamma_e$. Hence, it suffices to prove the following two additional inequalities to complete the reasoning, that is
\begin{align}
\sum_{e \in \mc{E}} \beta_e \left[ z_e(x_e+y_e+1) - y_e (x_e+y_e) + (x_e+y_e)^2 \right] 
&\le \sum_{e \in \mc{E}} \beta_e \left[4(x_e+z_e)^2+\frac{1}{4}(x_e+y_e)^2 \right],\label{eq:prooflinear}
\\
\sum_{e \in \mc{E}} \gamma_e \left[ z_e - y_e  + x_e+y_e \right] 
&\le\sum_{e \in \mc{E}} \gamma_e \left[4(x_e+z_e)+\frac{1}{4}(x_e+y_e) \right],
\label{eq:proofcostant}
\end{align}
where we recall that no toll is associated to the presence of $\beta_e$ or $\gamma_e$ in \eqref{eq:optimalconstant_d=2}. Summing these two inequalities with the inequality from \eqref{eq:ub_2} and \eqref{eq:ub_7} will, in fact, yield the desired claim. While the proof of \eqref{eq:proofcostant} is immediate, the argument used to show \eqref{eq:prooflinear} is similar to that following \eqref{eq:ub_2}. In fact, since $x_e\ge0$, we have $z_e(x_e +y_e+1) - y_e (x_e+y_e) + (x_e+y_e)^2 \le(x_e+z_e)(x_e+y_e+1) - (x_e+y_e)^2+ (x_e+y_e)^2$. Thus, we are left to show that for every non-negative integer $u$ and $v$ it is $4v^2+\frac{1}{4}u^2-v-uv\ge 0$, where we make use of the same coordinates introduced earlier. This inequality is satisfied by all non-negative integer points since $4v^2+\frac{1}{4}u^2-v-uv\ge 0$ describes the region outside an ellipse located in the $(u,v)$ plane entirely on the left of the line $u=1$ and entirely south of the line $v=1$, where the inequality is trivially satisfied for $u=v=0$. Finally, we observe that the technique used to bound the price of anarchy extends to coarse correlated equilibria due \mbox{to linearity of the expectation.}
\end{proof}

%% file: figures-tex/figureBF1.tex
\newcommand\radius{2.4}
\begin{figure}[b!]
\begin{center}
\hspace*{-5mm}
\raisebox{-0.55\height}{\begin{tikzpicture}[scale=0.45, transform shape]%
\begin{scope}
\clip (0,0)-- +(90:\radius + 2) arc (90:225-360:\radius + 2) -- cycle;
\draw[red,thick,fill=red,fill opacity=0.2] (0,0) circle (\radius + 0.75);
\draw[red,thick,fill=white] (0,0) circle (\radius - 0.75);
\end{scope}%
\begin{scope}
\clip (0,0)-- +(90:\radius + 2) arc (90:135:\radius + 2) -- cycle;
\draw[red,thick,fill=red,fill opacity=0.2] (\fpeval{\radius*cos(90*pi/180)},\fpeval{\radius*sin(90*pi/180)}) circle(0.75);
\end{scope}%
\begin{scope}
\clip (0,0)-- +(225:\radius + 2) arc (225:180:\radius + 2) -- cycle;
\draw[red,thick,fill=red,fill opacity=0.2] (\fpeval{\radius*cos(225*pi/180)},\fpeval{\radius*sin(225*pi/180)}) circle(0.75);
\end{scope}%
\begin{scope}
\clip (0,0)-- +(180:\radius + 2) arc (180:135:\radius + 2) -- cycle;
\draw[blue,thick,fill=blue,fill opacity=0.2] (0,0) circle (\radius + 0.75);
\draw[blue,thick,fill=white] (0,0) circle (\radius - 0.75);
\end{scope}%
\begin{scope}
\clip (0,0)-- +(180:\radius + 2) arc (180:225:\radius + 2) -- cycle;
\draw[blue,thick,fill=blue,fill opacity=0.2] (\fpeval{\radius*cos(180*pi/180)},\fpeval{\radius*sin(180*pi/180)}) circle(0.75);
\end{scope}%
\begin{scope}
\clip (0,0)-- +(135:\radius + 2) arc (135:90:\radius + 2) -- cycle;
\draw[blue,thick,fill=blue,fill opacity=0.2] (\fpeval{\radius*cos(135*pi/180)},\fpeval{\radius*sin(135*pi/180)}) circle(0.75);
\end{scope}%
\begin{scope}
\clip (0,0)-- +(45:\radius + 2) arc (45:180-360:\radius + 2) -- cycle;
\draw[red,thick,dashed,fill=red,fill opacity=0.2] (0,0) circle (\radius + 0.9);
\draw[red,thick,dashed,fill=white] (0,0) circle (\radius - 0.9);
\end{scope}%
\begin{scope}
\clip (0,0)-- +(45:\radius + 2) arc (45:90:\radius + 2) -- cycle;
\draw[red,thick,dashed,fill=red,fill opacity=0.2] (\fpeval{\radius*cos(45*pi/180)},\fpeval{\radius*sin(45*pi/180)}) circle(0.9);
\end{scope}%
\begin{scope}
\clip (0,0)-- +(180:\radius + 2) arc (180:135:\radius + 2) -- cycle;
\draw[red,thick,dashed,fill=red,fill opacity=0.2] (\fpeval{\radius*cos(180*pi/180)},\fpeval{\radius*sin(180*pi/180)}) circle(0.9);
\end{scope}%
\begin{scope}
\clip (0,0)-- +(135:\radius + 2) arc (135:90:\radius + 2) -- cycle;
\draw[blue,thick,dashed,fill=blue,fill opacity=0.2] (0,0) circle (\radius + 0.9);
\draw[blue,thick,dashed,fill=white] (0,0) circle (\radius - 0.9);
\end{scope}%
\begin{scope}
\clip (0,0)-- +(135:\radius + 2) arc (135:180:\radius + 2) -- cycle;
\draw[blue,thick,dashed,fill=blue,fill opacity=0.2] (\fpeval{\radius*cos(135*pi/180)},\fpeval{\radius*sin(135*pi/180)}) circle(0.9);
\end{scope}%
\begin{scope}
\clip (0,0)-- +(90:\radius + 2) arc (90:45:\radius + 2) -- cycle;
\draw[blue,thick,dashed,fill=blue,fill opacity=0.2] (\fpeval{\radius*cos(90*pi/180)},\fpeval{\radius*sin(90*pi/180)}) circle(0.9);
\end{scope}%
\fill[fill = white] (0,0) circle (\radius - 0.91);%
\foreach \n in {0,1,2,3,4,5,6,7}{
	\draw (\fpeval{\radius*cos(\n*pi/4)},\fpeval{\radius*sin(\n*pi/4)}) node[circle,draw]{$\mygamma/8$};}%
\node[red] at (\fpeval{(\radius+1.25)*cos(3*pi/8)},\fpeval{(\radius+1.25)*sin(3*pi/8)}) {\Huge$\NE{a}_1$};%
\node[blue] at (\fpeval{(\radius+1.25)*cos(7*pi/8)},\fpeval{(\radius+1.25)*sin(7*pi/8)}) {\Huge$\opt{a}_1~$};%
\node[red] at (\fpeval{(\radius+1.5)*cos(1*pi/8)},\fpeval{(\radius+1.5)*sin(1*pi/8)}) {\Huge$\NE{a}_2$};%
\node[blue] at (\fpeval{(\radius+1.5)*cos(5*pi/8)},\fpeval{(\radius+1.5)*sin(5*pi/8)}) {\Huge$\opt{a}_2$};%
\end{tikzpicture}}%%
\hspace*{-5mm}
\raisebox{-0.5\height}{
\begin{tikzpicture}[scale=0.45, transform shape]
\draw[blue,thick,fill=blue,fill opacity=0.2] (\fpeval{\radius*cos(180*pi/180)},\fpeval{\radius*sin(180*pi/180)}) circle(0.75);
\draw[blue,thick,dashed,fill=blue,fill opacity=0.2] (\fpeval{\radius*cos(135*pi/180)},\fpeval{\radius*sin(135*pi/180)}) circle(0.9);
\foreach \n in {0,1,2,3,4,5,6,7}{
	\draw (\fpeval{\radius*cos(\n*pi/4)},\fpeval{\radius*sin(\n*pi/4)}) node[circle,draw]{$\mybeta/8$};}
\node[blue] at (\fpeval{(\radius+1.25)*cos(8*pi/8)},\fpeval{(\radius+1.25)*sin(8*pi/8)}) {\Huge$\opt{a}_1~$};
\node[blue] at (\fpeval{(\radius+1.5)*cos(6*pi/8)},\fpeval{(\radius+1.5)*sin(6*pi/8)}) {\Huge$\opt{a}_2$};
\end{tikzpicture}}
\hspace*{2mm}
\raisebox{-0.5\height}{
\begin{tikzpicture}
\draw (0,-2.2) -- (0,1.7);	
\end{tikzpicture}}
\hspace*{-7mm}
\raisebox{-0.5\height}{
\begin{tikzpicture}[scale=0.45, transform shape]
\begin{scope}
\clip (0,0)-- +(120:\radius + 2) arc (120:0:\radius + 2) -- cycle;
\draw[red,thick,fill=red,fill opacity=0.2] (0,0) circle (\radius + 0.75);
\draw[red,thick,fill=white] (0,0) circle (\radius - 0.75);
\end{scope}

\begin{scope}
\clip (0,0)-- +(120:\radius + 2) arc (120:180:\radius + 2) -- cycle;
\draw[red,thick,fill=red,fill opacity=0.2] (\fpeval{\radius*cos(120*pi/180)},\fpeval{\radius*sin(120*pi/180)}) circle(0.75);
\end{scope}

\begin{scope}
\clip (0,0)-- +(0:\radius + 2) arc (0:-90:\radius + 2) -- cycle;
\draw[red,thick,fill=red,fill opacity=0.2] (\fpeval{\radius*cos(0*pi/180)},\fpeval{\radius*sin(0*pi/180)}) circle(0.75);
\end{scope}
\draw[blue,thick,fill=blue,fill opacity=0.2] (\fpeval{\radius*cos(240*pi/180)},\fpeval{\radius*sin(240*pi/180)}) circle(0.75);
\begin{scope}
\clip (0,0)-- +(0:\radius + 2) arc (0:-120:\radius + 2) -- cycle;
\draw[red,thick,dashed,fill=red,fill opacity=0.2] (0,0) circle (\radius + 0.9);
\draw[red,thick,dashed,fill=white] (0,0) circle (\radius - 0.9);
\end{scope}

\begin{scope}
\clip (0,0)-- +(0:\radius + 2) arc (0:45:\radius + 2) -- cycle;
\draw[red,thick,dashed,fill=red,fill opacity=0.2] (\fpeval{\radius*cos(0*pi/180)},\fpeval{\radius*sin(0*pi/180)}) circle(0.9);
\end{scope}

\begin{scope}
\clip (0,0)-- +(-120:\radius + 2) arc (-120:-180:\radius + 2) -- cycle;
\draw[red,thick,dashed,fill=red,fill opacity=0.2] (\fpeval{\radius*cos(-120*pi/180)},\fpeval{\radius*sin(-120*pi/180)}) circle(0.9);
\end{scope}

\draw[blue,thick,dashed,fill=blue,fill opacity=0.2] (\fpeval{\radius*cos(120*pi/180)},\fpeval{\radius*sin(120*pi/180)}) circle(0.9);

\fill[fill = white] (0,0) circle (\radius - 0.91);
\foreach \n in {0,1,2}{
	\draw (\fpeval{\radius*cos(\n*2*pi/3)},\fpeval{\radius*sin(\n*2*pi/3)}) node[circle,draw]{$\mygamma/3$};
}

\node[red] at (\fpeval{(\radius+1.25)*cos(2*pi/8)},\fpeval{(\radius+1.25)*sin(2*pi/8)}) {\Huge$\NE{a}_1$};

\node[blue] at (\fpeval{(\radius+1.5)*cos(4*pi/3)},\fpeval{(\radius+1.5)*sin(4*pi/3)}) {\Huge$\opt{a}_1$};

\node[red] at (\fpeval{(\radius+1.5)*cos(-2*pi/8)},\fpeval{(\radius+1.5)*sin(-2*pi/8)}) {\Huge$\NE{a}_2$};

\node[blue] at (\fpeval{(\radius+1.5)*cos(2*pi/3)},\fpeval{(\radius+1.5)*sin(2*pi/3)}) {\Huge$\opt{a}_2$};
\end{tikzpicture}
}
\raisebox{-0.55\height}{
\begin{tikzpicture}[scale=0.45, transform shape]
\draw[blue,thick,fill=blue,fill opacity=0.2] (\fpeval{\radius*cos(240*pi/180)},\fpeval{\radius*sin(240*pi/180)}) circle(0.75);
\draw[blue,thick,dashed,fill=blue,fill opacity=0.2] (\fpeval{\radius*cos(120*pi/180)},\fpeval{\radius*sin(120*pi/180)}) circle(0.9);

\fill[fill = white] (0,0) circle (\radius - 0.91);

\foreach \n in {0,1,2}{
	\draw (\fpeval{\radius*cos(\n*2*pi/3)},\fpeval{\radius*sin(\n*2*pi/3)}) node[circle,draw]{$\mybeta/3$};
}

\node[blue] at (\fpeval{(\radius+1.5)*cos(4*pi/3)},\fpeval{(\radius+1.5)*sin(4*pi/3)}) {\Huge$\opt{a}_1$};

\node[blue] at (\fpeval{(\radius+1.5)*cos(2*pi/3)},\fpeval{(\radius+1.5)*sin(2*pi/3)}) {\Huge$\opt{a}_2$};
\end{tikzpicture}
}
\end{center}
\vspace*{-7mm}
\caption{Game construction used to lower bound the price of anarchy for quadratic congestion games with fixed tolls $\tau\ge 3$ (left) and $\tau< 3$ (right). On the left (resp. right), the available actions of two of the eight (resp. three) agents are shown.  The solid red shape contains the resources utilized by the first user in the action $\NE{a}_1$, while the solid blue shape contains the resources utilized by the first user in the action $\opt{a}_1$. User 2 has similar actions but rotated clockwise on each circle by one resource. Each of the remaining agents' actions are defined similarly by rotating about the apparent ``ring''.}
\label{fig:fixedgame_high}
\end{figure}
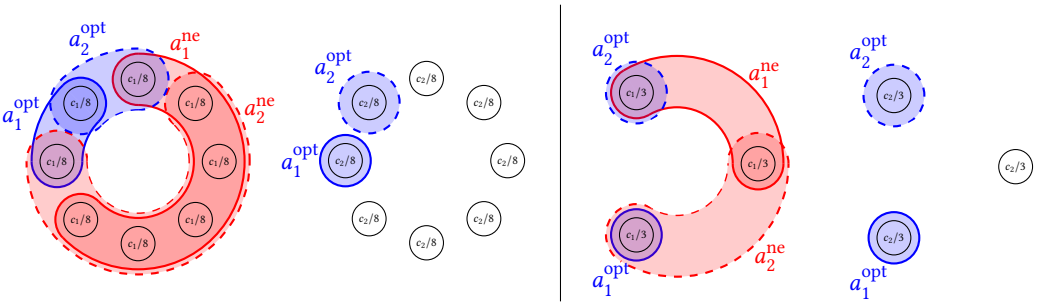

%% file: parts-supplementary/supplementary-marginalcost.tex
\section{(In)efficiency of the marginal cost mechanism}
\label{sec:pigouvtolls}

In this section we study the efficiency of the marginal cost mechanism, whereby the toll imposed to each user corresponds to her marginal contribution to the system cost. In the atomic setup, the marginal cost mechanism takes the form $\pig{T}(\ell)=\pig{\tau}$, and the corresponding tolls read $\pig{\tau}(x)=(x-1)(\ell(x)-\ell(x-1))$, where we set $\ell(0)=0$. We recall that $\pig{T}$ ensures that the best performing equilibrium is optimal, i.e., its price of stability is one \cite{marden2013distributed}. The following Corollary shows how to compute $\poa(\pig{T})$ through the solution of a linear program when bases are discrete convex (see \cref{foot:cvx}). We also provide the analytical expression of $\poa(\pig{T})$ for polynomial congestion games with $d=1$, and note that a similar argument carries over when~$d\ge1$. 

\begin{corollary}
\label{cor:marginalcost}
The price of anarchy of the marginal cost mechanism $\pig{T}(\ell)=\pig{\tau}$, with $\pig{\tau}(x)=(x-1)(\ell(x)-\ell(x-1))$ over congestion games with $n$ agents, and resource costs generated by a non-negative linear combination of convex basis functions $\{b_1,\dots,b_m\}$ equals $1/\opt{\rho}$, where $\opt{\rho}$ solves the following linear program
\begin{maxi*}
	{{\scriptstyle \rho \in \mathbb{R}, \nu \in \mathbb{R}_{\geq 0}}}
	{\rho} {}{}
	\addConstraint{}{}{}
	\addConstraint{\!\!\!\!\!\!\!\!\!\!\!\!\!\!\!\!\!\!\!\!\!\!\!\!\!\!\!b_j(\v)\v - \rho b_j(\u)\u + \nu 
    [(u^2-uv)b_j(u)-u(u-1)b_j(u-1)-v(u+1)b_j(u+1)]}{\ge 0}
	\addConstraint{}{}{\hspace*{-80mm}\forall u,v\in\{0,\dots,n\}\quad u+\v \leq n,\quad  \forall j \in \{1, \dots, m\},}
	\addConstraint{\!\!\!\!\!\!\!\!\!\!\!\!\!\!\!\!\!\!\!\!\!\!\!\!\!\!\!b_j(\v)\v - \rho b_j(\u)\u + \nu 
    [ub_j(u)(2n-u-v)+(u-1)b_j(u-1)(v-n)+(u+1)b_j(u+1)(u-n)]}{\ge 0}
	\addConstraint{}{}{\hspace*{-80mm}\forall u,v\in\{0,\dots,n\}\quad u+\v > n,\quad  \forall j \in \{1, \dots, m\}.}
\end{maxi*}
where we set $b_j(-1)=b_j(0)=b_j(n+1)=0$.\footnote{The result also holds under the weaker requirement that only $b_j(x)x$ are convex.} 

For affine congestion games with arbitrarily large number of agents, we have $\poa(\pig{T})=3$. Both results are tight for pure Nash equilibria, and also hold for coarse correlated equilibria.
\end{corollary}
\begin{proof}
We begin with the first claim, and observe that the marginal cost mechanism is linear, in the sense that $\pig{T}(\sum_{j=1}^m \alpha_j b_j) = \sum_{j=1}^m \alpha_j\pig{T}(b_j)$ for all non-negative $\alpha_j$ and for all basis functions. Additionally, the functions $f_j(x)=b_j(x)+(x-1)(b_j(x)-b_j(x-1))$ are non-decreasing in their domain. This is because
\[
\begin{split}
f_j(x+1)- f_j(x)
&=
b_j(x+1)+x(b_j(x+1)-b_j(x))-b_j(x)-(x-1)(b_j(x)-b_j(x-1))\\
&=(x+1)b_j(x+1)-xb_j(x)-(xb_j(x)-(x-1)b_j(x-1))\ge0,
\end{split}
\]
for all $x\in\{1,\dots,n-1\}$, where the inequality holds as each function $b_j(x)x$ is convex (since each $b_j(x)$ is so), and thus its discrete derivative is non-decreasing. It follows that the price of anarchy can be computed using the linear program in \eqref{eq:simplifiedLP-compute} which provides tight results for pure Nash equilibria that extend to coarse correlated equilibria. Substituting $f_j(x)=b_j(x)+(x-1)(b_j(x)-b_j(x-1))$ in \eqref{eq:simplifiedLP-compute} we obtain the desired result.

We now focus on affine congestion games, and prove that $\poa(\pig{T})=3$. Towards this goal, we observe that \cref{example:tolledpigouvian} is an example of an affine congestion game using the marginal cost mechanism. Thus, we conclude that $\poa(\pig\taxmechanism) \geq 3$ for affine \mbox{congestion games and pure Nash equilibria.}
\input{figures-tex/figure3.tex}

We now show that $\poa(\pig\taxmechanism) \leq 3$ whenever each resource $e$ is associated to a cost $\ell_e(x)=\alpha_e x + \beta_e$. In this case, we have $\pig\taxmechanism(\ell_e)=\pig\tau$, where $\pig\tau(x)=\alpha_e(x-1)$ is independent of $\beta_e$, thanks to its definition. For an equilibrium allocation $\NE{a} \in \actionset$ and $\opt{a} \in \actionset$ and optimal allocation, we have
\[
\begin{split}
     \SC(\NE{a}) &\leq \sum_{i=1}^n \C_i(\opt{a}_i, \NE{a}_{-i}) - \sum_{i=1}^n \C_i(\NE{a}) + \SC(\NE{a}) \\
    & = \sum_{e \in \resset} \alpha_e \! \left[ z_e (2 x_e + 2 y_e + 1) -y_e (2 x_e + 2 y_e - 1) + (x_e+y_e)^2 \right]  + \beta_e \! \left[ z_e - y_e + (x_e + y_e) \right] \\
    & \leq \sum_{e \in \resset} \alpha_e \! \left[ (x_e +y_e) - (x_e + y_e)^2   + 2(x_e + y_e)(x_e + z_e)  +  (x_e  +  z_e) \right] + \beta_e \! \left[ x_e  +  z_e \right]\\
    & \leq \sum_{e \in \resset} \alpha_e \! \left[ 3 (x_e+z_e)^2 \right] + \beta_e \left[ 3(x_e + z_e) \right] = 3 \cdot \SC(\opt{a}),
\end{split}
\]
where we utilize the notation $y_e = |\NE{a}|_e - x_e$, $z_e = |\opt a|_e - x_e$, and $x_e = |\{ i \in N \text{ s.t. } e \in \NE{a}_i \cap \opt a_i \}|$. The first inequality holds by definition of Nash equilibrium, and the second holds due to non-negativity of $x_e$ and $\alpha_e$. One verifies that the last inequality holds, using $\u=x_e+y_e\ge 0$, $\v=x_e+z_e\ge 0$, and observing that the region $ 3\v^2 + \u^2 -\u -2\u\v -\v \geq 0 $ is the exterior of an ellipse containing all $(\u,\v) \in \mathbb{Z}_{\ge0}^2$. We remark that the latter reasoning extends identical to coarse correlated equilibria exploiting linearity of the expectation. Rearranging the above inequality, we get $\poa(\pig\taxmechanism) \leq \SC(\NE{a}) / \SC(\opt{a}) \leq 3$ for pure Nash as well as coarse correlated equilibria.
\end{proof}

%% file: figures-tex/figure3.tex
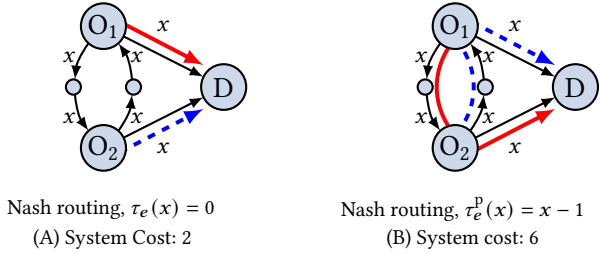
\begin{figure}[h!]
\centering
\begin{tikzpicture}
\newlength{\myedgewidth}
\setlength{\myedgewidth}{0.8pt}
  \Vertex[x=1.4, y=0.8,  color=pnasblueback, label={\large$\rm{O}_1$}, style={line width=\myedgewidth}]{C}
  \Vertex[x=1.4, y=-0.8,  color=pnasblueback, label={\large$\rm{O}_2$}, style={line width=\myedgewidth}]{D}
  \Vertex[x=1, y=0,  color=pnasblueback, size=0.2, style={line width=\myedgewidth}]{betweenCDleft}
  \Vertex[x=1.8, y=0,  color=pnasblueback, size=0.2, style={line width=\myedgewidth}]{betweenCDright}
  \Vertex[x=3, y=0,  color=pnasblueback, label={\large\rm{D}}, style={line width=\myedgewidth}]{E}
  \Edge[Direct, bend=-15, lw=\myedgewidth, color=black](C)(betweenCDleft)
  \Edge[Direct, bend=-15, lw=\myedgewidth, color=black](betweenCDright)(C)
  \Edge[Direct, bend=-15, lw=\myedgewidth, color=black](betweenCDleft)(D)
  \Edge[Direct, bend=-15, lw=\myedgewidth, color=black](D)(betweenCDright)
  \Edge[Direct, lw=\myedgewidth, color=black](C)(E)
  \Edge[Direct, lw=\myedgewidth, color=black](D)(E)
  \Text[x=2.2, y=0.8, fontsize=\small]{$x$}
  \Text[x=2.2, y=-0.8, fontsize=\small]{$x$}
  \Text[x=0.95, y=0.4, fontsize=\small]{$x$}
  \Text[x=1.85, y=0.4, fontsize=\small]{$x$}
  \Text[x=0.95, y=-0.4, fontsize=\small]{$x$}
  \Text[x=1.85, y=-0.4, fontsize=\small]{$x$}
  \Text[x=1.5, y=-1.6, fontsize=\footnotesize]{Nash routing, $\tau_e(x)=0$}
  \Text[x=1.52, y=-2, fontsize=\footnotesize]{(A)~~System Cost: 2}%
  \Vertex[x=1.34, y=1, Pseudo, size=0]{N2bid}
  \Vertex[x=3, y=0.15, Pseudo]{N3}
  \Edge[color=red, Direct](N2bid)(N3)
  \Vertex[x=1.34, y=-1, Pseudo, size=0]{bisN2bid}
  \Vertex[x=3, y=-0.15, Pseudo]{bisN3}
  \Edge[color=blue, Direct, style={dashed}](bisN2bid)(bisN3)  
\end{tikzpicture}
\hspace*{10mm}
\begin{tikzpicture}
\setlength{\myedgewidth}{0.8pt}
  \Vertex[x=1.4, y=0.8,  color=pnasblueback, label={\large$\rm{O}_1$}, style={line width=\myedgewidth}]{C}
  \Vertex[x=1.4, y=-0.8,  color=pnasblueback, label={\large$\rm{O}_2$}, style={line width=\myedgewidth}]{D}
  \Vertex[x=1, y=0,  color=pnasblueback, size=0.2, style={line width=\myedgewidth}]{betweenCDleft}
  \Vertex[x=1.8, y=0,  color=pnasblueback, size=0.2, style={line width=\myedgewidth}]{betweenCDright}
  \Vertex[x=3, y=0,  color=pnasblueback, label={\large\rm{D}}, style={line width=\myedgewidth}]{E}
  \Edge[Direct, bend=-15, lw=\myedgewidth, color=black](C)(betweenCDleft)
  \Edge[Direct, bend=-15, lw=\myedgewidth, color=black](betweenCDright)(C)
  \Edge[Direct, bend=-15, lw=\myedgewidth, color=black](betweenCDleft)(D)
  \Edge[Direct, bend=-15, lw=\myedgewidth, color=black](D)(betweenCDright)
  \Edge[Direct, lw=\myedgewidth, color=black](C)(E)
  \Edge[Direct, lw=\myedgewidth, color=black](D)(E)
  \Text[x=2.2, y=0.8, fontsize=\small]{$x$}
  \Text[x=2.2, y=-0.8, fontsize=\small]{$x$}
  \Text[x=0.95, y=0.4, fontsize=\small]{$x$}
  \Text[x=1.85, y=0.4, fontsize=\small]{$x$}
  \Text[x=0.95, y=-0.4, fontsize=\small]{$x$}
  \Text[x=1.85, y=-0.4, fontsize=\small]{$x$}
  \Text[x=1.5, y=-1.6, fontsize=\footnotesize]{Nash routing, $\pig{\tau}_e(x)=x-1$}
  \Text[x=1.5, y=-2, fontsize=\footnotesize]{(B)~~System cost: 6}
  \Vertex[x=1.34, y=1, Pseudo, size=0]{N2bid}
  \Vertex[x=3, y=0.15, Pseudo]{N3}
  \Vertex[x=1.22, y=0.85, Pseudo]{N4}
  \Vertex[x=1.22, y=-0.85, Pseudo]{N5}
  \Edge[color=blue, Direct, style={dashed}](N2bid)(N3)
  \Edge[color=blue, style={dashed}, bend=-40](N5)(N4)
  \Vertex[x=1.34, y=-1, Pseudo, size=0]{bisN2bid}
  \Vertex[x=3, y=-0.15, Pseudo]{bisN3}
  \Vertex[x=1.58, y=-0.85, Pseudo]{bisN4}
  \Vertex[x=1.58, y=0.85, Pseudo]{bisN5}
  \Edge[color=red, Direct](bisN2bid)(bisN3)
  \Edge[color=red, bend=-40](bisN5)(bisN4)
\end{tikzpicture}
\vspace*{-2mm}
\caption{
Instance used to demonstrates that the price of anarchy associated to the marginal cost toll mechanism is at least $3$ in affine congestion games. Two users are willing to travel from $\rm{O}_1/\rm{O}_2$ to $\rm{D}$, where each edge features a cost $\ell_e(x)=x$. In the un-tolled case (A) the system cost at the worst Nash-equilibrium is $2$. The situation worsens when using marginal cost tolls (B), as the worst Nash equilibrium gives a system cost of $6$.}
\vspace*{-2mm}
\label{example:tolledpigouvian}
\end{figure}

%% file: parts-supplementary/conclusions.tex
\section{Conclusions and open problems}
This work derives optimal local tolling mechanisms and corresponding prices of anarchy for atomic congestion games. We do so for both the setup where tolls are congestion-aware and congestion-independent. Finally, we derive price of anarchy values for the marginal cost mechanism. Our results generalize those of \cite{caragiannis2010taxes}, and show that the efficiency of optimal tolls utilizing solely local information is comparable to that of existing tolls using global information \cite{BiloV19}. Further, we show that utilizing the marginal cost mechanism on the atomic setup is worse than levying no toll.

\paragraph{Open questions}
Our work leaves a number of open questions, two of which are discussed next. 

\begin{itemize}
\item[-] While we observed that the price of anarchy for optimally tolled affine congestion games matches that of affine load balancing games on identical machines, we conjecture such result holds more generally, at least for polynomial congestion games.
\item[-] In this manuscript we focused on the worst-case efficiency metric both with respect to the game instance, and the resulting equilibrium. It is currently unclear if, and to what extent, optimizing the price of anarchy impacts other more optimistic performance metrics.
\end{itemize}

%% file: parts-appendix/appendix-optimal-toll.tex
\begin{lemma} \label{lem:resource_decomposition}
    Consider the class of congestion games $\cg$. For any linear tolling mechanism $\taxmechanism$, it is
\[
\poa(\taxmechanism) = \sup_{G \in \cg(\mathbb{Z}_{\ge0})} \frac{\PNEcost(G,\taxmechanism)}{\mincost(G)},
\]
where $\cg(\mathbb{Z}_{\ge0}) \subset \cg$ is the subclass of games with $\alpha_j \in \mathbb{Z}_{\ge0}$ for all $j \in \{1, \dots, m\}$, \mbox{for all resources in $\resset$.}
\end{lemma}
\begin{proof}
We divide the proof in two steps. First, we show that 
\be
\poa(\taxmechanism) = \sup_{G \in \cg(\mathbb{Q}_{\ge0})} \frac{\PNEcost(G,\taxmechanism)}{\mincost(G)},
\label{eq:equalityQ}
\ee 
where $\cg(\mathbb{Q}_{\ge0}) \subset \cg$ is the subclass of games with $\alpha_j \in \mathbb{Q}_{\ge0}$ for all $j \in \{1, \dots, m\}$, for all resources in $\resset$.  
Towards this goal, observe that \eqref{eq:equalityQ} holds trivially with $\ge$ in place of the equality sign, as $\mathbb{R}_{\ge0}\supset\mathbb{Q}_{\ge0}$. To show that the converse inequality also holds, observe that the price of anarchy of a given linear mechanisms $T$ (computed over all meaningful instances where $\PNEcost(G,T)>0$) can be computed utilizing the linear program reported in \eqref{linprog:characterize_poa}. By strong duality, we have $\poa(T)=1/C^\star$, where $C^\star$ is the value of the dual program of \eqref{linprog:characterize_poa}, i.e., 
\begin{mini!}
	{\scriptstyle \theta(x,y,z,j)}{\nquad \sum_{x,y,z,j} b_j(x+z)(x+z)\theta(x,y,z,j)} {}{C^\star =}
	\addConstraint{\nquad \sum_{x,y,z,j} \left[f_j(x+y)y-f_j(x+y+1)z\right]\theta(x,y,z,j)}{\le 0 \label{eq:nash-inequality}}
	\addConstraint{\nquad\sum_{x,y,z,j} b_j(x+y)(x+y)\theta(x,y,z,j)}{=1}
	\addConstraint{\nquad \theta(x,y,z,j)}{\ge0 \label{eq:positivity}}{\quad \forall (x,y,z,j) \in \mc{I}}
\end{mini!}
where we define $b_j(0)=f_j(0)=f_j(n+1)=0$ for convenience, $\mc{I}=\{(x,y,z,j)\in\mathbb{Z}_{\ge0}^4~\text{s.t.}~1\le x+y+z\le n,~1\le j\le m\}$, and the minimum is intended over the entire tuple $\{\theta(x,y,z,j)\}_{(x,y,z,j)\in\mc{I}}$. Let $\{\theta^\star(x,y,z,j)\}_{(x,y,z,j)\in\mc{I}}$ denote an optimal solution (which exists, due to the non-emptiness and boundedness of the constraint set, which can be proven using the same argument in \cite[Thm. 2]{paccagnan2018distributed}). If all $\theta^\star(x,y,z,j)$ are rational, then consider the game $G$ defined as follows. For every $i\in \{1,\dots,n\}$ and for every $(x,y,z,j)\in\mc{I}$, we create a resource identified with $e(x,y,z,j,i)$, and assign to it the resource cost $\alpha_jb_j$, where $\alpha_j=\theta^\star(x,y,z,j)/n$. The game $G$ features $n$ players, where player $p\in\{1,\dots,n\}$ can either select the resources in the allocation $a^{\rm opt}_p$ or in $a^{\rm ne}_p$, defined by
\[
\begin{split}
a^{\rm opt}_p &= \cup_{i=1}^n\cup_{j=1}^m
\{e(x,y,z,j,i)\,:~x+y \ge 1 + ((i-p)\text{\,mod\,}n)\},\\
a^{\rm ne}_p &= \cup_{i=1}^n\cup_{j=1}^m
\{e(x,y,z,j,i)\,:~x+z \ge 1 + ((i-p+z)\text{\,mod\,}n)\}.\\
\end{split}
\]
Note that the above construction is an extension of that appearing in \cite{paccagnan2018distributed} to the case of multiple basis functions. Since $G$ has 
\[
\begin{split}
\PNEcost(G,T)&=\sum_{x,y,z,j} b_j(x+y)(x+y)\theta^\star(x,y,z,j)=1,\\
\mincost(G)&\le\sum_{x,y,z,j} b_j(x+z)(x+z)\theta^\star(x,y,z,j)=C^\star,
\end{split}
\]
(see \cite[Thm. 2]{paccagnan2018distributed} for this), its price of anarchy is no smaller than $1/C^\star$. Observe that $G$ features only non-negative rational resource costs' coefficients (i.e., $G\in\cg(\mathbb{Q}_{\ge0})$), therefore \eqref{eq:equalityQ} follows readily. 

If at least one entry in the tuple $\{\theta^\star(x,y,z,j)\}_{(x,y,z,j)\in\mc{I}}$ is not rational, we will prove the existence of a sequence of games $G^k\in\cg(\mathbb{Q}_{\ge0})$ whose worst-case efficiency converges to $\poa(T)$ as $k\rightarrow\infty$. This would imply that \eqref{eq:equalityQ} holds with $\le$ in place of the equality sign, concluding the proof. To do so, let us consider the set 
\[
S=\{\{\theta(x,y,z,j)\}_{(x,y,z,j)\in\mc{I}}~\text{ s.t. \eqref{eq:nash-inequality},  and \eqref{eq:positivity} hold}\}.
\]
Observe that $S$ is non-empty, and that for any tuple belonging to $S$, we can find a sequence of non-negative rational tuples $\{ \{\theta^k(a,x,b,j)\}_{(x,y,z,j)\in\mc{I}} \}_{k=1}^\infty$ (i.e., $\theta^k(a,x,b,j)\in\mb{Q}_{\ge0}$ for all $a,x,b,j$ and $k$), that converges to it. 

Let $\{ \{\theta^k(a,x,b,j)\}_{(x,y,z,j)\in\mc{I}} \}_{k=1}^\infty$ be a sequence of tuples converging to $\{\theta^\star(a,x,b,j)\}_{(x,y,z,j)\in\mc{I}}$, which belongs to $S$. For each tuple $\{\theta^k(a,x,b,j)\}_{(x,y,z,j)\in\mc{I}}$ in the sequence, define the game $G^k$ following the same construction introduced above with $\theta^k(a,x,b,j)$ in place of $\theta^\star(a,x,b,j)$. Following the same reasoning as above, it is $\PNEcost(G^k,T)=\sum_{x,y,z,j} b_j(x+y)(x+y)\theta^k(x,y,z,j)$, and $\mincost(G^k)\le\sum_{x,y,z,j} b_j(x+z)(x+z)\theta^k(x,y,z,j)$. Therefore
\[
\poa^k=\frac{\PNEcost(G^k,T)}{\mincost(G^k)}\ge \frac{\sum_{x,y,z,j} b_j(x+y)(x+y)\theta^k(x,y,z,j)}{\sum_{x,y,z,j} b_j(x+z)(x+z)\theta^k(x,y,z,j)},
\]
from which we conclude that 
\[
\lim_{k\to\infty}\poa^k \ge\! \lim_{k\to\infty}\frac{\sum_{x,y,z,j} b_j(x+y)(x+y)\theta^k(x,y,z,j)}{\sum_{x,y,z,j} b_j(x+z)(x+z)\theta^k(x,y,z,j)}\!=\!\frac{\sum_{x,y,z,j} b_j(x+y)(x+y)\theta^\star(x,y,z,j)}{\sum_{x,y,z,j} b_j(x+z)(x+z)\theta^\star(x,y,z,j)}=\frac{1}{C^\star},
\]
as $\theta^k(a,x,b,j)\rightarrow\theta^\star(a,x,b,j)$ for $k\rightarrow\infty$. This completes the first step. 
 
The second and final step consist in showing that 
\[
\sup_{G \in \cg(\mathbb{Q}_{\ge0})} \frac{\PNEcost(G,\taxmechanism)}{\mincost(G)} = 
\sup_{G \in \cg(\mathbb{Z}_{\ge0})} \frac{\PNEcost(G,\taxmechanism)}{\mincost(G)}.
\]
Towards this goal, for any given game from the above-defined sequence $G^k \in \cg(\mathbb{Q}_{\ge0})$, let $d_{G^k} $ denote the lowest common denominator among the resource cost coefficients $\alpha_j$, across all the resources of the game. Define $\hat \alpha_j = \alpha_j \cdot d_{G^k} \in \mb{Z}_{\ge0}$ for all $j \in \{1, \dots, m\}$, for all resources in $\resset$. Since the tolling mechanisms $T$ is linear by assumption, the equilibrium conditions are independent to uniform scaling of the resource costs and tolls by the coefficient $d_{G^k}$. Therefore any game in the sequence $G^k$ with tolling mechanism $T$ and resource cost coefficients $\{\alpha_j\}_{j=1}^m$ has the same worst-case equilibrium efficiency as a game $\hat{G}^k$ which is identical to $G^k$ except that it has resource cost coefficients $\{\hat \alpha_j\}_{j=1}^m$. Observing that $\hat G^k$ belongs to $\cg(\mathbb{Z}_{\ge0})$ concludes the proof.
\end{proof}

%% file: parts-appendix/appendix-explicit-expression.tex
\begin{lemma}
\label{lem:f_nondecreasing}
Let $b: \mb{N} \to \mb{R}_{\geq 0}$ be a nondecreasing, convex function, and let $0 < \rho \leq 1$ be a given parameter. Further, define the function $f : \{1, \dots, n\} \to \mb{R}$ such that $f(1)=b(1)$ and
\begin{equation} \label{eq:fdefinition}
    f(\jRC+1) \doteq \min_{ \ellRC_\jRC \in\{1,\dots,n\}}
        \frac{\min\{\jRC,n-\ellRC_\jRC\} \cdot f(\jRC) - b(\jRC)\jRC \cdot \rho + b(\ellRC_\jRC)\ellRC_\jRC}{\min\{\ellRC_\jRC,n-\jRC\}},
\end{equation}
for all $\jRC\in\{1,\dots,n-1\}$. Then, for the lowest value $1 \leq \hat \jRC \leq n-1$ such that $f(\hat \jRC+1) < f(\hat \jRC)$, it must hold that $f(\jRC+1) < f(\jRC)$ for all $\jRC\in\{\hat \jRC,\dots,n-1\}$.
\end{lemma}

\begin{proof}
The proof is presented in two parts as follows: in Part~1, we identify inequalities given that $f(\hat \jRC+1) < f(\hat \jRC)$, for $1 \leq \hat \jRC \leq n-1$ as defined in the claim; and, in Part~2, we use a recursive argument to prove that $f(\jRC+1) < f(\jRC)$ holds for all $\hat \jRC+1 \leq \jRC \leq n-1$, using the inequalities derived in Part~1.

\paragraph{Part~1.} 
We define $\ellRC^*_\jRC$ as one of the arguments that minimize the right-hand side of \eqref{eq:fdefinition} for each $\jRC\in\{1,\dots,n-1\}$. By assumption, it must hold that $f(\hat \jRC+1) < f(\hat \jRC)$, which implies that
\begin{align*}
    f(\hat \jRC) > \min_{ \ellRC \in\{1,\dots,n\}} &
            \frac{\min\{\hat \jRC,n-\ellRC\}}{\min\{\ellRC,n-\hat \jRC\}} f(\hat \jRC)
            - \frac{b(\hat \jRC)\hat \jRC}{\min\{\ellRC,n-\hat \jRC\}} \rho
            + \frac{b(\ellRC)\ellRC}{\min\{\ellRC,n-\hat \jRC\}} \\[2mm]
        = \min_{ \ellRC \in\{1,\dots,n\}} &
            \frac{\min\{\hat \jRC-1,n-\ellRC\}}{\min\{\ellRC,n-\hat \jRC+1\}} f(\hat \jRC-1)
            - \frac{b(\hat \jRC-1)(\hat \jRC-1)}{\min\{\ellRC,n-\hat \jRC+1\}} \rho
            + \frac{b(\ellRC)\ellRC}{\min\{\ellRC,n-\hat \jRC+1\}} \\
        &   + \frac{\min\{\hat \jRC,n-\ellRC\}}{\min\{\ellRC,n-\hat \jRC\}} f(\hat \jRC)
            - \frac{\min\{\hat \jRC-1,n-\ellRC\}}{\min\{\ellRC,n-\hat \jRC+1\}} f(\hat \jRC-1)
            - \frac{b(\hat \jRC)\hat \jRC}{\min\{\ellRC,n-\hat \jRC\}} \rho \\
        &   + \frac{b(\hat \jRC-1)(\hat \jRC-1)}{\min\{\ellRC,n-\hat \jRC+1\}} \rho
            + \frac{b(\ellRC)\ellRC}{\min\{\ellRC,n-\hat \jRC\}}
            - \frac{b(\ellRC)\ellRC}{\min\{\ellRC,n-\hat \jRC+1\}},
\end{align*}
where the strict inequality holds by the definition of $f(\hat \jRC+1)$. Recall that
\[ f(\hat \jRC) \doteq \min_{ \ellRC \in\{1,\dots,n\}}            \frac{\min\{\hat \jRC-1,n-\ellRC\}}{\min\{\ellRC,n-\hat \jRC+1\}} f(\hat \jRC-1)
            - \frac{b(\hat \jRC-1)(\hat \jRC-1)}{\min\{\ellRC,n-\hat \jRC+1\}} \rho
            + \frac{b(\ellRC)\ellRC}{\min\{\ellRC,n-\hat \jRC+1\}}.
\]
Thus, if $\ellRC^*_{\hat \jRC} \leq n-\hat \jRC$, the above strict inequality with $f(\hat \jRC)$ can only be satisfied if
\[ f(\hat \jRC+1) < f(\hat \jRC) \leq \hat \jRC f(\hat \jRC) - (\hat \jRC-1) f(\hat \jRC-1) 
    < [ b(\hat \jRC)\hat \jRC - b(\hat \jRC-1)(\hat \jRC-1) ] \cdot \rho. \]
Similarly, if $\ellRC^*_{\hat \jRC} \geq n-\hat \jRC+1$, then it must hold that
\begin{alignat*}{2}
     & &(n-\ellRC^*_{\hat \jRC}) \bigg[ \frac{f(\hat \jRC)}{n-\hat \jRC} - \frac{f(\hat \jRC-1)}{n-\hat \jRC+1} \bigg]
        + \bigg[ \frac{1}{n-\hat \jRC} - \frac{1}{n-\hat \jRC+1} \bigg] b(\ellRC^*_{\hat \jRC}) \ellRC^*_{\hat \jRC}
        &< \bigg[ \frac{b(\hat \jRC)\hat \jRC}{n-\hat \jRC} 
            - \frac{b(\hat \jRC-1)(\hat \jRC-1)}{n-\hat \jRC+1} \bigg] \cdot \rho \\[2mm]
    \implies & &  \bigg[ \frac{1}{n-\hat \jRC} - \frac{1}{n-\hat \jRC+1} \bigg] 
        [(n-\ellRC^*_{\hat \jRC}) f(\hat \jRC) + b(\ellRC^*_{\hat \jRC}) \ellRC^*_{\hat \jRC}]
        &<  \bigg[ \frac{b(\hat \jRC)\hat \jRC}{n-\hat \jRC} 
            - \frac{b(\hat \jRC-1)(\hat \jRC-1)}{n-\hat \jRC+1} \bigg] \cdot \rho \\[2mm]
    \iff & &  f(\hat \jRC+1) &< [b(\hat \jRC)\hat \jRC - b(\hat \jRC-1)(\hat \jRC-1)] \rho,
\end{alignat*}
where the first line implies the second line because $f(\hat \jRC) \geq f(\hat \jRC-1)$, by the definition of $\hat \jRC$ in the claim, and the second line is equivalent to the third by the definitions of $f(\hat \jRC+1)$ and $\ellRC^*_{\hat \jRC}$. This concludes Part~1 of the proof.
\paragraph{Part~2.} 
In this part of the proof, we show by recursion that if $f(\hat \jRC+1)<f(\hat \jRC)$, then $f(\jRC+1)<f(\jRC)$ for all $\jRC\in\{\hat \jRC+1,...,n-1\}$. We do so by showing that, if $f(\jRC)<f(\jRC-1)<\dots<f(\hat \jRC+1)$ for any $\jRC\in\{\hat \jRC+1,\dots,n-1\}$, then it must hold that $f(\jRC+1)<f(\jRC)$. Thus, in the following reasoning, we assume that $\jRC\in\{\hat \jRC+1,\dots,n-1\}$, and that $f(\jRC)<f(\jRC-1)<\dots<f(\hat \jRC+1)$. 

\medskip \noindent We begin with the case of $\ellRC^*_{\jRC-1} < n-\jRC+1$, which gives us that $\ellRC^*_{\jRC-1} \leq n-\jRC$. Observe that
\begingroup
\allowdisplaybreaks
\begin{align*}
    f(\jRC+1) &\doteq \min_{ \ellRC_\jRC \in\{1,\dots,n\}}
        \frac{\min\{\jRC,n-\ellRC_\jRC\}}{\min\{\ellRC_\jRC,n-\jRC\}} f(\jRC+1) 
        - \frac{b(\jRC)\jRC}{\min\{\ellRC_\jRC,n-\jRC\}} \rho 
        + \frac{b(\ellRC_\jRC)\ellRC_\jRC}{\min\{\ellRC_\jRC,n-\jRC\}} \\[2mm]
    &= \min_{ \ellRC_\jRC \in\{1,\dots,n\}}
        \frac{\min\{\jRC-1,n-\ellRC_\jRC\}}{\min\{\ellRC_\jRC,n-\jRC+1\}} f(\jRC-1)
        - \frac{b(\jRC-1)(\jRC-1)}{\min\{\ellRC_\jRC,n-\jRC+1\}} \rho
        + \frac{b(\ellRC_\jRC)\ellRC_\jRC}{\min\{\ellRC_\jRC,n-\jRC+1\}} \\
    & \hspace{45pt} + \frac{\min\{\jRC,n-\ellRC_\jRC\}}{\min\{\ellRC_\jRC,n-\jRC\}} f(\jRC+1) 
        - \frac{\min\{\jRC-1,n-\ellRC_\jRC\}}{\min\{\ellRC_\jRC,n-\jRC+1\}} f(\jRC-1)
        - \frac{b(\jRC)\jRC}{\min\{\ellRC_\jRC,n-\jRC\}} \rho \\
    & \hspace{45pt} + \frac{b(\jRC-1)(\jRC-1)}{\min\{\ellRC_\jRC,n-\jRC+1\}} \rho
        + \frac{b(\ellRC_\jRC)\ellRC_\jRC}{\min\{\ellRC_\jRC,n-\jRC\}}
        - \frac{b(\ellRC_\jRC)\ellRC_\jRC}{\min\{\ellRC_\jRC,n-\jRC+1\}}\\[2mm]
    &\leq f(\jRC) + \frac{\jRC}{\ellRC^*_{\jRC-1}} f(\jRC) - \frac{\jRC-1}{\ellRC^*_{\jRC-1}} f(\jRC-1) 
        - \frac{b(\jRC)\jRC-b(\jRC-1)(\jRC-1)}{\ellRC^*_{\jRC-1}} \rho \\[2mm]
    &< f(\jRC) + \frac{1}{\ellRC^*_{\jRC-1}} f(\hat \jRC+1) - \frac{1}{\ellRC^*_{\jRC-1}}[b(\jRC)\jRC-b(\jRC-1)(\jRC-1)] \rho < f(\jRC),
\end{align*}
\endgroup
where the first inequality holds by evaluating the minimization at $\ellRC_\jRC = \ellRC^*_{\jRC-1}$, the second inequality holds because $f(\jRC)<f(\jRC-1)$ and $f(\jRC) \leq f(\hat \jRC+1)$, by assumption, and the final inequality holds by the result showed in Part~1 and because $b(\cdot)$ is nondecreasing and convex.

\medskip \noindent Next, we consider the scenario in which $\ellRC^*_{\jRC-1} > n-\jRC+1$. Observe that
\begin{align*}
    f(\jRC+1) &\leq f(\jRC) 
        + (n-\ellRC^*_{\jRC-1}) \bigg[ \frac{f(\jRC)}{n-\jRC} - \frac{f(\jRC-1)}{n-\jRC+1} \bigg] 
        + \bigg[ \frac{1}{n-\jRC} - \frac{1}{n-\jRC+1} \bigg] b(\ellRC^*_{\jRC-1})\ellRC^*_{\jRC-1} \\
    &\qquad - \frac{b(\jRC)\jRC}{n-\jRC} \rho + \frac{b(\jRC-1)(\jRC-1)}{n-\jRC+1} \rho \\[2mm]
    &< f(\jRC)
        + \bigg[ \frac{1}{n-\jRC} - \frac{1}{n-\jRC+1} \bigg] 
            [(n-\ellRC^*_{\jRC-1}) f(\jRC-1) + b(\ellRC^*_{\jRC-1})\ellRC^*_{\jRC-1}] \\
    &\qquad - \frac{b(\jRC)\jRC}{n-\jRC} \rho + \frac{b(\jRC-1)(\jRC-1)}{n-\jRC+1} \rho \\[2mm]
    &= f(\jRC) 
        + \bigg[ \frac{1}{n-\jRC}-\frac{1}{n-\jRC+1} \bigg] [ (n-\jRC+1)f(\jRC) + b(\jRC-1)(\jRC-1)\rho ]\\
    &\qquad - \frac{b(\jRC)\jRC}{n-\jRC} \rho + \frac{b(\jRC-1)(\jRC-1)}{n-\jRC+1} \rho \\
    &\leq f(\jRC) + \frac{1}{n-\jRC}f(\hat \jRC+1) - \frac{1}{n-\jRC} [b(\jRC)\jRC-b(\jRC-1)(\jRC-1)] \rho \\[2mm]
    &< f(\jRC),
\end{align*}
where the first inequality holds by evaluating the minimization at $\ellRC_\jRC = \ellRC^*_{\jRC-1}$, the second inequality holds because $f(\jRC)<f(\jRC-1)$, by assumption, the equality holds by the definitions of $f(\jRC)$ and $\ellRC^*_{\jRC-1}$, the third inequality holds because $f(\jRC)\leq f(\hat \jRC+1)$, by assumption, and the final inequality holds by the identity we showed in Part~1 and because $b$ is nondecreasing and convex.

\medskip\noindent Finally, we consider the scenario in which $\ellRC^*_{\jRC-1} = n-\jRC+1$. Observe that
\begin{align*}
    f(\jRC+1) &\leq f(\jRC) + \frac{\jRC-1}{n-\jRC} f(\jRC) - \frac{\jRC-1}{n-\jRC+1} f(\jRC-1)
        - \frac{b(\jRC)\jRC}{n-\jRC} \rho + \frac{b(\jRC-1)(\jRC-1)}{n-\jRC+1} \rho\\
    &\qquad + \bigg[ \frac{1}{n-\jRC} - \frac{1}{n-\jRC+1} \bigg] b(n-\jRC+1)(n-\jRC+1) \\[2mm]
    &< f(\jRC)
        + \bigg[ \frac{1}{n-\jRC} - \frac{1}{n-\jRC+1} \bigg] 
            [(n-\ellRC^*_{\jRC-1}) f(\jRC-1) + b(\ellRC^*_{\jRC-1})\ellRC^*_{\jRC-1}]
            \\
    &\qquad  - \frac{b(\jRC)\jRC}{n-\jRC} \rho + \frac{b(\jRC-1)(\jRC-1)}{n-\jRC+1} \rho \\[2mm]
    &< f(\jRC),
\end{align*}
where the first inequality holds by evaluating the minimization at $\ellRC_\jRC = \ellRC^*_{\jRC-1}$, the second inequality holds because $f(\jRC)<f(\jRC-1)$, by assumption, and the final inequality holds by the same reasoning as for $\ellRC^*_{\jRC-1} > n-\jRC+1$.
\end{proof}

%% file: parts-appendix/appendix-largen.tex
\begin{lemma}
\label{lem:largen}
For given $d\ge1$, consider the linear program defined in \eqref{eq:overconstraintedsimplifiedlp}, and let $(f^{\rm opt}, \rho^{\rm opt})$ be a solution. Further let $\myf$ be defined in \eqref{eq:definefext} utilizing $(f^{\rm opt}, \rho^{\rm opt})$. The price of anarchy of $\myf$ over congestion games of degree $d$ with possibly infinitely many agents is upper bounded by $1/\rho^{\infty}$, where
\[
\rho^{\infty}=\min\left\{{\opt\rho},~~{\beta - d\left(1+\frac{2}{\bar n}\right)^{d+1} \left(\frac{\beta}{d+1}\right)^{1+\frac{1}{d}}} \right\}.
\]
This result is tight for pure Nash equilibria and holds for coarse correlated equilibria too.
\end{lemma}
\begin{proof}
Since $\myf$ is non-decreasing by construction, we characterize its performance over games with a maximum of $n$ agents through the linear program in \eqref{eq:simplifiedLP-compute} with $b(x)=x^d$, that is
\begin{maxi}
	{{\scriptstyle \rho \in \mathbb{R}, \nu \in \mathbb{R}_{\geq 0}}}
	{\nquad\rho} {\label{eq:characterize_finfty}}{}
	\addConstraint{\nquad\!\!\!\!\! b(\v)\v - \rho b(\u)\u + \nu[\myf(\u)\u - \myf(\u+1)\v]}{\ge 0}{\quad\forall u,v\in\{0,\dots,n\}\quad u+\v \leq n,}
	\addConstraint{\nquad\!\!\!\!\! b(\v)\v - \rho b(\u)\u + \nu[\myf(\u)(n-\v) - \myf(\u+1)(n-\u)]}{\ge 0}{\quad\forall u,v\in\{0,\dots,n\}\quad u+\v > n.}
\end{maxi}
Upper bounding the price of anarchy of $\myf$ amounts to finding a feasible solution to this linear program; the challenging task is that we intend to do so for $n$ arbitrary large. We claim that $(\rho,\nu)=(\rho^\infty,1)$ is feasible for any (possibly infinite) $n$, so that the claimed upper bound on the price of anarchy follows. 
To prove this, we divide the discussion in two parts as the expressions are different for $\u+\v\le n$ and $\u+\v> n$. Before doing so, we study the degenerate case of $\u=0$, for which the constraints reduce to $\myf(1)\le 1$, which holds as $\myf$ is feasible for the linear program in  \eqref{eq:overconstraintedsimplifiedlp} which already includes this constraint.

\begin{itemize}
\item[--]	Case of $\u+\v\le n$, $\u\ge 1$. The constraints read as 
\be
\v^{d+1}-\rho \u^{d+1}+\nu[\u\myf(\u)-\v\myf(\u+1)]\ge 0
\label{eq:constraintpart1}
\ee
which we want to hold for any integers $\u\ge 1$, $\v\ge 0$ (the bound on the indices $\u+\v\le n$ can be dropped as we are interested in the case of arbitrary $n$).
\input{figures-tex/figurecase1}
\begin{itemize}
\item[$\bullet$]
In the region where $0\le \v\le \u$, $1\le \u< \bar{n}/2$ (region A in \cref{fig:figure1proof}) these constraints certainly hold with $\rho\le\rho^{\rm opt}$, $\nu=1$. This follows because $\myf(\u)=f^{\rm opt}(\u)$ when $\u\le \bar{n}/2$, and $f^{\rm opt}$ is feasible for the program in \eqref{eq:overconstraintedsimplifiedlp} which includes these constraints.
\item[$\bullet$] 
In the region where $\v> \u$, $1\le \u< \bar{n}/2$ (region B in \cref{fig:figure1proof}) these constraints also hold with $\rho\le\rho^{\rm opt}$, $\nu=1$ thanks to \cref{lemma1} part a).
\item[$\bullet$] We are left with the region where $\u\ge\bar{n}/2$, $\v\ge 0$  (region C in \cref{fig:figure1proof}). In this case, $\myf(\u)=\beta \cdot \u^d$ by definition. With the choice of $\nu=1$, the constraints in \eqref{eq:constraintpart1} read 
\be
\v^{d+1}-\rho \u^{d+1}+\beta \u^{d+1}-\beta \v(\u+1)^d\ge 0.
\label{eq:inequalityfext1}
\ee
Observe that, for a fixed choice of integer $u\ge\bar n / 2$, the left hand side of \eqref{eq:inequalityfext1} is a convex function of $v$ for $v\ge0$. Its corresponding minimum value over the non-negative reals is 
\[
\left(\frac{\beta}{d+1}\right)^{1+\frac{1}{d}}(u+1)^{d+1}-
\beta\left(\frac{\beta}{d+1}\right)^{\frac{1}{d}}(u+1)^{d+1}+
(\beta - \rho) u^{d+1}.
\]
Hence, \eqref{eq:inequalityfext1} is satisfied for a fixed choice of $u\ge\bar n / 2$ and all integers $v\ge0$ if the latter expression is non-negative. Simple algebra shows that this is the case if 
\be
\rho \le {\beta - d\left(1+\frac{1}{u}\right)^{d+1} \left(\frac{\beta}{d+1}\right)^{1+\frac{1}{d}}}.
\label{eq:rho-largen-withu}
\ee 
Since we would like \eqref{eq:inequalityfext1} to hold for all $u\ge\bar n / 2$, and since the right-hand side in \eqref{eq:rho-largen-withu} is increasing in $n$, it suffices to ask for \eqref{eq:rho-largen-withu} to hold at the lowest admissible $u$, i.e. $u=\bar{n}/2$. Therefore, in order for \eqref{eq:inequalityfext1} to be satisfied for any $\u\ge \bar{n}/2$, it suffices to select 
\be
\rho \le {\beta - d\left(1+\frac{2}{\bar n}\right)^{d+1} \left(\frac{\beta}{d+1}\right)^{1+\frac{1}{d}}}.
\label{eq:secondrho}
\ee
\end{itemize}
\item[--]	Case of $\u+\v> n$, $\u\ge 1$. \cref{lemma1} part b) shows that the constraints corresponding to $\u+\v>n$ are satisfied for arbitrary $n$ with the choice of $\nu=1$, and $\rho$ as in \eqref{eq:secondrho} thanks to the fact that $\myf:\mb{N}\rightarrow\mb{R}$ is non-decreasing.
\end{itemize}
In conclusion, we verified that $(\rho,\nu)=(\rho^\infty,1)$ is feasible for the program in \eqref{eq:characterize_finfty}. It follows that the price of anarchy of $\myf$ over games with arbitrarily large $n$ is upper bounded as in the claim.	
\end{proof}

\begin{lemma}
\leavevmode
\begin{itemize}
	\item[a)] Let $f:\{1,\dots,n\}\rightarrow\mb{R}$, $\rho\ge 0$, and $f(x)\le x^d$ for all $1\le x\le n$ and $d\ge1$. The constraints
	$
	\v^{d+1}-\rho \u^{d+1}+\u f(\u)-\v f(\u+1)\ge 0
	$ obtained for any $\v\in\mb{N}$, $\v\ge \u$, $\u\in\{1,\dots,n-1\}$ are satisfied if the same inequality holds for all $\v=\u\in \{1,\dots,n-1\}$.
	
		\item[b)] Let $f:\mb{N}\rightarrow\mb{R}$ be non-decreasing, $d\ge1$, $\rho\ge0$. If the constraints $\v^{d+1}-\rho \u^{d+1}+\u f(\u)-\v f(\u+1)\ge 0$ hold for all non-negative integers $u,v$, then $\v^{d+1}-\rho \u^{d+1}+(n-\v) f(\u)-(n-\u) f(\u+1)\ge 0$ hold for all non-negative integers $u,v$ with $u+v>n$, for any choice of $n\ge1$ integer.
		\end{itemize}
		\label{lemma1}
\end{lemma}
\begin{proof} We prove the two claims separately.\\
{\it First claim.} For $\v=\u$ the constraints of interest reduces to $\rho \u^{d+1} \le  \u^{d+1}+\u(f(\u)-f(\u+1)),$ while for $\v=\u+p$ (with $p\ge 1$) the constraint reads as $\rho \u^{d+1} \le   (\u+p)^{d+1}+\u f(\u)-(\u+p)f(\u+1).$ Proving the claim amounts to showing  
\[
\u^{d+1}+\u(f(\u)-f(\u+1)) \le (\u+p)^{d+1}+\u f(\u)-(\u+p)f(\u+1)
 \iff f(u+1)\le \frac{(u+p)^{d+1}-u^{d+1}}{p}
\]
for $p\ge 1$. The right hand side is minimized at $p=1$ (it describes the slope of the secant to the function $u^{d+1}$ at abscissas $u$ and $u+p$) due to the convexity of $u^{d+1}$. Therefore, it suffices to ensure that $f(u+1)\le (u+1)^{d+1}-u^{d+1}$ for any choice of $\u\in\{1,\dots,n-1\}$. By assumption, $f(u+1)\le (u+1)^d$, so that we can equivalently prove $(u+1)^d\le(u+1)^{d+1}-u^{d+1}$. The latter holds, as required, for all $\u\in\{1,\dots,n-1\}$ since 
\[
(u+1)^d\le (u+1)^{d+1}-u^{d+1} \iff (u+1)^d\le (u+1)^{d}(u+1)-u^{d+1} \iff u^d\le (u+1)^d.
\]

\noindent{\it Second claim.} By assumption, for all non-negative integers $u,v$ it is $\rho u^{d+1}\le v^{d+1}+uf(u)-vf(u+1)$, while we intend to show that for all non-negative integers $u,v$ with $u+v>n$ it is $\rho u^{d+1}\le v^{d+1}+(n-v)f(u)-(n-u)f(u+1)$, regardless for the choice of $n\ge1$ integer. Proving this is equivalent to showing 
\[
uf(u)-vf(u+1)\le (n-v)f(u)-(n-u)f(u+1)
\iff (u+v-n)(f(u)-f(u+1))\le 0,
\]
which holds, as required, due to the fact that $u+v-n> 0$ and $f(u)-f(u+1)\le0$ due to $f$ being non-decreasing.
\end{proof}

%% file: figures-tex/figurecase1.tex
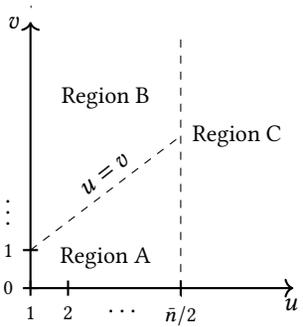
\begin{figure}[h!]
\begin{center}
\begin{tikzpicture}[scale=1]
\draw[thick,->] (0,0) -- (3.5,0);
\draw[thick] (3.7,0) -- (3.7,0) node[anchor=north east]{$\u$};
\draw[thick,->] (0,0) -- (0,3.5); 
\draw[thick] (0,3.7) -- (0,3.7) node[anchor=north east]{$\v$};
\draw[dashed] (0,0.5) -- (2,2) ;
\draw[dashed] (2,0) -- (2,3.3) ;
\node[rotate = 38] at (1,1.5) {$\u=\v$};
\node at (1,0.4) {\small Region A};
\node at (1,2.5) {\small Region B};
\node at (2.75,2) {\small Region C};
\draw[thick] (0,0.1) -- (0,-0.1) node[anchor=north]{\footnotesize$1$};
\draw[thick] (0.5,0.1) -- (0.5,-0.1)node[anchor=north]{\footnotesize$2$};
\draw[thick] (2,0.1) -- (2,-0.1);
\draw[thick] (2,-0.1) -- (2,-0.1)node[anchor=north]{\footnotesize$\bar{n}/2$};
\node[thick] (a) at (1.25,-0.3) {$\dots$};
\draw[thick] (0.1,0) -- (-0.1,0)node[anchor=east]{\footnotesize$0$};
\draw[thick] (0.1,0.5) -- (-0.1,0.5)node[anchor=east]{\footnotesize$1$};
\node[thick] (a) at (-0.3,1.1) {$\vdots$};
\end{tikzpicture}
\end{center}
\caption{Regions A, B, and C utilized in proof for the case $\u+\v\le n$.}
\label{fig:figure1proof}
\end{figure}

%% file: parts-appendix/appendix-constant-toll.tex
\begin{lemma}
\label{lem:v>u_does_not_matter}
Let $n\in\mb{N}$, and assume that the basis functions $\{b_1,\dots,b_m\}$, $b_j:\mb{N}\rightarrow \mb{R}$ are convex (in the discrete sense), positive, and non-decreasing, for all $j=1,\dots,m$. Then, the constraints appearing in \eqref{eq:LP_constanttoll_nusigma} with $u=0$, $v=\{1,\dots,n\}$ are satisfied if the constraint with $(u,v)=(0,1)$ holds. 
Similarly, the constraints with $u\in\{1,\dots,n\}$, $v\in\{0,\dots,n\}$ and $u<v$ are satisfied if the constraints with $u\in\{1,\dots,n\}$, $v\in\{0,\dots,n\}$ and $u\ge v$ hold.
\end{lemma}
\begin{proof}
We begin with the constraints obtained for $u=0$, which read as $b_j(v)v-\nu b_j(1)v-\sigma_jv\ge0$. The constraint with $v=0$ holds trivially, while the tightest constraint for $v>0$ is obtained with $v=1$ due to the fact that $b_j(v)v$ is increasing for $v>0$. This shows that it is sufficient to consider the constraint with $(u,v)=(0,1)$.

We now consider the constraints obtained for $u\ge1$ and divide the proof in three parts, according to the regions in \cref{fig:figure2proof}.
\input{figures-tex/figure-regions.tex}
\begin{itemize}
\item In the region where $v>u$ and $u+v\le n$ (Region A in \cref{fig:figure2proof}), we show that if the constraint obtained with $v=u$ holds, then the constraints with $v>u$ also hold. Note that feasible values of $\nu$ are upper bounded by $\nu\le 1$. This is because the constraint with $(u,v)=(0,1)$ reads as $(1-\nu)b_j(1)\ge \sigma_j$, and since $\sigma_j\ge0$, $b_j(1)>0$, every feasible $\nu$ must satisfy $1-\nu\ge0$. The constraint with $v=u+p$, $p\ge1$ read as $b_j(u+p)(u+p) - \rho b_j(\u)\u + \nu [b_j(\u)\u - b_j(\u+1)(u+p)] - \sigma_j p \geq 0$, the tightest of which is obtained for the largest feasible value of $\sigma_j$, that is $\sigma_j=(1-\nu)b_j(1)$. The constraint with $v=u$ reads as $ub_j(u)-\rho u b_j(u)+\nu u(b_j(u)-b_j(u+1))\ge0$. Therefore, we intend to show that for every $\rho$ and $0\le\nu\le1$ satisfying 
\[ub_j(u)-\rho u b_j(u)+\nu u(b_j(u)-b_j(u+1))\ge0,\quad\text{it also holds for all $p\ge 1$ that}\]
\[b_j(u+p)(u+p) - \rho b_j(\u)\u + \nu [b_j(\u)\u - b_j(\u+1)(u+p)] - (1-\nu)b_j(1) p \ge 0.\] 
Since both constraints describe straight lines in the plane $(\nu,\rho)$, it suffices to verify that this is true at the extreme points $\nu=0$ and $\nu=1$. 

When $\nu=0$ the constraint with $v=u$ and $v=u+p$ read as $\rho b_j(\u)\u\le b_j(\u)\u$, and  $ \rho b_j(\u)\u \le b_j(u+p)(u+p)  -b_j(1) p $. We therefore intend to show that 
\be
\label{eq:convexity-const-toll}
b_j(\u)\u\le b_j(u+p)(u+p)  -b_j(1) p \iff \frac{b_j(u+p)(u+p)-ub_j(u)}{p}\ge b_j(1),
\ee
which holds since $b_j(u)u$ is convex and $b_j(u)$ is non-decreasing so that $\frac{b_j(u+p)(u+p)-ub_j(u)}{p}\ge b_j(u+1)(u+1)-ub_j(u)\ge b_j(u+1)\ge b_j(1)$. When $\nu=1$, following a similar reasoning, we are left to show that $ub_j(u)\le b_j(u+p)(u+p)-pb_j(u+1)$, which holds thanks to the non-decreasingness of $b_j(u)$, indeed $b_j(u+p)(u+p)-pb_j(u+1)\ge b_j(u+p)u\ge  b_j(u)u$.
\item In the region where $v>u$,  $u+v\ge n$ and $u\le n/2$ (Region B in \cref{fig:figure2proof}), we intend to show that the following constraints hold $b_j(\v)\v - \rho b_j(\u)\u + \nu [b_j(\u)(n-\v) - b_j(\u +1)(n-\u)] + \sigma_j(\u-\v) \geq 0 $. We do so by observing that the proof of the previous point did not require at all that $u+v\le n$. Therefore, exploiting the same proof, we have $b_j(\v)\v - \rho b_j(\u)\u + \nu [b_j(\u)\u - b_j(\u+1)\v] + \sigma_j(\u - \v) \geq 0$ also for $u+v\ge n$, i.e., in the region B of interest. We exploit this to conclude, and, in particular, we show that the satisfaction of latter constraint implies the desired result. Towards this goal we need to show that, for $u+v> n$ it is  
\[
 b_j(\u)(n-\v) - b_j(\u +1)(n-\u) 
 \ge  b_j(\u)\u - b_j(\u+1)\v 
\Leftrightarrow
 (b_j(\u)- b_j(\u +1))(n-\u-\v)\ge0,
 \]
which holds due to the non-decreasingness of $b_j$ and to $u+v>n$.
 
\item In the region where $v>u$,  $u+v\ge n$ and $u> n/2$ (Region C in \cref{fig:figure2proof}), we use the same approach of that in the first point. In particular, we intend to show that the when the constraints with $v=u$ hold, also the constraints with $v>u$ do so. Following a similar reasoning as in the above, this amount to showing for every $\rho$ and $0\le\nu\le1$ satisfying 
\[b_j(\u)\u - \rho b_j(\u)\u + \nu (n-\u)(b_j(\u) - b_j(\u +1))  \geq 0,\qquad\text{it also holds for all $p\ge 1$ that}\]
\[b_j(\u+p)(\u+p) - \rho b_j(\u)\u + \nu [b_j(\u)(n-\u-p) - b_j(\u +1)(n-\u)] -(1-\nu)b_j(1)p \geq 0.\]
Since both constraints describe straight lines in the plane $(\nu,\rho)$, it suffices to verify that this is true at the extreme points $\nu=0$ and $\nu=1$. When $\nu=0$ we are left with $b_j(\u)\u\le b_j(u+p)(u+p)  -b_j(1) p$, which we already proved in \eqref{eq:convexity-const-toll}. When $\nu=1$, we need to show that $b_j(u+p)(u+p)-pb_j(u)\ge ub_j(u)$, which holds by non-decreasingness of~$b_j$.
\end{itemize}
\end{proof}

%% file: figures-tex/figure-regions.tex
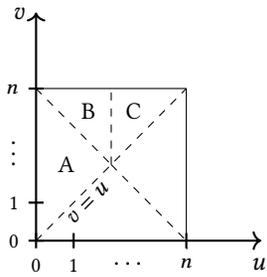
\begin{figure}[h!]
\begin{center}
\begin{tikzpicture}[scale=1]
\draw[thick,->] (0,0) -- (3,0);
\draw[thick] (3.2,-0.1) -- (3.2,-0.1) node[anchor=north east]{$\u$};
\draw[thick,->] (0,0) -- (0,3); 
\draw[thick] (0,3.2) -- (0,3.2) node[anchor=north east]{$\v$};
\draw[-] (2,0) -- (2,2) ;
\draw[-] (0,2) -- (2,2) ;
\draw[dashed] (0,2) -- (1,1) ;
\draw[dashed] (1,1) -- (1,2) ;
\draw[dashed] (1,1) -- (2,0) ;
\draw[dashed] (0,0) -- (2,2) ;
\node[rotate = 45] at (0.7,0.5) {\small$\v=\u$};
\node at (0.4,1) {\small A};
\node at (0.7,1.7) {\small B};
\node at (1.3,1.7) {\small C};
\draw[thick] (0,0.1) -- (0,-0.1) node[anchor=north]{\footnotesize$0$};
\draw[thick] (0.5,0.1) -- (0.5,-0.1)node[anchor=north]{\footnotesize$1$};
\draw[thick] (2,0.1) -- (2,-0.1);
\draw[thick] (2,-0.1) -- (2,-0.1)node[anchor=north]{\footnotesize$n$};
\node[thick] (a) at (1.25,-0.3) {$\dots$};
\draw[thick] (0.1,0) -- (-0.1,0)node[anchor=east]{\footnotesize$0$};
\draw[thick] (0.1,0.5) -- (-0.1,0.5)node[anchor=east]{\footnotesize$1$};
\node[thick] (a) at (-0.3,1.25) {$\vdots$};
\draw[thick] (0.1,2) -- (-0.1,2)node[anchor=east]{\footnotesize$n$};
\end{tikzpicture}
\vspace*{-2mm}
\end{center}
\caption{Regions A, B, and C utilized in proof of \cref{lem:v>u_does_not_matter}.\vspace*{-2mm}}
\label{fig:figure2proof}
\end{figure}

%% file: parts-body/ack-biblio.tex
%\begin{acks}
%ACKnolwedge XYZ
%\end{acks}

%% BIBLIO
\bibliographystyle{ACM-Reference-Format}
\bibliography{references.bib}

%% file: main-TEAC-final.bbl

\begin{thebibliography}{37}

%%% the Web address.

\ifx \showCODEN    \undefined \def \showCODEN     #1{\unskip}     \fi
\ifx \showDOI      \undefined \def \showDOI       #1{#1}\fi
\ifx \showISBNx    \undefined \def \showISBNx     #1{\unskip}     \fi
\ifx \showISBNxiii \undefined \def \showISBNxiii  #1{\unskip}     \fi
\ifx \showISSN     \undefined \def \showISSN      #1{\unskip}     \fi
\ifx \showLCCN     \undefined \def \showLCCN      #1{\unskip}     \fi
\ifx \shownote     \undefined \def \shownote      #1{#1}          \fi
\ifx \showarticletitle \undefined \def \showarticletitle #1{#1}   \fi
\ifx \showURL      \undefined \def \showURL       {\relax}        \fi
% invisible to TeX
\providecommand\bibfield[2]{#2}
\providecommand\bibinfo[2]{#2}
\providecommand\natexlab[1]{#1}
\providecommand\showeprint[2][]{arXiv:#2}

\bibitem[\protect\citeauthoryear{Aland, Dumrauf, Gairing, Monien, and
  Schoppmann}{Aland et~al\mbox{.}}{2011}]%
        {aland2006exact}
\bibfield{author}{\bibinfo{person}{Sebastian Aland}, \bibinfo{person}{Dominic
  Dumrauf}, \bibinfo{person}{Martin Gairing}, \bibinfo{person}{Burkhard
  Monien}, {and} \bibinfo{person}{Florian Schoppmann}.}
  \bibinfo{year}{2011}\natexlab{}.
\newblock \showarticletitle{Exact price of anarchy for polynomial congestion
  games}.
\newblock \bibinfo{journal}{\emph{SIAM J. Comput.}} \bibinfo{volume}{40},
  \bibinfo{number}{5} (\bibinfo{year}{2011}), \bibinfo{pages}{1211--1233}.
\newblock


\bibitem[\protect\citeauthoryear{Anshelevich, Dasgupta, Kleinberg, Tardos,
  Wexler, and Roughgarden}{Anshelevich et~al\mbox{.}}{2008}]%
        {anshelevich2008price}
\bibfield{author}{\bibinfo{person}{Elliot Anshelevich},
  \bibinfo{person}{Anirban Dasgupta}, \bibinfo{person}{Jon Kleinberg},
  \bibinfo{person}{Eva Tardos}, \bibinfo{person}{Tom Wexler}, {and}
  \bibinfo{person}{Tim Roughgarden}.} \bibinfo{year}{2008}\natexlab{}.
\newblock \showarticletitle{The price of stability for network design with fair
  cost allocation}.
\newblock \bibinfo{journal}{\emph{SIAM J. Comput.}} \bibinfo{volume}{38},
  \bibinfo{number}{4} (\bibinfo{year}{2008}), \bibinfo{pages}{1602--1623}.
\newblock


\bibitem[\protect\citeauthoryear{Awerbuch, Azar, and Epstein}{Awerbuch
  et~al\mbox{.}}{2013}]%
        {awerbuch2005price}
\bibfield{author}{\bibinfo{person}{Baruch Awerbuch}, \bibinfo{person}{Yossi
  Azar}, {and} \bibinfo{person}{Amir Epstein}.}
  \bibinfo{year}{2013}\natexlab{}.
\newblock \showarticletitle{The price of routing unsplittable flow}.
\newblock \bibinfo{journal}{\emph{SIAM J. Comput.}} \bibinfo{volume}{42},
  \bibinfo{number}{1} (\bibinfo{year}{2013}), \bibinfo{pages}{160--177}.
\newblock


\bibitem[\protect\citeauthoryear{Beckmann, McGuire, and Winsten}{Beckmann
  et~al\mbox{.}}{1956}]%
        {Beckmann1956}
\bibfield{author}{\bibinfo{person}{Martin Beckmann}, \bibinfo{person}{Charles~B
  McGuire}, {and} \bibinfo{person}{Christopher~B Winsten}.}
  \bibinfo{year}{1956}\natexlab{}.
\newblock \bibinfo{booktitle}{\emph{Studies in the Economics of
  Transportation}}.
\newblock \bibinfo{type}{{T}echnical {R}eport}.
\newblock


\bibitem[\protect\citeauthoryear{Bergendorff, Hearn, and Ramana}{Bergendorff
  et~al\mbox{.}}{1997}]%
        {bergendorff1997congestion}
\bibfield{author}{\bibinfo{person}{Pia Bergendorff}, \bibinfo{person}{Donald~W
  Hearn}, {and} \bibinfo{person}{Motakuri~V Ramana}.}
  \bibinfo{year}{1997}\natexlab{}.
\newblock \showarticletitle{Congestion toll pricing of traffic networks}.
\newblock In \bibinfo{booktitle}{\emph{Network Optimization}}.
  \bibinfo{publisher}{Springer}, \bibinfo{pages}{51--71}.
\newblock


\bibitem[\protect\citeauthoryear{Bil{\`{o}} and Vinci}{Bil{\`{o}} and
  Vinci}{2019}]%
        {BiloV19}
\bibfield{author}{\bibinfo{person}{Vittorio Bil{\`{o}}} {and}
  \bibinfo{person}{Cosimo Vinci}.} \bibinfo{year}{2019}\natexlab{}.
\newblock \showarticletitle{Dynamic taxes for polynomial congestion games}.
\newblock \bibinfo{journal}{\emph{{ACM} Transactions on Economics and
  Computation}} \bibinfo{volume}{7}, \bibinfo{number}{3}
  (\bibinfo{year}{2019}), \bibinfo{pages}{15:1--15:36}.
\newblock


\bibitem[\protect\citeauthoryear{Caragiannis, Flammini, Kaklamanis,
  Kanellopoulos, and Moscardelli}{Caragiannis et~al\mbox{.}}{2011}]%
        {CaragiannisFKKM11}
\bibfield{author}{\bibinfo{person}{Ioannis Caragiannis},
  \bibinfo{person}{Michele Flammini}, \bibinfo{person}{Christos Kaklamanis},
  \bibinfo{person}{Panagiotis Kanellopoulos}, {and} \bibinfo{person}{Luca
  Moscardelli}.} \bibinfo{year}{2011}\natexlab{}.
\newblock \showarticletitle{Tight bounds for selfish and greedy load
  balancing}.
\newblock \bibinfo{journal}{\emph{Algorithmica}} \bibinfo{volume}{61},
  \bibinfo{number}{3} (\bibinfo{year}{2011}), \bibinfo{pages}{606--637}.
\newblock


\bibitem[\protect\citeauthoryear{Caragiannis, Kaklamanis, and
  Kanellopoulos}{Caragiannis et~al\mbox{.}}{2010}]%
        {caragiannis2010taxes}
\bibfield{author}{\bibinfo{person}{Ioannis Caragiannis},
  \bibinfo{person}{Christos Kaklamanis}, {and} \bibinfo{person}{Panagiotis
  Kanellopoulos}.} \bibinfo{year}{2010}\natexlab{}.
\newblock \showarticletitle{Taxes for linear atomic congestion games}.
\newblock \bibinfo{journal}{\emph{{ACM} Transactions on Algorithms}}
  \bibinfo{volume}{7}, \bibinfo{number}{1} (\bibinfo{year}{2010}),
  \bibinfo{pages}{13:1--13:31}.
\newblock


\bibitem[\protect\citeauthoryear{Chandan, Paccagnan, Ferguson, and
  Marden}{Chandan et~al\mbox{.}}{2019d}]%
        {DBLP:conf/netecon/ChandanPFM19}
\bibfield{author}{\bibinfo{person}{Rahul Chandan}, \bibinfo{person}{Dario
  Paccagnan}, \bibinfo{person}{Bryce~L. Ferguson}, {and}
  \bibinfo{person}{Jason~R. Marden}.} \bibinfo{year}{2019}\natexlab{d}.
\newblock \showarticletitle{Computing optimal taxes in atomic congestion
  games}. In \bibinfo{booktitle}{\emph{Proceedings of the 14th Workshop on the
  Economics of Networks, Systems and Computation, NetEcon, Phoenix, Arizona,
  USA, June 28, 2019}}. \bibinfo{publisher}{{ACM}}, \bibinfo{pages}{2:1}.
\newblock


\bibitem[\protect\citeauthoryear{Chandan, Paccagnan, and Marden}{Chandan
  et~al\mbox{.}}{2019a}]%
        {Chandancode19}
\bibfield{author}{\bibinfo{person}{Rahul Chandan}, \bibinfo{person}{Dario
  Paccagnan}, {and} \bibinfo{person}{Jason~R Marden}.}
  \bibinfo{year}{2019}\natexlab{a}.
\newblock \bibinfo{title}{MATLAB and {P}ython packages to compute and optimize
  the price of anarchy}.
\newblock
  \bibinfo{howpublished}{\url{https://github.com/rahul-chandan/resalloc-poa}}.
\newblock


\bibitem[\protect\citeauthoryear{Chandan, Paccagnan, and Marden}{Chandan
  et~al\mbox{.}}{2019b}]%
        {chandan19optimaljournal}
\bibfield{author}{\bibinfo{person}{Rahul Chandan}, \bibinfo{person}{Dario
  Paccagnan}, {and} \bibinfo{person}{Jason~R Marden}.}
  \bibinfo{year}{2019}\natexlab{b}.
\newblock \showarticletitle{Optimal mechanisms for distributed
  resource-allocation}.
\newblock \bibinfo{journal}{\emph{arXiv preprint arXiv:1911.07823}}
  (\bibinfo{year}{2019}).
\newblock


\bibitem[\protect\citeauthoryear{Chandan, Paccagnan, and Marden}{Chandan
  et~al\mbox{.}}{2019c}]%
        {chandan2019when}
\bibfield{author}{\bibinfo{person}{Rahul Chandan}, \bibinfo{person}{Dario
  Paccagnan}, {and} \bibinfo{person}{Jason~R. Marden}.}
  \bibinfo{year}{2019}\natexlab{c}.
\newblock \showarticletitle{When smoothness is not enough: {T}oward exact
  quantification and optimization of the price-of-anarchy}. In
  \bibinfo{booktitle}{\emph{58th {IEEE} Conference on Decision and Control,
  {CDC} 2019, Nice, France, December 11-13, 2019}}.
  \bibinfo{publisher}{{IEEE}}, \bibinfo{pages}{4041--4046}.
\newblock


\bibitem[\protect\citeauthoryear{Christodoulou and Koutsoupias}{Christodoulou
  and Koutsoupias}{2005}]%
        {christodoulou2005price}
\bibfield{author}{\bibinfo{person}{George Christodoulou} {and}
  \bibinfo{person}{Elias Koutsoupias}.} \bibinfo{year}{2005}\natexlab{}.
\newblock \showarticletitle{The price of anarchy of finite congestion games}.
  In \bibinfo{booktitle}{\emph{Proceedings of the 37th Annual {ACM} Symposium
  on Theory of Computing, Baltimore, MD, USA, May 22-24, 2005}},
  \bibfield{editor}{\bibinfo{person}{Harold~N. Gabow} {and}
  \bibinfo{person}{Ronald Fagin}} (Eds.). \bibinfo{publisher}{{ACM}},
  \bibinfo{pages}{67--73}.
\newblock


\bibitem[\protect\citeauthoryear{Cole, Dodis, and Roughgarden}{Cole
  et~al\mbox{.}}{2006}]%
        {cole2006much}
\bibfield{author}{\bibinfo{person}{Richard Cole}, \bibinfo{person}{Yevgeniy
  Dodis}, {and} \bibinfo{person}{Tim Roughgarden}.}
  \bibinfo{year}{2006}\natexlab{}.
\newblock \showarticletitle{How much can taxes help selfish routing?}
\newblock \bibinfo{journal}{\emph{J. Comput. System Sci.}}
  \bibinfo{volume}{72}, \bibinfo{number}{3} (\bibinfo{year}{2006}),
  \bibinfo{pages}{444--467}.
\newblock


\bibitem[\protect\citeauthoryear{Fotakis and Spirakis}{Fotakis and
  Spirakis}{2008}]%
        {fotakis2008cost}
\bibfield{author}{\bibinfo{person}{Dimitris Fotakis} {and}
  \bibinfo{person}{Paul~G. Spirakis}.} \bibinfo{year}{2008}\natexlab{}.
\newblock \showarticletitle{Cost-balancing tolls for atomic network congestion
  games}.
\newblock \bibinfo{journal}{\emph{Internet Mathematics}} \bibinfo{volume}{5},
  \bibinfo{number}{4} (\bibinfo{year}{2008}), \bibinfo{pages}{343--363}.
\newblock


\bibitem[\protect\citeauthoryear{Harks}{Harks}{2019}]%
        {harks2019pricing}
\bibfield{author}{\bibinfo{person}{Tobias Harks}.}
  \bibinfo{year}{2019}\natexlab{}.
\newblock \showarticletitle{Pricing in resource allocation games based on
  duality gaps}.
\newblock \bibinfo{journal}{\emph{arXiv preprint arXiv:1907.01976}}
  (\bibinfo{year}{2019}).
\newblock


\bibitem[\protect\citeauthoryear{Harks, Kleinert, Klimm, and
  M{\"{o}}hring}{Harks et~al\mbox{.}}{2015}]%
        {HarksKKM15}
\bibfield{author}{\bibinfo{person}{Tobias Harks}, \bibinfo{person}{Ingo
  Kleinert}, \bibinfo{person}{Max Klimm}, {and} \bibinfo{person}{Rolf~H.
  M{\"{o}}hring}.} \bibinfo{year}{2015}\natexlab{}.
\newblock \showarticletitle{Computing network tolls with support constraints}.
\newblock \bibinfo{journal}{\emph{Networks}} \bibinfo{volume}{65},
  \bibinfo{number}{3} (\bibinfo{year}{2015}), \bibinfo{pages}{262--285}.
\newblock


\bibitem[\protect\citeauthoryear{Hearn and Ramana}{Hearn and Ramana}{1998}]%
        {hearn1998solving}
\bibfield{author}{\bibinfo{person}{Donald~W. Hearn} {and}
  \bibinfo{person}{Motakuri~V. Ramana}.} \bibinfo{year}{1998}\natexlab{}.
\newblock \bibinfo{booktitle}{\emph{Solving Congestion Toll Pricing Models}}.
\newblock \bibinfo{publisher}{Springer US}, \bibinfo{address}{Boston, MA},
  \bibinfo{pages}{109--124}.
\newblock


\bibitem[\protect\citeauthoryear{Hoefer, Olbrich, and Skopalik}{Hoefer
  et~al\mbox{.}}{2008}]%
        {HoeferOS08}
\bibfield{author}{\bibinfo{person}{Martin Hoefer}, \bibinfo{person}{Lars
  Olbrich}, {and} \bibinfo{person}{Alexander Skopalik}.}
  \bibinfo{year}{2008}\natexlab{}.
\newblock \showarticletitle{Taxing subnetworks}. In
  \bibinfo{booktitle}{\emph{Internet and Network Economics, 4th International
  Workshop, {WINE} 2008, Shanghai, China, December 17-20, 2008. Proceedings}}
  \emph{(\bibinfo{series}{Lecture Notes in Computer Science})},
  \bibfield{editor}{\bibinfo{person}{Christos~H. Papadimitriou} {and}
  \bibinfo{person}{Shuzhong Zhang}} (Eds.), Vol.~\bibinfo{volume}{5385}.
  \bibinfo{publisher}{Springer}, \bibinfo{pages}{286--294}.
\newblock


\bibitem[\protect\citeauthoryear{Jelinek, Klaas, and Sch{\"{a}}fer}{Jelinek
  et~al\mbox{.}}{2014}]%
        {JelinekKS14}
\bibfield{author}{\bibinfo{person}{Tomas Jelinek}, \bibinfo{person}{Marcus
  Klaas}, {and} \bibinfo{person}{Guido Sch{\"{a}}fer}.}
  \bibinfo{year}{2014}\natexlab{}.
\newblock \showarticletitle{Computing optimal tolls with arc restrictions and
  heterogeneous players}. In \bibinfo{booktitle}{\emph{31st International
  Symposium on Theoretical Aspects of Computer Science {(STACS} 2014), {STACS}
  2014, March 5-8, 2014, Lyon, France}} \emph{(\bibinfo{series}{LIPIcs})},
  \bibfield{editor}{\bibinfo{person}{Ernst~W. Mayr} {and}
  \bibinfo{person}{Natacha Portier}} (Eds.), Vol.~\bibinfo{volume}{25}.
  \bibinfo{publisher}{Schloss Dagstuhl - Leibniz-Zentrum f{\"{u}}r Informatik},
  \bibinfo{pages}{433--444}.
\newblock


\bibitem[\protect\citeauthoryear{Koutsoupias and Papadimitriou}{Koutsoupias and
  Papadimitriou}{1999}]%
        {koutsoupias1999worst}
\bibfield{author}{\bibinfo{person}{Elias Koutsoupias} {and}
  \bibinfo{person}{Christos~H. Papadimitriou}.}
  \bibinfo{year}{1999}\natexlab{}.
\newblock \showarticletitle{Worst-case equilibria}. In
  \bibinfo{booktitle}{\emph{{STACS} 99, 16th Annual Symposium on Theoretical
  Aspects of Computer Science, Trier, Germany, March 4-6, 1999, Proceedings}}
  \emph{(\bibinfo{series}{Lecture Notes in Computer Science})},
  \bibfield{editor}{\bibinfo{person}{Christoph Meinel} {and}
  \bibinfo{person}{Sophie Tison}} (Eds.), Vol.~\bibinfo{volume}{1563}.
  \bibinfo{publisher}{Springer}, \bibinfo{pages}{404--413}.
\newblock


\bibitem[\protect\citeauthoryear{Laffont and Martimort}{Laffont and
  Martimort}{2002}]%
        {laffont2009theory}
\bibfield{author}{\bibinfo{person}{Jean-Jacques Laffont} {and}
  \bibinfo{person}{David Martimort}.} \bibinfo{year}{2002}\natexlab{}.
\newblock \bibinfo{booktitle}{\emph{The Theory of Incentives: The
  Principal-Agent Model}}.
\newblock \bibinfo{publisher}{Princeton University Press}.
\newblock
\showISBNx{9780691091846}


\bibitem[\protect\citeauthoryear{Marden and Wierman}{Marden and
  Wierman}{2013}]%
        {marden2013distributed}
\bibfield{author}{\bibinfo{person}{Jason~R. Marden} {and} \bibinfo{person}{Adam
  Wierman}.} \bibinfo{year}{2013}\natexlab{}.
\newblock \showarticletitle{Distributed welfare games}.
\newblock \bibinfo{journal}{\emph{Operations Research}} \bibinfo{volume}{61},
  \bibinfo{number}{1} (\bibinfo{year}{2013}), \bibinfo{pages}{155--168}.
\newblock


\bibitem[\protect\citeauthoryear{Meir and Parkes}{Meir and Parkes}{2016}]%
        {meir2016marginal}
\bibfield{author}{\bibinfo{person}{Reshef Meir} {and} \bibinfo{person}{David~C.
  Parkes}.} \bibinfo{year}{2016}\natexlab{}.
\newblock \showarticletitle{When are marginal congestion tolls optimal?}. In
  \bibinfo{booktitle}{\emph{Proceedings of the Ninth International Workshop on
  Agents in Traffic and Transportation {(ATT} 2016) co-located with the 25th
  International Joint Conference On Artificial Intelligence {(IJCAI} 2016), New
  York, USA, July 10, 2016}} \emph{(\bibinfo{series}{{CEUR} Workshop
  Proceedings})}, \bibfield{editor}{\bibinfo{person}{Ana L{\'{u}}cia~C.
  Bazzan}, \bibinfo{person}{Franziska Kl{\"{u}}gl}, \bibinfo{person}{Sascha
  Ossowski}, {and} \bibinfo{person}{Giuseppe Vizzari}} (Eds.),
  Vol.~\bibinfo{volume}{1678}. \bibinfo{publisher}{CEUR-WS.org}.
\newblock


\bibitem[\protect\citeauthoryear{Milchtaich}{Milchtaich}{2020}]%
        {milchtaich2020internalization}
\bibfield{author}{\bibinfo{person}{Igal Milchtaich}.}
  \bibinfo{year}{2020}\natexlab{}.
\newblock \showarticletitle{Internalization of social cost in congestion
  games}.
\newblock \bibinfo{journal}{\emph{Economic Theory}} (\bibinfo{year}{2020}),
  \bibinfo{pages}{1--44}.
\newblock


\bibitem[\protect\citeauthoryear{Morrison}{Morrison}{1986}]%
        {morrison1986survey}
\bibfield{author}{\bibinfo{person}{Steven~A Morrison}.}
  \bibinfo{year}{1986}\natexlab{}.
\newblock \showarticletitle{A survey of road pricing}.
\newblock \bibinfo{journal}{\emph{Transportation Research Part A: General}}
  \bibinfo{volume}{20}, \bibinfo{number}{2} (\bibinfo{year}{1986}),
  \bibinfo{pages}{87--97}.
\newblock


\bibitem[\protect\citeauthoryear{Paccagnan, Chandan, Ferguson, and
  Marden}{Paccagnan et~al\mbox{.}}{2019}]%
        {paccagnan2019incentivizing}
\bibfield{author}{\bibinfo{person}{Dario Paccagnan}, \bibinfo{person}{Rahul
  Chandan}, \bibinfo{person}{Bryce~L Ferguson}, {and} \bibinfo{person}{Jason~R
  Marden}.} \bibinfo{year}{2019}\natexlab{}.
\newblock \showarticletitle{Incentivizing efficient use of shared
  infrastructure: Optimal tolls in congestion games}.
\newblock \bibinfo{journal}{\emph{arXiv preprint arXiv:1911.09806}}
  (\bibinfo{year}{2019}).
\newblock


\bibitem[\protect\citeauthoryear{Paccagnan, Chandan, and Marden}{Paccagnan
  et~al\mbox{.}}{2020}]%
        {paccagnan2018distributed}
\bibfield{author}{\bibinfo{person}{Dario Paccagnan}, \bibinfo{person}{Rahul
  Chandan}, {and} \bibinfo{person}{Jason~R. Marden}.}
  \bibinfo{year}{2020}\natexlab{}.
\newblock \showarticletitle{Utility design for distributed resource
  allocation—part {I:} Characterizing and optimizing the exact price of
  anarchy}.
\newblock \bibinfo{journal}{\emph{IEEE Trans. Automat. Control}}
  \bibinfo{volume}{65}, \bibinfo{number}{11} (\bibinfo{year}{2020}),
  \bibinfo{pages}{4616--4631}.
\newblock


\bibitem[\protect\citeauthoryear{Paccagnan, Gentile, Parise, Kamgarpour, and
  Lygeros}{Paccagnan et~al\mbox{.}}{2018}]%
        {paccagnan2018nash}
\bibfield{author}{\bibinfo{person}{Dario Paccagnan}, \bibinfo{person}{Basilio
  Gentile}, \bibinfo{person}{Francesca Parise}, \bibinfo{person}{Maryam
  Kamgarpour}, {and} \bibinfo{person}{John Lygeros}.}
  \bibinfo{year}{2018}\natexlab{}.
\newblock \showarticletitle{{N}ash and {W}ardrop equilibria in aggregative
  games with coupling constraints}.
\newblock \bibinfo{journal}{\emph{IEEE Trans. Automat. Control}}
  \bibinfo{volume}{64}, \bibinfo{number}{4} (\bibinfo{year}{2018}),
  \bibinfo{pages}{1373--1388}.
\newblock


\bibitem[\protect\citeauthoryear{Pigou}{Pigou}{1920}]%
        {pigou}
\bibfield{author}{\bibinfo{person}{Arthur~C Pigou}.}
  \bibinfo{year}{1920}\natexlab{}.
\newblock \bibinfo{booktitle}{\emph{The Economics of Welfare}}.
\newblock \bibinfo{publisher}{Macmillan}.
\newblock


\bibitem[\protect\citeauthoryear{Rosenthal}{Rosenthal}{1973}]%
        {rosenthal1973class}
\bibfield{author}{\bibinfo{person}{Robert~W Rosenthal}.}
  \bibinfo{year}{1973}\natexlab{}.
\newblock \showarticletitle{A class of games possessing pure-strategy {N}ash
  equilibria}.
\newblock \bibinfo{journal}{\emph{International Journal of Game Theory}}
  \bibinfo{volume}{2}, \bibinfo{number}{1} (\bibinfo{year}{1973}),
  \bibinfo{pages}{65--67}.
\newblock


\bibitem[\protect\citeauthoryear{Roughgarden}{Roughgarden}{2015}]%
        {roughgarden2015intrinsic}
\bibfield{author}{\bibinfo{person}{Tim Roughgarden}.}
  \bibinfo{year}{2015}\natexlab{}.
\newblock \showarticletitle{Intrinsic robustness of the price of anarchy}.
\newblock \bibinfo{journal}{\emph{Journal of the ACM (JACM)}}
  \bibinfo{volume}{62}, \bibinfo{number}{5} (\bibinfo{year}{2015}),
  \bibinfo{pages}{32:1--32:42}.
\newblock


\bibitem[\protect\citeauthoryear{Sandholm}{Sandholm}{2002}]%
        {sandholm2002evolutionary}
\bibfield{author}{\bibinfo{person}{William~H Sandholm}.}
  \bibinfo{year}{2002}\natexlab{}.
\newblock \showarticletitle{Evolutionary implementation and congestion
  pricing}.
\newblock \bibinfo{journal}{\emph{The Review of Economic Studies}}
  \bibinfo{volume}{69}, \bibinfo{number}{3} (\bibinfo{year}{2002}),
  \bibinfo{pages}{667--689}.
\newblock


\bibitem[\protect\citeauthoryear{Skopalik and Vijayalakshmi}{Skopalik and
  Vijayalakshmi}{2020}]%
        {skopalik2020improving}
\bibfield{author}{\bibinfo{person}{Alexander Skopalik} {and}
  \bibinfo{person}{Vipin~Ravindran Vijayalakshmi}.}
  \bibinfo{year}{2020}\natexlab{}.
\newblock \showarticletitle{Improving approximate pure {N}ash equilibria in
  congestion games}.
\newblock \bibinfo{journal}{\emph{arXiv preprint arXiv:2007.15520}}
  (\bibinfo{year}{2020}).
\newblock


\bibitem[\protect\citeauthoryear{Suri, T{\'{o}}th, and Zhou}{Suri
  et~al\mbox{.}}{2007}]%
        {suri2007selfish}
\bibfield{author}{\bibinfo{person}{Subhash Suri}, \bibinfo{person}{Csaba~D.
  T{\'{o}}th}, {and} \bibinfo{person}{Yunhong Zhou}.}
  \bibinfo{year}{2007}\natexlab{}.
\newblock \showarticletitle{Selfish load balancing and atomic congestion
  games}.
\newblock \bibinfo{journal}{\emph{Algorithmica}} \bibinfo{volume}{47},
  \bibinfo{number}{1} (\bibinfo{year}{2007}), \bibinfo{pages}{79--96}.
\newblock


\bibitem[\protect\citeauthoryear{Tekin, Liu, Southwell, Huang, and Ahmad}{Tekin
  et~al\mbox{.}}{2012}]%
        {tekin2012atomic}
\bibfield{author}{\bibinfo{person}{Cem Tekin}, \bibinfo{person}{Mingyan Liu},
  \bibinfo{person}{Richard Southwell}, \bibinfo{person}{Jianwei Huang}, {and}
  \bibinfo{person}{Sahand Haji~Ali Ahmad}.} \bibinfo{year}{2012}\natexlab{}.
\newblock \showarticletitle{Atomic congestion games on graphs and their
  applications in networking}.
\newblock \bibinfo{journal}{\emph{IEEE/ACM Transactions on Networking}}
  \bibinfo{volume}{20}, \bibinfo{number}{5} (\bibinfo{year}{2012}),
  \bibinfo{pages}{1541--1552}.
\newblock


\bibitem[\protect\citeauthoryear{Wardrop}{Wardrop}{1952}]%
        {wardrop1952road}
\bibfield{author}{\bibinfo{person}{John~Glen Wardrop}.}
  \bibinfo{year}{1952}\natexlab{}.
\newblock \showarticletitle{Some theoretical aspects of road traffic research}.
\newblock \bibinfo{journal}{\emph{Proceedings of the institution of civil
  engineers}} \bibinfo{volume}{1}, \bibinfo{number}{3} (\bibinfo{year}{1952}),
  \bibinfo{pages}{325--362}.
\newblock


\end{thebibliography}
